\documentclass[10pt,a4paper,reqno]{amsart}
\usepackage{amsthm,amsfonts,amssymb,amsmath,euscript,comment, enumerate}
\usepackage{graphicx}
\usepackage{color}
\usepackage{bbm}
\usepackage{geometry}
\usepackage{amsfonts}
\usepackage{xcolor}
\usepackage{seqsplit}

\newtheorem{thm}{Theorem}[section]
\newtheorem{prp}[thm]{Proposition}
\newtheorem{cor}[thm]{Corollary}
\newtheorem{lm}[thm]{Lemma}
\newtheorem{df}[thm]{Definition}
\newtheorem{rk}[thm]{Remark}
\newtheorem{proposition}[thm]{Proposition}
\newtheorem{lemma}[thm]{Lemma}

\newcommand{\m}[1]{\mathbb{#1}}

\newcommand{\q}[1]{\mathcal{#1}}

\newcommand{\wht}[1]{\widetilde{#1}}
\newcommand{\gra}[1]{\mathbf{#1}}

\newcommand{\ep}{\varepsilon}

\newcommand{\f}{\frac}
\newcommand{\rd}{\partial}
\newcommand{\nab}{\nabla}
\newcommand{\alp}{\alpha}
\newcommand{\bt}{\beta}
\newcommand{\bA}{{\bf A}}
\newcommand{\bB}{{\bf B}}
\newcommand{\bC}{{\bf C}}
\newcommand{\gi}{(g^{-1})}
\newcommand{\mfg}{\mathfrak g}
\newcommand{\bfG}{{\bf \Gamma}}
\newcommand{\ls}{\lesssim}
\newcommand{\de}{\delta}

\newcommand{\co}[1]{\cos\left(#1\right)}
\newcommand{\si}[1]{\sin\left(#1\right)}
\numberwithin{equation}{section}

\begin{document}

\title
{High-frequency backreaction for the Einstein equations under polarized $\mathbb U(1)$ symmetry}

\begin{abstract}
{K}nown examples in plane symmetry or Gowdy symmetry show that given a $1$-parameter family of solutions to the vacuum Einstein equations, it may have a weak limit which does not satisfy the vacuum equations, but instead has a non-trivial stress-energy-momentum tensor. We consider this phenomenon under polarized $\mathbb U(1)$ symmetry -- a much weaker symmetry than most of the known examples -- such that the stress-energy-momentum tensor can be identified with that of multiple families of null dust propagating in distinct directions. 
We prove that any generic local-in-time small-data polarized-$\mathbb U(1)$-symmetric solution to the Einstein--multiple null dust system can be achieved as a weak limit of vacuum solutions. Our construction allows the number of families to be arbitrarily large, and appears to be the first construction of such examples with more than two families. 
\end{abstract}

\author{C\'ecile Huneau}
\address{Institut Fourier, Universit\'e Grenoble-Alpes, 100 rue des maths, 38610 Gi\`eres, France}
\email{cecile.huneau@univ-grenoble-alpes.fr}
\author{Jonathan Luk}
\address{Department of Mathematics, Stanford University, CA 94304, USA}
\email{jluk@stanford.edu}
	
	\maketitle

\section{Introduction}

There is a long tradition in the physics literature since the work of Isaacson \cite{Isaacson1, Isaacson2} studying high frequency backreaction in general relativity \cite{CBHF, GW1, GW2, MT}. {In particular}, it has been observed that suitably scaled small-amplitude but high-frequency gravitational waves give rise to a non-trivial ``$O(1)$'' contribution to the background metric which mimics an ``effective matter field''. For the purpose of this paper, we consider the following mathematical {formulation} by Burnett \cite{Burnett}: On a fixed manifold, consider a one-parameter family of Lorentzian metrics $g_\lambda$ for $\lambda \in (0,\lambda_0]$, $\lambda_0\in \mathbb R$ which satisfies the Einstein vacuum equations\footnote{We use the notation that for a given metric $g$, $R_{\mu\nu}$ denotes its Ricci tensor while $R(g)$ denotes its scalar curvature.}
\begin{equation}\label{Einstein.vac}
R_{\mu\nu}(g_\lambda)=0.
\end{equation}
Assume moreover that there exists a metric $g_0$ such that as $\lambda\to 0$,
\begin{itemize}
\item $g_{\lambda}$ converges to $g_0$ {uniformly on compact sets}; and 
\item the derivatives of $g_{\lambda}$ converges weakly in {$L^2$} to the derivatives of $g_0$.
\end{itemize}
Due to the nonlinearity of the equations, it can happen that the limiting metric $g_0$ is a solution to the Einstein equations which is \underline{not} vacuum, but has a \underline{non-trivial} stress-energy-momentum tensor, i.e., 
\begin{equation}\label{Einstein.w.matter}
R_{\mu \nu}(g_0)-\frac{1}{2}(g_0)_{\mu \nu} R(g_0)= T_{\mu \nu}
\end{equation}
for some $T_{\mu\nu}\neq 0$. We will call the tensor $T_{\mu\nu}$ for the limiting metric the \emph{effective stress-energy-momentum tensor}.

Two questions arise in this context: First, what are the possible matter models\footnote{{There are important caveats for this question though: $T_{\mu\nu}$ that arises may not always be identifiable with some well-known matter models. Also, in general for a given $T_{\mu\nu}$, there is no reason to expect that it corresponds uniquely to some matter models.}} that can arise as limiting effective stress-energy-momentum tensor? Second, if a certain matter model arise as such limits, is it true that \underline{all} solutions (at least in a certain solution regime) to the Einstein equations coupled with that matter model can be achieved as a limit of vacuum solutions in the sense described above?

In connection to the questions above, Burnett \cite{Burnett} conjectured that all such limits can be identified as solutions to the Einstein--massless Vlasov system. Moreover, in the same work, he also asked the question whether all solutions to the Einstein--massless Vlasov system can be achieved as a limit (in the sense described above) of solutions to the Einstein vacuum equations. Both of these questions remain open. In the direction of constraining possible effective stress-energy-momentum tensors, it is known by the work of Green--Wald \cite{GW1} that the effective stress-energy-momentum tensor $T_{\mu\nu}$ verifies non-trivial constraints: in fact, it must be traceless and satisfy the weak energy condition. These conditions are consistent with - but clearly much weaker than - requiring $T_{\mu\nu}$ be identifiable as that of massless Vlasov matter. In the direction of constructing examples of effective stress-energy-momentum tensors, in the only known examples in the literature, $T_{\mu\nu}$ of the limiting spacetime can be identified with that of one or two families of null dust. This can indeed be viewed at least formally as a special (singular) case of massless Vlasov matter.

The aim of this paper is to construct {one-parameter families of metrics} such that the effective stress-energy-momentum tensor represents that of a sum of an \underline{arbitrary} finite number of families of null dust propagating in different directions. In particular, our work can be viewed as a first step towards approaching the questions of Burnett: one may hope to eventually construct more general solutions to the Einstein--massless Vlasov system by taking the number of families of null dust to infinity. Our construction assumes that the one-parameter families of metrics and the limiting metric are polarized $\mathbb U(1)$ symmetric. {In} this setting, not only do we construct particular examples, we show that in fact any generic, polarized $\mathbb U(1)$-symmetric, small data solution to the Einstein--null dust system (with an arbitrary finite number of families of null dust) arise as a weak limit of solutions to the Einstein vacuum equations.

The mathematical challenge in studying this problem is that the metrics one construct necessarily oscillate with high frequency. In the process, one needs to handle solutions with very weak uniform estimates and deal with issues familiar in the setting of low-regularity problems. {In fact, in order to obtain a non-trivial stress-energy-momentum tensor in the limit, the one parameter family of metrics can at best obey uniform $W^{1,\infty}$ estimates, but not any uniform estimates for the higher order derivatives. On the other hand, even to ensure uniform time of existence of solutions, one in general needs the metrics to be in $W^{2,2}$ uniformly, cf.~\cite{KRS}.} Most previous examples \cite{Burnett, GW2} rely on the fact that under plane or polarized Gowdy symmetry, the solutions to the Einstein vacuum equations, even when the regularity is low, are explicit. On the other hand, these symmetries limits the effective stress-energy-momentum tensor to represent at most two families of null dusts.

{One way to obtain a low regularity local existence result for the Einstein equations which is consistent with the high frequency oscillations described above is to assume that the initial data are more regular in \emph{some} directions.} {This has been achieved in the recent work of Luk--Rodnianski \cite{LR1, LR2}, which} constructed\footnote{The original motivation for such a low-regularity result in \cite{LR1, LR2} is to study the {propagation and} interaction of impulsive gravitational waves (solutions such that the Riemann curvature tensor has a delta singularity on an embedded null hypersurface).} via solving a characteristic initial problem a class of low-regularity solutions to the Einstein vacuum equation on $S\times \mathbb R^2$ (where $S$ is a $2$-dimensional surface) with \underline{no} exact symmetries such that all tangential derivatives of the metric along $S$ are regular, while general derivatives of the metric are only in $L^2$. With such an existence result, one can construct a sequence of metrics which in the limit give rise to an effective stress-energy-momentum tensor of \underline{two} families of null dust, with \underline{no} symmetry assumptions; see {the forthcoming} \cite{LRHF}. Nevertheless, even in this construction, since the metric is regular along $S$, the oscillations in the metrics are limited to $2$ dimensions and therefore the effective stress-energy-momentum tensor can represent at most two families of null dust. 

{In contrast, our present work gives the first result such that the metrics $g_\lambda$ are allowed to oscillate in more than $2$ dimensions. As a consequence, we obtain stress-energy-momentum tensors representing an arbitrary number of families of null dust. To obtain our result, we do not attempt to prove general existence results which are consistent with the metric oscillating in more than $2$ dimensions. Instead, our analysis is specific to the problem at hand, and we construct a parametrix to capture the high frequency oscillations of the metric. Moreover, we rely heavily on the structure of the Einstein equations under polarized $\mathbb U(1)$ symmetry. We refer the reader Section~\ref{ideas} for further discussions on the ideas of the proof, after we give an informal discussion of the main result of the paper in the next subsection.}

\subsection{Main result}\label{sec:main.results}

To construct our examples, we will work under polarized $\m U(1)$ symmetry\footnote{Notice that in contrast to {plane symmetry or Gowdy symmetry which we discussed above}, the symmetry group $\m U(1)$ is only $1$-dimensional.}, which gives certain important technical simplifications that we will explain later. In addition to merely constructing examples, we will in fact show that under polarized $\m U(1)$ symmetry, up to some technical conditions, \underline{most} sufficient small and sufficient regular local solutions $g_0$ to \eqref{Einstein.w.matter} with $T_{\mu\nu}$ given by a sum of a finite number of families of null dust can be weakly approximated by vacuum solutions $g_\lambda$ to \eqref{Einstein.vac} in the sense that we described above. 

We now further describe our setting: We study solution{s} of the vacuum Einstein equations of the form $(I\times \m R^{3}, ^{(4)}g)$, where $I\subset \mathbb R$ is an interval, with
$$^{(4)}g= e^{-2\phi}g+ e^{2\phi}(dx^3)^2,$$
where\footnote{{Abusing notation slightly, we will also view $\phi$ as a function $\phi:I\times \m R^{3}\to \m R$ and $g$ as a metric on $I\times \m R^{3}$, which \emph{do not depend on the variable $x^3$}.}} $\phi:I\times \m R^{2}\to \m R$ is a scalar function and $g$ is a Lorentzian metric on $I\times \m R^{2}$. The vector field $\partial_{x_3}$ is Killing and hypersurface orthogonal.
Then Einstein vacuum equations $R(^{(4)}g)_{\mu\nu}=0$ are equivalent to the following system for $(g,\phi)$:
\begin{equation}
\label{sys}\left\{
\begin{array}{l}
\Box_g \phi = 0,\\
R_{\mu \nu}(g)= 2\partial_\mu \phi \partial_\nu \phi.
\end{array}
\right.
\end{equation}
As mentioned above, we want to show that the weak limits of vacuum solutions (i.e., solutions to \eqref{sys}) give rise to solutions with a stress-energy-momentum tensor which represents a sum of $N$ families of null dust, traveling in arbitrary directions (as long as the directions associated to each pair of families are ``angularly separated'', see Section~\ref{sec.main.2nd}). More precisely, we consider a quadruple $(g,\phi, F_{\bA}, u_{\bA})$, with $\bA\in \mathcal A$ for some finite set $\mathcal A$ with $|\mathcal A|=N$, where $g$ is a Lorentzian metric on $I\times \m R^2$, $\phi:I\times \m R^2\to \m R$ is a scalar function, $F_{\bA}:I\times \m R^2\to \m R_{\geq 0}$ is the \emph{density} of the null dust for each $\bA$ and $u_{\bA}:I\times \m R^2\to \m R$ is an eikonal function such that $(du_{\bA})^\sharp$ is the \emph{direction} of propagation of the null dust for each $\bA$, which is a solution to the system
\begin{equation}\label{back}
\left\{\begin{array}{l}
R_{\mu \nu}(g)= 2\partial_\mu \phi \partial_\nu \phi + \sum_{{\bA}} (F_{{\bA}})^2\partial_\mu u_{{\bA}} \partial_\nu u_{{\bA}},\\
\Box_{g}\phi= 0,\\
2(g^{-1})^{\alpha \beta}\partial_{\alpha} u_{{\bA}} \partial_{\beta} F_{{\bA}} + (\Box_{{g}} u_{{\bA}}) F_{{\bA}} = 0,\\
(g^{-1})^{\alpha \beta}\partial_\alpha u_{{\bA}} \partial_\beta u_{{\bA}}=0.
\end{array}
\right.
\end{equation}
Notice that the system \eqref{back} corresponds to the following system\footnote{To see this, simply note that (1) with the ansatz of the metric, and $\Box_g\phi=0$, we have for $\alp,\bt=0,1,2,$ $R_{\alp\bt}(^{(4)}g)=R_{\alp\bt}(g)-2\rd_\alp \phi \rd_\bt \phi$, $R_{3\alp}=R_{33}=0$; (2) $\Box_g= \Box_{^{(4)}g}$ and (3) $u_{\bA}$, $F_{\bA}$ are independent of $x^3$. For these formulae, the reader can consult \cite[Appendix~VII]{livrecb}.} in the original $(3+1)$-dimensional spacetime $I\times \m R^3$:
\begin{equation*}
\left\{\begin{array}{l}
R_{\mu \nu}(^{(4)} g)= \sum_{{\bA}} (F_{{\bA}})^2\partial_\mu u_{{\bA}} \partial_\nu u_{{\bA}},\\
2(^{(4)}g^{-1})^{\alpha \beta}\partial_{\alpha} u_{{\bA}} \partial_{\beta} F_{{\bA}} + (\Box_{^{(4)}g} u_{{\bA}}) F_{{\bA}} = 0,\\
(^{(4)}g^{-1})^{\alpha \beta}\partial_\alpha u_{{\bA}} \partial_\beta u_{{\bA}}=0.
\end{array}
\right.
\end{equation*}
In other words, \eqref{back} indeed corresponds to the Einstein--null dust system with $N$ families of null dust.

Our main theorem can be stated informally as follows (and we refer the readers to {Theorem~\ref{main.thm.2}} for a precise statement):
\begin{thm}\label{main.intro}
Let $(g_0,\phi_0, F_{\bA}, u_{\bA})$ be a sufficiently small and sufficiently regular local-in-time asymptotically conic solution to \eqref{back} such that
\begin{itemize}
\item The initial hypersurface is maximal;
\item The $u_{\bA}$'s are angularly separated;
\item A genericity condition holds.
\end{itemize}
Then $(g_0,\phi_0)$ can be weakly approximated by a $1$-parameter family of solutions $(g_{\lambda},\phi_{\lambda})$ for $\lambda \in (0,\lambda_0)$, $\lambda_0\in\mathbb R$ to \eqref{sys}, i.e., in a suitable coordinate system, as $\lambda\to 0$, $(g_{\lambda},\phi_{\lambda})\to (g_{0},\phi_{0})$ {uniformly on compact sets} and the derivatives $(\rd g_{\lambda},\rd \phi_{\lambda})\rightharpoonup (\rd g_{0},\rd \phi_{0})$ weakly in {$L^2$} (for each component).
\end{thm}

It is not unlikely that the methods introduced in this paper can also be applied to the situation where the spacetime is only assumed to be $\m U(1)$ symmetric without polarization. In that case, the Einstein vacuum equations no longer reduce to a system of Einstein--scalar wave equations, but instead reduce to a system of Einstein--wave map equations. The wave map equations are more nonlinear, but the semilinear terms satisfy the \emph{classical null condition}, which therefore give hope to extending the results in the present paper. However, {the} full $(3+1)$-dimensional situation, even when restricted to ``almost-$\m U(1)$-symmetry'', seem{s} to require ideas beyond those introduced here.

\subsection{Ideas of the proof}\label{ideas}

The main challenge in constructing our examples is that in order for backreaction to occur, $\phi_\lambda$ must oscillate in a manner that only their first derivatives can be assumed to be uniformly bounded and that their second derivatives (or in fact any ``$1+\ep$ derivative'') must blow up as $\lambda\to 0$. This is much rougher than the required regularity under which the Einstein equations are known to be locally well-posed (see \cite{KRS}). In \cite{LRHF}, this issue was dealt with exactly with an improved local well-posedness result \cite{LR2}, devised first to understand impulsive gravitational waves, in the general setting where $\mathcal M=S\times \mathbb R^2$ and that the $S$-tangent derivatives of the metric are more regular.

In this paper, however, we take another route. Instead of proving local well-posedness result in general function spaces which $(g_\lambda,\phi_\lambda)$ belong, we exploit the fact that the $(g_\lambda,\phi_\lambda)$ are constructed to (weakly) approximate $(g_0,\phi_0)$ and rely more on the regularity of $(g_0,\phi_0)$. In the proof, we construct a parametrix of the high frequency part of $(g_\lambda,\phi_\lambda)$ and prove that it is indeed a good approximation of the actual solution. In particular, we show that the ``high frequency waves'' in $\phi_\lambda$ approximately travel along characteristic hypersurfaces of the \underline{background} metric $g_0$. As a consequence of our approach, we also obtain a precise description of $(g_\lambda,\phi_\lambda)$ as $\lambda\to 0$. {Before we enter a more detailed discussion of the proof,} let us note that in order to carry out this analysis, we rely on the following three facts specific to $\m U(1)$ symmetry:
\begin{enumerate}
\item The Einstein vacuum equations reduce under the symmetry assumptions to the Einstein--scalar field system. Under this reduction, it is only the scalar field that contributes to the backreaction. In particular, the ``high-frequency wave part'' and the ``geometry part'' of the solution ``separate''.
\item By introducing an elliptic gauge, the metric of the reduced system can be completely recovered {by} \underline{elliptic} equations (i.e., the ``geometry part'' of the solution is more regular).
\item Since the reduced equations are in $(2+1)$-dimensions, the elliptic gauge can be introduced such that the equations for the metric of the reduced system are \underline{semilinear}\footnote{This can be viewed as the consequence of the uniformization theorem, together with the conformal invariance of the $2$-dimensional Laplacian.}.
\end{enumerate}

\subsubsection{Model problem and first construction of parametrix}
{Working} in an elliptic gauge, and thinking of $\mfg$ as a metric component, the equations can roughly be modeled by the following model problem:
\begin{equation}\label{model.prob}
\left\{
\begin{array}{l}
\Box_{g(\mfg)} \phi= 0,\\
\Delta \mfg = (\partial \phi)^2,
\end{array}
\right.
\end{equation}
where $g$ is a Lorentzian metric and $\Delta$ is the Laplacian for flat $\mathbb R^2$. This formulation captures the advantages of $\mathbb U(1)$ symmetry as described earlier, in that the ``dynamical'' part and the ``geometric'' part ``sep{a}rate'', and that the geometric part satisfies a semilinear elliptic system.

We will construct our desired one-parameter family of solutions to \eqref{sys} in the form
\begin{equation}
\label{ansatz}\phi_{\lambda}= \phi_0+\sum_\bA \lambda F_\bA \cos\left(\frac{u_\bA}{\lambda}\right)+\wht \phi_{\lambda},\quad \mfg_{\lambda}= \mfg_0 + \wht \mfg_{\lambda},
\end{equation}
such that $(\partial \wht \phi_\lambda, \partial \wht \mfg_{\lambda}) \rightarrow (0,0)$ {uniformly on compact sets} as $\lambda \to 0$.
First, by virtue of the fact that $u_{\bA}$ satisfies the eikonal equation for the background metric, and that $F_{\bA}$ satisfies the transport equation in \eqref{back}, $\sum_\bA \lambda F_\bA \cos\left(\frac{u_\bA}{\lambda}\right)$ can be viewed as an approximate solution to the wave equation, since $\Box_g\left(\sum_\bA \lambda F_\bA \cos\left(\frac{u_\bA}{\lambda}\right)\right)$ consists of terms either of size $O(\lambda)$ or bounded in terms of $\wht \mfg_{\lambda}$. Moreover,
$$2\partial_\mu \phi_{\lambda}\partial_\nu \phi_{\lambda} \rightharpoonup 2\partial_\mu \phi_0 \partial_\nu \phi_0 + \sum_{{\bA}} (F_{{\bA}})^2\partial_\mu u_{{\bA}} \partial_\nu u_{{\bA}},$$
weakly in $L^2$ (with uniform $L^\infty$ bounds) which is exactly the form of the stress-energy-momentum tensor in \eqref{back}.

However, in order for this parametrix to be useful, we at the very least need to say that $\wht \phi_{\lambda}$ is better than the main term $\sum_\bA \lambda F_\bA \cos\left(\frac{u_\bA}{\lambda}\right)$. Suppose, say, by an expansion in $\lambda$, we hope to obtain that $\wht \phi_{\lambda}=O(\lambda^2)$, with a loss of $\lambda^{-1}$ for every derivative. Plugging this into \eqref{model.prob}, we have on the RHS of the elliptic equation a term
\begin{equation}\label{exemple} \partial \wht \phi_{\lambda} \partial  \left(\lambda F_\bA \cos\left(\frac{u_\bA}{\lambda}\right)\right) =O(\lambda),
\end{equation}
so that standard elliptic estimates give
\begin{equation}\label{intro.ex.g}
\| \wht \mfg_{\lambda} \|_{W^{1,\infty}} \leq \lambda.
\end{equation}
Now we plug this back into the wave equation $\Box_g \wht \phi_{\lambda}$ (for which $\wht g_{\lambda}$ enters through $\Box_g$). When estimating the $H^1$ norm of $\wht \phi_{\lambda}$ using energy estimates, we need to control (in $L^2$) a term of the form
$$\frac{1}{\lambda}\left(g_{\lambda}^{\alpha \beta}-g_0^{\alpha \beta}\right)\partial_\alpha u_{\bA}\partial_\beta u_{\bA} \co{\frac{u_\bA}{\lambda}}.$$
However, with \eqref{intro.ex.g}, this term is only $O(1)$ (instead of $O(\lambda)$) and the estimates cannot be closed. In order to deal with this, we need a more precise parametrix, which in particular captures the fact the $O(\lambda)$ term in $\wht \phi_{\lambda}$ are of high frequency, and moreover, that one can gain a smallness parameter from these high frequency terms for appropriate inversions of $\Delta$ and $\Box_{g}$.

\subsubsection{Well-preparedness of the eikonal functions and a more precise construction of the parametrix}

The more precise parametrix for $\phi_\lambda$ takes the form (for details see Section~\ref{sec.parametrix})
\begin{align*}
\phi_\lambda= &\phi_0+\sum_\bA \lambda F_\bA \cos\left(\frac{u_\bA}{\lambda}\right) +\sum_\bA \lambda^2 \wht F_\bA \si{\frac{u_\bA}{\lambda}}\\
&+\sum_\bA \lambda^2 \wht F_\bA^{(2)}\co{\frac{2u_\bA}{\lambda}}+\sum_\bA  \lambda^2\wht F_\bA^{(3)} \si{\frac{3u_\bA}{\lambda }} + \mathcal E_\lambda,
\end{align*}
where one should think of $\wht F_\bA$, $\wht F_\bA^{(2)}$ and $\wht F_\bA^{(3)}$ as bounded and $\mathcal E_\lambda$ as a smaller remainder, bounded in $H^1$ by $\lambda^2$. We also decompose the metric components accordingly:
$$\mfg=\mfg_0+\mfg_1+\mfg_2+\mfg_3,$$
so that roughly speaking $\mfg_1$ are $O(\lambda^2)$ (high-frequency) terms and $\mfg_2$ are $O(\lambda^3)$ (high-frequency) terms.

In order to control the parametrix up to this order, notice that certain terms (in particular $\wht F_{\bA}$, $\mathcal E_\lambda$, $\mfg_2$ and $\mfg_3$; see Section~\ref{sec.parametrix}) couple. As a consequence, one needs to carefully handle the regularity of these terms (in addition to estimating their sizes in terms of $\lambda$).

In the parametrix for $\phi_\lambda$, we only explicitly keep track of the terms oscillating in null directions (with phase function proportional to $u_{\bA}$). This construction is based on the facts that
\begin{itemize}
\item Any high frequency term must be highly oscillatory in the spatial direction, and one gains $\lambda^2$ upon inverting $\Delta$;
\item Any high frequency term arising from non-parallel interaction (i.e., from $u_{\bA}$ and $u_{\bB}$ with $\bA\neq \bB$) can be treated ``elliptically'' for the wave part, and inverting $\Box_g$ gives a smallness of $\lambda^2$.
\end{itemize}
To achieve these, we need to use the assumption that $\{u_{\bA}\}_{\bA\in \mathcal A}$ is angularly separated (see Definition~\ref{ang.sep}) and also to exploit the symmetry $(u_{\bA}, F_{\bA})\mapsto (c_{\bA} u_{\bA}, c_{\bA}^{-1} F_{\bA})$ to prepare the eikonal functions well (see Section~\ref{sec.prepare}).

There are also additional useful structures hidden in the expression of the parametrix which we exploit. For instance, when controlling terms as in \eqref{exemple}, we use that the interaction of the $F$ term with any of $(\wht F, \wht F^{(2)}, \wht F^{(3)})$ is necessarily high frequency, and one can indeed gain powers of $\lambda$.

\subsubsection{Handling the $\rd_t$ derivatives of the error of the metric}
One important challenge we face is that the $\rd_t$ derivative of the metric obeys worse estimates (compared to the spatial derivatives). This is because the metric components solve elliptic equations on a spatial slice, and to control their $\rd_t$ derivative requires differentiating the equation. {Because of} similar reasons, and {also} the fact that $\wht F_{\bA}$ couples with $\mfg_3$, the second time derivative $\rd_t^2\wht F_{\bA}$ also obey worse estimates compared to when at least one of the derivatives is a spatial derivatives.

Fortunately, one can still handle the situation with the worse estimates using the structure of the Einstein equations! We highlight a few points below: Note that in particular, except for point (3), we strongly use the fact that we are solving the Einstein equations and the structure that we rely on is \underline{not} captured by the model problem \eqref{model.prob}.

\begin{enumerate}
\item Hidden in the evolution is the propagation of the maximality of the hypersurfaces. This allows us to rewrite the $\rd_t$ derivative of a metric component (more precisely, $\gamma$, to be introduced in \eqref{g.form}) in terms of spatial derivatives of (precisely, $\beta$). Hence, such a term is better than expected.
\item Importantly, using in particular the above observation, one shows that the two uncontrollable $\rd_t$-derivative error terms, one involving $\rd_t \mfg_3$ and one involving $\rd_t^2 \wht F_{\bA}$, cancel; see Proposition~\ref{cancellation}.
\item Another type of error terms is of the form $\rd_t\mfg_3$ multiplied by a low frequency term. Here, one can exploit the low frequency term using an integration by parts argument; see Proposition~\ref{E.rdtg3.est}.
\end{enumerate}

\subsection{Outline of the paper}
The remainder of the paper will be structured as follows: 
\begin{itemize}
	\item In \textbf{Section \ref{sec.notations}}, we introduce the notations and functions spaces that we will use.
	\item In \textbf{Section \ref{sec.gauge}}, we give a precise definition of our gauge condition, and {state} the local {well-posed}ness results for \eqref{sys} and \eqref{back} {from our companion paper \cite{HL}}.
	\item In \textbf{Section \ref{sec.main}}, we give a precise statement of the main theorem (cf.~Theorems~\ref{main.intro} and \ref{main.thm.2}). We also discuss preliminary steps of the proof: rescaling of the eikonal functions, construction of initial data to the one-parameter family of solutions, and construction of the parametrix for $\phi_{\lambda}$ and $g_{\lambda}$. 
	\item In next few sections are devoted to the proof of the main theorem:
	\begin{itemize}
	\item The proof proceeds by a bootstrap argument, and the main setup and bootstrap assumptions are given in \textbf{Section~\ref{sec.bootstrap}}.
	\item In \textbf{Section~\ref{sec.scalar}}, we derive the estimates for the scalar field.
	\item In \textbf{Section~\ref{secelliptic}}, we derive the estimates for the metric components.
	\item Finally, in \textbf{Section~\ref{sec.concl}}, we conclude the proof.
	\end{itemize}
	\item In \textbf{Appendix~\ref{weightedsobolev}}, we collect some results about Sobolev embedding, product estimates and elliptic estimates in weighted Sobolev spaces in $\mathbb R^2$. 	
\end{itemize}

\subsection*{Acknowledgements} J. Luk thanks Igor Rodnianski and Robert Wald for stimulating discussions. Most of this work was carried out when both authors were in Cambridge University. {C. Huneau is supported by the ANR-16-CE40-0012-01.} J. Luk is supported in part by a Terman fellowship. 

\section{Notations and function spaces}\label{sec.notations}

{\bf Ambient space and coordinates}
In this paper, we will be working on the ambient manifold $\q M:=I\times \mathbb R^2$, where $I\subset \mathbb R$ is an interval. The space will be equipped with a system of coordinates $(t,x^1,x^2)$. We will use $x^i$ with the lower case Latin index $i,j=1,2$ and will also sometime denote $x^0=t$.

{\bf Conventions with indices} We will use the following conventions:
\begin{itemize}
\item Lower case Latin indices run through the spatial indices $1,2$, while lower case Greek indices run through all the spacetime indices.
\item Repeat indices are always summed over: where lower case Latin indices sum over the spatial indices $1,2$ and lower case Greek indices sum over all indices $0,1,2$.
\item Unless otherwise stated, lower case Latin indices are always raised and lowered with respect to the standard Euclidean metric $\delta_{ij}$.
\item In contrast, lower case Greek indices are raised and lowered with respect to the spacetime metric $g$. In cases where there are more than one spacetime metric in the immediate context, we will not use this convention but will instead spell out explicitly how indices are raised and lowered.
\end{itemize}

{\bf Differential operators} We will use the following conventions for differential operators:
\begin{itemize}
\item $\rd$ denotes partial derivatives in the coordinate system $(t,x^1,x^2)$. We will frequently write $\rd_i$ for $\rd_{x^i}$. In particular, we denote
$$|\rd \xi|^2=(\rd_t\xi)^2+\sum_{i=1}^2(\rd_{x^i}\xi)^2.$$
\item The above $\rd$ notation also applied to rank-$r$ covariant tensors $\xi_{\mu_1\dots\mu_r}$ tangential to $I\times \mathbb R^2$ to mean
$$|\rd\xi|^2=\sum_{\mu_1,\dots,\mu_r=t,x^1,x^2}|\rd \xi_{\mu_1\dots\mu_r}|^2$$
and to rank-$r$ {covariant} tensors $\xi_{i_1\dotsi_r}$ tangential to $\mathbb R^2$ to mean
$$|\rd\xi|^2=\sum_{i_1{\dots i_r}=x^1,x^2}|\rd \xi_{i_1{\dots i_r}}|^2.$$
\item $\Delta$ and $\nabla$ denotes the spatial Laplacian and the spatial gradient on $\m R^2$ with the standard \underline{Euclidean metric}. In particular, we use the convention
$$|\nabla\xi|^2=\sum_{i=1}^2|\rd_{x^i}\xi|^2.$$
\item $D$ denotes the Levi--Civita connection associated to the \underline{spacetime metric $g$}.
\item $\Box_g$ denotes the Laplace--Beltrami operator on functions, i.e., 
$$\Box_g\xi:=\f{1}{\sqrt{|\det g|}}\rd_\mu(\gi^{\mu\nu}\sqrt{|\det g|}\rd_\nu\xi).$$
\item $\q L$ denotes the Lie derivatives.
\item $e_0$ defines the vector field $e_0=\rd_t-\beta^i\rd_{x^i}$ (where $\beta$ will be introduced in \eqref{g.form}). We will often use the differential operator $\q L_{e_0}$.
\item $L$ denotes the Euclidean conformal Killing operator acting on vectors on $\m R^2$ to give a symmetric traceless (with respect to $\delta$) covariant $2$-tensor, i.e., 
$$(L\xi)_{ij}:=\delta_{j\ell}\rd_i\xi^\ell+\delta_{i\ell}\rd_j\xi^\ell-\delta_{ij}\rd_k\xi^k.$$
\end{itemize}

{\bf Functions spaces} We will work with standard function spaces $L^p$, $H^k$, $C^m$, $C^\infty_c$, etc. and assume the standard definitions. The following conventions will be important:
\begin{itemize}
\item \underline{Unless otherwise stated, all function spaces will be taken on $\m R^2$} and the measures will be taken to be the 2D Lebesgue measure $dx$. 
\item When applied to quantities defined on a spacetime $I \times \m R^2$, the norms $L^p$, $H^k$, $C^m$ denote \underline{fixed-time} norms (unless otherwise stated). In particular, if in an estimate the time $t\in I$ in question is not explicitly stated, then it means that the estimate holds for \underline{all} $t\in I$ for the time interval $I$ that is appropriate for the context.
\end{itemize}
We will also work in weighted Sobolev spaces, which are well-suited to elliptic equations. We recall here the definition, together with the definition of weighted H\"older space. The properties of these spaces that we need are listed in Appendix \ref{weightedsobolev}.
\begin{df} \label{def.spaces}
Let  $m\in \m N$, $1 \leq p<\infty$, $\delta \in \mathbb{R}$. The weighted Sobolev space $W^m_{\delta,p}$ is the completion of $C^\infty_0$ under the norm 
	$$\|u\|_{W^m_{\delta,p}}=\sum_{|\beta|\leq m}\|(1+|x|^2)^{\frac{\delta +|\beta|}{2}}{\nab}^\beta u\|_{L^p}.$$
We will use the notation $H^m_\delta = W^m_{\delta,2}$ {and $L^p_{\de}=W^0_{\de,p}$}.

	The weighted H\"older space $C^m_{\delta}$ is the complete space of $m$-times continuously differentiable functions under the norm 
	$$\|u\|_{C^m_{\delta}}=\sum_{|\beta|\leq m}\|(1+|x|^2)^{\frac{\delta +|\beta|}{2}}{\nab}^\beta u\|_{L^\infty}.$$
\end{df}

Finally, let us introduce the convention that we will use the above function spaces for both tensors and scalars on $\m R^2$, where the norms in the case of tensors are understood componentwise.

\section{The elliptic gauge and local well-posedness for \eqref{sys} and \eqref{back}}\label{sec.gauge}

\subsection{Elliptic gauge}\label{sec.elliptic.gauge}
We write the $(2+1)$-dimensional metric $g$ on $\q M:=I\times \mathbb R^2$ in the form
\begin{equation}\label{g.form.0}
g=-N^2dt^2 + \bar{g}_{ij}(dx^i + \beta^i dt)(dx^j + \beta^jdt).
\end{equation}
Let $\Sigma_t:=\{(s,x^1,x^2): s=t\}$ {and} $e_0= \partial_t -\beta^i\rd_i${, which is a future directed normal to $\Sigma_t$.}
We introduce the second fundamental form of the embedding $\Sigma_t \subset \q M$
\begin{equation}\label{K}K_{ij}=-\frac{1}{2N}\q L_{e_0} \bar{g}_{ij}.
\end{equation}
 We decompose $K$ into its trace and traceless parts.
\begin{equation}\label{K.tr.trfree} 
K_{ij}=:H_{ij}+\frac{1}{2}\bar{g}_{ij}\tau.
\end{equation}
Here, $\tau:=\mbox{tr}_{\bar{g}} {K}$ and $H_{ij}$ is therefore traceless with respect to $\bar{g}$.

Introduce the following \emph{gauge conditions}:
\begin{itemize}
	\item $\bar{g}$ is conformally flat{, i.e., for some function $\gamma$,}
	\begin{equation}\label{uniformized.g}
	\bar{g}_{ij}=e^{2\gamma}\delta_{ij};
	\end{equation}
	\item The constant $t$-hypersurfaces $\Sigma_t$ are maximal
	$$\tau=0.$$
\end{itemize}
{By \eqref{g.form.0}, it follows that 
\begin{equation}\label{g.form}
g=-N^2dt^2 + e^{2\gamma}\delta_{ij}(dx^i + \beta^i dt)(dx^j + \beta^jdt).
\end{equation}
Hence the determinant of $g$ is given by
\begin{equation}\label{g.det}
\det(g)=e^{2\gamma}\beta^2(-e^{4\gamma}\beta^2)+e^{2\gamma}(e^{2\gamma}(-N^2+e^{2\gamma}|\beta|^2)-e^{4\gamma}\beta^1\beta^1)=-e^{4\gamma}N^2.
\end{equation}
Moreover, the inverse $g^{-1}$ is given by
\begin{equation}\label{g.inverse}
g^{-1}=\frac{1}{N^2}\left(\begin{array}{ccc}-1 & \beta^1 & \beta^2\\
\beta^1 & N^2e^{-2\gamma}-\beta^1\beta^1 & -\beta^1 \beta^2\\
\beta^2 & -\beta^1 \beta^2 & N^2e^{-2\gamma}-\beta^2\beta^2
\end{array}
\right).
\end{equation}
}
In this gauge, the first equation in \eqref{back} implies that $H$, $\gamma$, $N$ and $\beta$ satisfy elliptic equations\footnote{This follows from the computations given in our companion paper \cite[Appendix~B]{HL}.}
\begin{align}
&{\de^{ik}}\partial_{{k}} H_{ij}= -\frac{e^{2\gamma}}{N}\left(2 (e_0 \phi)( \partial_j \phi) + \sum_{\bA}F_{\bA}^2 (e_0 u_\bA)( \partial_j u_{\bA})\right), \label{elliptic.1}\\
&\Delta \gamma = -|\nabla \phi|^2 -\frac{1}{2}\sum F_\bA^2|\nabla u_\bA|^2-\frac{e^{2\gamma}}{N^2}\left( (e_0 \phi)^2+\frac{1}{2}\sum_{\bA}F_{\bA}^2
( e_0 u_{{\bA}})^2\right) - \frac{1}{2}e^{-2\gamma}|H|^2,\label{elliptic.2}\\
&\Delta N =Ne^{-2\gamma}|H|^2+ \frac{e^{2\gamma}}{N}\left(2(e_0 \phi)^2+\sum_{\bA}F_{\bA}^2
(e_0 u_{{\bA}})^2\right),\label{elliptic.3}\\
& (L\beta)_{ij}=2Ne^{-2\gamma}H_{ij},\label{elliptic.4}
\end{align}
where $L$ is the conformal Killing operator given by 
\begin{equation}\label{L.def}
(L\beta)_{ij}:=\delta_{j\ell}\rd_i\beta^\ell+\delta_{i\ell}\rd_j\beta^\ell-\delta_{ij}\rd_k\beta^k.
\end{equation}
\subsection{Initial data}\label{sec.id}

{We recall in this section the notion of initial data that we introduce in \cite{HL}. This applies for the system \eqref{back}, and therefore in particular also for \eqref{sys}.}
\begin{df}[Admissible initial data]\label{def.data}
For $-\frac{1}{2}<\delta<0$, $k \geq 3$, $R>0$ and $\mathcal A$ a finite set, an {\bf admissible initial data set} with respect to the elliptic gauge for \eqref{back} consists of 
\begin{enumerate}
\item {{a} conformally flat} intrinsic metric $e^{2\gamma}\delta_{ij}\restriction_{\Sigma_0}$ {which admits a decomposition}
$${\gamma=-{\gamma_{asymp}} \chi({|x|}) \log ({|x|}) +\widetilde{\gamma},}$$
{where ${\gamma_{asymp}}\geq 0$ is a constant, $\chi({|x|})$ is a fixed smooth cutoff function with $\chi=0$ for ${|x|}\leq 1$ and $\chi=1$ for ${|x|}\geq 2$, and $\tilde{\gamma}\in H^{k+2}_{\delta}$}{;}
\item {{a} second fundamental form $(H_{ij})\restriction_{\Sigma_0} \in H^{k+1}_{\delta+1}$ which is traceless;}
 \item $(\f{1}{N}(e_0\phi),\nabla \phi) \restriction_{\Sigma_0}\in H^k$, compactly supported in $B(0,R)$;
		\item $F_{\bA} \restriction_{\Sigma_0} \in H^k$, compactly supported in $B(0,R)$ for every $\bA\in \mathcal A$;
		\item $u_{\bA}\restriction_{\Sigma_0}$ such that ${\inf_{x\in \mathbb R^2}}|\nabla u_{\bA}\restriction_{\Sigma_0}|{(x)}> {C_{eik}^{-1}}$ {for some $C_{eik}>0$} and $\left(\nabla u_{\bA}\restriction_{\Sigma_0}-\overrightarrow{c_{\bA}} \right)\in H^{k+1}_\delta$, where $\overrightarrow{c_{\bA}}$ is a constant vector field for every $\bA\in \mathcal A$.
\end{enumerate}
{$\gamma$ and $H$}
are required to satisfy the following {\bf constraint equations}:
\begin{align}
&{\de^{ik}\partial_k} H_{ij}=-\f{2e^{2\gamma}}{N}(e_0\phi)\partial_j \phi -\sum_{\bA} e^\gamma F_\bA^2|\nabla u_\bA|\partial_j u_\bA,\label{mom}\\
&\Delta \gamma + e^{-2\gamma}\left(\f{e^{4\gamma}}{N^2}(e_0\phi)^2+\frac{1}{2}|H|^2\right)+|\nabla \phi|^2+\sum_\bA F_\bA^2|\nabla u_\bA|^2=0.\label{ham}
\end{align}
\end{df}

{As we show in \cite{HL} (following \cite{Huneau.constraints}), one can define a notion of ``admissible free initial data'', which consist of (rescaled versions of) data for the matter field, so that one can uniquely construct initial data set as in Definition~\ref{def.data} by solving the constraints. We recall this definition here:
}
\begin{df}[Admissible free initial data]\label{def.free.data}
Define $\dot{\phi}$, $\breve{F}_\bA$ as follows:
\begin{equation}\label{data.rescaled}
\dot{\phi}=\frac{e^{2\gamma}}{N} (e_0 \phi),\quad\breve{F}_\bA= F_{\bA} e^{\frac{\gamma}{2}},
\end{equation}
where $\gamma$ is as in {\eqref{g.form}}.

For $-\frac{1}{2}<\delta<0$, $k \geq 3$, $R>0$ and $\mathcal A$ a finite set, an {\bf admissible free initial data set} with respect to the elliptic gauge is given by the following:
\begin{enumerate}
		\item $(\dot{\phi},\nabla \phi) \restriction_{\Sigma_0}\in H^k$, compactly supported in $B(0,R)$;
		\item $\breve{F}_{\bA} \restriction_{\Sigma_0} \in H^k$, compactly supported in $B(0,R)$ for every $\bA\in \mathcal A$;
		\item $u_{\bA}\restriction_{\Sigma_0}$ such that ${\inf_{x\in \mathbb R^2}}|\nabla u_{\bA}\restriction_{\Sigma_0}|{(x)}> {C_{eik}^{-1}}$ {for some $C_{eik}>0$} and $\left(\nabla u_{\bA}\restriction_{\Sigma_0}-\overrightarrow{c_{\bA}} \right)\in H^{{k+1}}_\delta$, where $\overrightarrow{c_{\bA}}$ is a constant vector field for every $\bA\in \mathcal A$.
	\end{enumerate}
Moreover, $(\dot{\phi},\nabla \phi,\breve{F}_\bA, u_\bA)\restriction_{\Sigma_0}$ is required to satisfy
\begin{equation}\label{main.data.cond}
\int_{\mathbb R^2} \left(-2\dot{\phi}\rd_j \phi- {\sum_{\bA}}\breve{F}_{{\bA}}^2|\nab u_{{\bA}}|\rd_j u_{{\bA}}\right) \, dx =0.
\end{equation}
\end{df}

\subsection{Local well-posedness}\label{seclwpintro}

The following is our main result {in \cite{HL}} on {the} local well-posedness for \eqref{back} (and therefore also \eqref{sys}). {While the theorem requires} a smallness assumption \eqref{smallness.fd}, {it is important to note that the smallness is \underline{not} needed for the higher norms. This is important for the applications to Theorem~\ref{main.intro} since} the approximating solutions $g_\lambda$ that we construct necessarily oscillate with very high frequency.\footnote{{We note that in the statement of Theorem~\ref{lwp} in \cite{HL}, different components of the metric are in spaces of different weight. While that is useful for the proof of Theorem~\ref{lwp}, it is irrelevant for this paper. We simply take the worst weight in the norms in \cite{HL} so that the decay part of the metric components $\gamma$, $N$ and $\beta$ are all in $H^{k+2}_\de$.}}
\begin{thm}[\cite{HL}]\label{lwp}
Let $-\frac{1}{2}<\delta<0$, $k \geq 3$, $R>0$ and $\mathcal A$ be a finite set. Given a free initial data set as in Definition \ref{def.free.data} such that
	\begin{equation}\label{smallness.fd}
\|\dot{\phi}\|_{L^\infty}+\|\nabla \phi\|_{L^\infty} + {\max_{\bA}}\|\breve{F}_{\bA}\|_{L^\infty}\leq \ep{,}
	\end{equation}
	{and}
	\begin{equation}\label{init.u.bd}
	{C_{eik}:= {\left(\min_{\bA} \inf_{x\in \mathbb R^2} |\nab u_{\bA}|(x)\right)^{-1} + \max_{\bA}}\|\nab u_{\bA}-\overrightarrow{c_{\bA}}\|_{H^{k+1}_{\delta}}<\infty,}
	\end{equation}
	and 
	$$C_{high}:=\|\dot{\phi}\|_{H^k}+\|\nabla \phi\|_{H^k}+ \|\breve{F}_{\bA}\|_{H^k} <\infty.$$
	Then{, for any $C_{eik}$ and $C_{high}$,} there exists a constant $\ep_{low}=\ep_{low}({C_{eik}, k,\delta,R})>0$ \underline{independent of $C_{high}$} and a $T=T(C_{high},{C_{eik},} k, \delta, R)>0$ such that if $\ep<\ep_{low}$, there exists a unique solution to \eqref{back} in elliptic gauge on $[0,T]\times \m R^2 $.
	Moreover, {the following holds for some constant $C_h = C_h(C_{eik}, C_{high}, k,\de,R)>0$:} 
	\begin{itemize}
	\item {The following estimates hold for $\phi$, $F_{\bA}$ and $u_{\bA}$ for all $\bA\in \mathcal A$ for $t\in [0,T]$:}
	\begin{align*}
	\|\nabla \phi\|_{H^k} +\|\partial_t \phi \|_{H^k} +\|\partial^2_t \phi\|_{H^{k-1}}\leq &C_h, \\
	{\max_\bA}\left(\|F_\bA \|_{H^k} + \|\partial_t F_{\bA}\|_{H^{k-1}}+\|{\partial^2_t} F_{\bA}\|_{H^{k-2}}\right)\leq &C_h, \\
	{\left(\min_{\bA} \inf_{x\in \mathbb R^2} |\nab u_{\bA}|(x)\right)^{-1}+}{\max_\bA}\left(\|\nabla u_{\bA}-\overrightarrow{c_{\bA}} \|_{ H^{k}_\delta} +\|e^{\gamma}N^{-1} (e_0 u_{\bA}) -|\overrightarrow{c_{\bA}}|\|_{H^{k}_\delta}\right)\leq & C_h,\\
	{\max_\bA}\left(\|\partial_t\nabla u_{\bA} \|_{ H^{k-1}_\delta}{+\|\partial_t^2\nabla u_{\bA} \|_{ H^{k-2}_\delta}}+\| \partial_t \left(\f{{e^{\gamma}}}{N} e_0 u_\bA \right)\|_{ H^{k-1}_\delta}
			+\|\partial_t^2\left(\f{{e^{\gamma}}}{N} e_0 u_\bA \right)\|_{H^{k-2}_\delta}\right)\leq &C_h.
				\end{align*}				
	\item The metric components $\gamma$ {and $N$} can be decomposed as
				$$\gamma = \gamma_{asymp}\chi(|x|)\log(|x|) + \wht \gamma,\quad N = 1 +N_{asymp}(t)\chi(|x|)\log(|x|) + \wht N,$$
				{with $\gamma_{asymp}\leq 0$ a constant, $N_{asymp}(t)\geq 0$ a function of $t$ alone and $\chi(|x|)$ is a fixed smooth, non-negative cutoff function supported in $\{{|x|}\geq 1\}$ which is identically $1$ for ${|x|}\geq 2$. }
	\item $\gamma$, $N$ and $\beta$ obey the following estimates for $t\in [0,T]$:
	\begin{align*}
	{|\gamma_{asymp}|} +\|{\wht \gamma}\|_{ H^{k+2}_{\delta}} + \|\partial_t {\wht \gamma}\|_{H^{k+1}_{{\delta}}} + \|\partial_t^2 {\wht \gamma} \|_{H^k_{{\delta}}}\leq &C_h, \\
	|N_{asymp}| +|\partial_t N_{asymp}|+|\partial^2_t N_{asymp}|\leq &C_h,\\
	\|\wht N\|_{H^{k+2}_\delta}+ \|\partial_t \wht N\|_{H^{k+1}_\delta}+\|\partial^2_t \wht N\|_{H^{k}_\delta}\leq &C_h,\\
	\|\beta \|_{H^{k+2}_{\delta}} + \|\partial_t \beta\|_{H^{k+1}_\delta}+\|\partial^2_t \beta\|_{H^{k}_\delta} \leq &C_h.
	\end{align*}
	\item The support of $\phi$ and $F_{\bA}$ satisfies\footnote{{Here, $J^+$ denotes the causal future.}} $${supp(\phi,F_{\bA}) \subset J^+(\Sigma_0{\cap} B(0, R)).}$$
	\end{itemize}	
\end{thm}

\begin{rk}
{From now on, let us fix a cutoff function $\chi$ with the properties as in the statement of Theorem~\ref{lwp}.}
\end{rk}

{Finally, as discussed in \cite{HL}, when specialized in the case that the data are genuinely small - in the sense that even the high norms are small -} the time of existence can be taken to be $T=1$:
\begin{cor}[\cite{HL}]\label{lwp.small}
Suppose the assumptions of Theorem~\ref{lwp} hold and let $\{c_{\bA}\}_{\bA\in \mathcal A}$ be a collection of constant vector fields on the plane. There exists $\ep_{small}=\ep_{small}(\delta,k,R,c_{\bA})$ such that if $C_{high}$ and $\ep$ in Theorem \ref{lwp} both satisfy
$$C_{high},\ep\leq \ep_{{small}}$$
and moreover
$$\sum_{\bA}\|\nabla u_{\bA}- \overrightarrow{c_{\bA}}\|_{H^{k+1}_\delta}\leq \ep_{{small}},$$
then the unique solution exists in $[0,1]\times \m R^2$. Moreover, there exists $C_0{= C_0(\de, k, R,c_{\bA})}$ such that {all the estimates in Theorem~\ref{lwp} hold with $C_h$ replaced by $C_0\ep$.}

\end{cor}

\section{Main result}\label{sec.main}

\subsection{Second version of the main theorem}\label{sec.main.2nd}

\begin{df}\label{ang.sep}
Given a solution to \eqref{back} on $I\times \m R^2$ for $I\subset \m R$ {(which is as regular as that in Theorem \ref{lwp}), w}e say that the set of eikonal functions $\{u_{\bA}\}_{\bA\in \mathcal A}$ is {\bf angularly separated} if there exists $\eta'\in (0,1)$ such that 
$$\f{\delta^{ij}(\rd_i u_{\bA_1})(\rd_j u_{\bA_2})}{|\nab u_{\bA_1}||\nab u_{\bA_2}|}(t,x)<1-\eta', \quad \forall (t,x)\in I\times\mathbb R^2,\quad\forall \bA_1\neq \bA_2.$$
\end{df}

\begin{thm}[{Main theorem}]\label{main.thm.2}
Let\footnote{Note that while the constants $k$, $\delta$ and $R$ do not appear explicitly in the statement of the theorem, they are needed in the definition of admissible free initial data set.} $k \geq 10$, $-\f 12<\delta<0$, $R>0$ and $\mathcal A$ be a finite set. Given an admissible free initial data set $(\dot{\phi},\nabla \phi,\breve F_\bA, u_\bA)\restriction_{\Sigma_0}$such that 
\begin{itemize}
\item $$|\nab u_\bA\restriction_{\Sigma_0}|>\f 12;$$
\item $u_\bA \restriction_{\Sigma_0}$ is angularly separated\footnote{Strictly speaking, angular separation was defined for spacetimes on $I\times \m R^2$, where $I$ in an interval, but it is easy to extend its definition to include the case $I=\{0\}$.};
\item (Smallness condition) $${\|\nab u_{\bA} - \overrightarrow{c_{\bA}}\|_{H^{11}_\de}} + \|\dot{\phi}\|_{H^{10}}+\|\nabla \phi\|_{H^{10}}+ \|\breve{F}_{\bA}\|_{H^{10}}\leq \ep{;}$$
\item (Genericity condition on initial data) {there} exists a point $p\in \m R^2$ such that 
	\begin{equation}\label{genericity}
	\left(\begin{array}{l}(\partial_1 \phi_0)\restriction_{\Sigma_0}\\ (\partial_1 \dot{\phi_0})\restriction_{\Sigma_0}\end{array}\right)(p)\mbox{ and }\left(\begin{array}{l}(\partial_2 \phi_0)\restriction_{\Sigma_0}\\ (\partial_2 \dot{\phi_0})\restriction_{\Sigma_0}\end{array}\right)(p)\mbox{ are linearly independent.}
	\end{equation}
\end{itemize}
Then, there exists $\ep_0>0$ such that if $\ep<\ep_0$, {a unique solution $(g_0,\phi_0, (F_0)_{\bA}, (u_0)_{\bA})$ to \eqref{back} arising from the given admissible free initial data set} exists on a time interval $[0,1]$, and there exists a one-parameter family of solutions $(g_{\lambda},\phi_{\lambda})$ to \eqref{sys} for $\lambda \in (0,\lambda_0)$ {(for some $\lambda_0\in \m R$ sufficiently small), which are all defined on the time interval $[0,1]$,} such that
	$$(g_{\lambda},\phi_{\lambda}) \rightarrow (g_0,\phi_0)\mbox{ { uniformly on compact sets}}$$
	 and 
	$$(\rd g_{\lambda}, \rd \phi_{\lambda})\rightharpoonup (\rd g_0,\rd \phi_0)\mbox{ { weakly in $L^2$}}$$ 
	with $\rd g_{\lambda}, \rd \phi_{\lambda}\in L^\infty_{loc}$ uniformly. 
\end{thm}

\begin{rk}
As we will see (for instance by considering \eqref{g.para.def} and the estimates for $\mfg$), $\rd g_\lambda$ converges in a stronger sense compared to $\rd \phi_\lambda$. In particular, $\rd g_\lambda \to \rd g_0$ uniformly (in $L^\infty$) on compact sets (of $[0,1]\times \m R^2$).
\end{rk}

{Clearly the existence and uniqueness of $(g_0,\phi_0, (F_0)_{\bA}, (u_0)_{\bA})$ - which from now on will be called the \emph{background solution} - follows from Corollary~\ref{lwp.small}. We therefore have to prove the existence of solutions $(g_{\lambda},\phi_{\lambda})$ to \eqref{sys} and their convergence, for some appropriately chosen initial data (cf. Lemma~\ref{lmini}). Since $F_{\bA} \equiv 0 $ for the solutions to \eqref{sys}, \textbf{we will from now on denote $F_{\bA} = (F_0)_{\bA}$ and $u_{\bA} = (u_0)_{\bA}$} to lighten the notations.
}

\subsection{Angularly separated, spatially adapted and null adapted eikonal functions}\label{sec.prepare}

In Theorem \ref{main.thm.2}, the eikonal functions $\{u_\bA\}_{\bA\in \mathcal A}$ are required to be initially angularly separated with $|\nab u_\bA \restriction_{\Sigma_0}|>\f 12$. It is easy to see that by choosing $\ep_0$ in Theorem \ref{main.thm.2} sufficiently small, $\{u_\bA\}_{\bA\in \mathcal A}$ is in fact angularly separated in $[0,1]\times \m R^2$ with the lower bound $|\nab u_\bA|>\f 14$ {(cf. Lemma~\ref{lm:as.local})}. However, in order to carry out the proof of Theorem \ref{main.thm.2} (and in particular to construct the parametrix), it will be convenient to modify the eikonal functions such that they satisfy the additional properties of being spatially adapted and null adapted (see Definitions \ref{def.spatial} and \ref{def.null} below). That this can always be achieved is proven in Lemma \ref{lemma.adapt}. 

We begin our discussion with the following lemma, which states that our assumptions on $\{u_\bA\}_{\bA\in \mathcal A}$ in Theorem \ref{main.thm.2} together with the smallness assumption, give us good control of $\{u_\bA\}_{\bA\in \mathcal A}$ in the spacetime. 
\begin{lm}\label{lm:as.local}
Under the assumptions of Theorem \ref{main.thm.2}, if $\ep_0$ is sufficiently small, the eikonal functions $\{u_\bA\}_{\bA\in \mathcal A}$ are angularly separated\footnote{{Notice that in $[0,1]\times \m R^2$ the estimate in Definition~\ref{ang.sep} may hold with a different constant compared to the initial data, but we will call both of these constants $\eta'$}} in $[0,1]\times \m R^2$ and satisfy $|\nab u_\bA|>\f 14$.
\end{lm}
\begin{proof}
This can be easily justified using the estimates in Corollary~\ref{lwp.small}.
\end{proof}

We now discuss the modifications to the eikonal functions. To begin, we make the easy observation that the following symmetry of the system \eqref{back} allows us to modify the eikonal functions:
\begin{lm}
Suppose $(g,\phi,F_{\bA},u_{\bA})$ is a solution to \eqref{back} on $I\times \m R^2$ for $I\subset \m R$ in the sense of Theorem \ref{lwp}. For any sets of positive\footnote{We choose $c_{\bA}$ to be positive so that we also preserve the property that $(e_0 u_\bA)>0$.} constants $\{c_{\bA}\}_{\bA\in \mathcal A}\in \mathbb R_{>0}^{|\mathcal A|}$, if we define
$$F_{\bA}'=c_{\bA}^{-1} F_{\bA},\quad u_{\bA}'=c_{\bA} u_{\bA},$$
then $(g,\phi,F_{\bA}',u_{\bA}')$ is also a solution to \eqref{back}.
\end{lm}

We now define the notions of spatially adapted and null adapted eikonal functions:
\begin{df}\label{def.spatial}
Given a solution to \eqref{back} on $I\times \m R^2$ for $I\subset \m R$ {(which is as regular as that in Theorem \ref{lwp})}, we say that {the} set of eikonal functions $\{u_{\bA}\}_{\bA\in \mathcal A'}$ is {\bf spatially adapted} for $\mathcal A'\subset \mathcal A$ if  
\begin{enumerate}
\item $|\nab u_{\bA}|> \f 14, \quad \forall \bA\in\mathcal A$,
\item $|\nab (u_{\bA_1}\pm u_{\bA_2})|> \f 14, \quad \forall \bA_1, \bA_2\in \mathcal A',\, \bA_1\neq \bA_2$,
\item $|\nab (u_{\bA_1}\pm 2u_{\bA_2})|> \f 14, \quad \forall \bA_1, \bA_2\in \mathcal A',\, \bA_1\neq \bA_2$,
\item $|\nab (u_{\bA_1}\pm 3u_{\bA_2})|> \f 14, \quad \forall \bA_1, \bA_2\in \mathcal A',\, \bA_1\neq \bA_2$.
\end{enumerate}
\end{df}

\begin{df}\label{def.null}
Given a solution to \eqref{back} on $I\times \m R^2$ for $I\subset \m R$ {(which is as regular as that in Theorem \ref{lwp})}, we say that {the} set of eikonal functions $\{u_{\bA}\}_{\bA\in \mathcal A'}$ is {\bf null adapted} for $\mathcal A'\subset \mathcal A$ if there exists a constant $\eta>0$ such that for any choice of the signs $\pm$, $\pm_1$, $\pm_2$, it holds that\footnote{Here, $\gamma$ is as in \eqref{g.form}.}
\begin{enumerate}
\item $|(g^{-1})^{\alp\bt} \rd_\alp (u_{\bA_1}\pm u_{\bA_2})\rd_\beta(u_{\bA_1}\pm u_{\bA_2}))|> \eta e^{-2\gamma}, \quad \forall \bA_1, \bA_2\in \mathcal A',\, \bA_1\neq \bA_2$,
\item $|(g^{-1})^{\alp\bt} \rd_\alp (u_{\bA_1}\pm 2 u_{\bA_2})\rd_\beta(u_{\bA_1}\pm 2 u_{\bA_2}))|> \eta e^{-2\gamma}, \quad \forall \bA_1, \bA_2\in \mathcal A',\, \bA_1\neq \bA_2$,
\item $|(g^{-1})^{\alp\bt} \rd_\alp (u_{\bA_1}\pm_1 u_{\bA_2}\pm_2 u_{\bA_3})\rd_\beta(u_{\bA_1}\pm_1 u_{\bA_2}\pm_2 u_{\bA_3}))|> \eta e^{-2\gamma}, \quad \forall \bA_1, \bA_2,\bA_3\in \mathcal A'$ all different.
\end{enumerate}
\end{df}

\begin{lm}\label{lemma.adapt}
Given a solution to \eqref{back} on $I\times \m R^2$ for $I\subset \m R$ {(which is as regular as that in Theorem \ref{lwp})} such that the eikonal functions $\{u_{\bA}\}_{\bA\in \mathcal A}$ are angularly separated with $|\nab u_{\bA}|> \f 14, \, \forall \bA\in\mathcal A$, there exists a collection of positive constants $\{c_\bA\}_{\bA\in \mathcal A}$ such that for $u'_\bA=c_\bA u_\bA$, $\{u'_{\bA}\}_{\bA\in \mathcal A}$ is spatially adapted and null adapted.
\end{lm}
\begin{proof}
We order the eikonal functions as $u_{\bA_1}, u_{\bA_2},\dots, u_{\bA_{{N}}}$, where ${N}=|\mathcal A|$. We then proceed inductively on $m$ to show that
\begin{enumerate}
\item The set $\cup_{i=1}^m \{u'_{\bA_i}\}$ (where $u_{\bA_i}'=c_{\bA_i} u_{\bA_i}$) is both spatially adapted and null adapted.
\item For every $i_1< i_2\leq m$ and $j>m$, the following inequality holds (notice that the following expression involves the primed $u'$ for $i_1$, $i_2$ and the unprimed $u$ for $j$):
\begin{equation}\label{well.prepared}
|\nab u'_{\bA_{i_2}}||\nab u_{\bA_j}|-\nab u'_{\bA_{i_2}}\cdot\nab u_{\bA_j}\geq 2\left(|\nab u'_{\bA_{i_1}}||\nab u_{\bA_j}|-\nab u'_{\bA_{i_1}}\cdot\nab u_{\bA_j}\right).
\end{equation}
\end{enumerate}
Once we have completed this induction, the lemma follows from the case $m={N}$.

Notice that the base case where $m=1$ is trivial. In the following, we will carry out the inductive step, assuming that $c_{\bA_i}, \dots, c_{\bA_{m-1}}$ ha{ve} been chosen such that (1) and (2) above hold with $m$ replaced by $m-1$. We will then choose $c_{\bA_m}>0$ such that (1) and (2) hold.

{\bf Proof of \eqref{well.prepared}{.}}
By the induction hypothesis, it suffices to prove \eqref{well.prepared} with $i_2=m$ and with any $i_1=i<m$, $j>m$. By the assumption on angular separation together with the lower bounds on $|\nab u_{\bA}|$, we have 
$$|\nab u'_{\bA_{m}}||\nab u_{\bA_j}|-\nab u'_{\bA_{m}}\cdot\nab u_{\bA_j}=c_{\bA_m}\left(|\nab u_{\bA_{m}}||\nab u_{\bA_j}|-\nab u_{\bA_{m}}\cdot\nab u_{\bA_j}\right)>\f{c_{\bA_m}\eta'}{16}.$$
On the other hand, since $c_{\bA_i}$ has been fixed, $|\nab u'_{\bA_{i}}||\nab u_{\bA_j}|-\nab u'_{\bA_{i}}\cdot\nab u_{\bA_j}$ is bounded above. We can therefore choose
\begin{equation}\label{c.cond.1}
c_{\bA_m}\geq \f{2\max_{i<m<j}\sup_{\m R^2}|\nab u'_{\bA_{i}}||\nab u_{\bA_j}|}{\min_{j>m}\inf_{\m R^2}\left(|\nab u'_{\bA_{m}}||\nab u_{\bA_j}|-\nab u'_{\bA_{m}}\cdot\nab u_{\bA_j}\right)}
\end{equation}
so that
\begin{equation*}
\begin{split}
|\nab u'_{\bA_{m}}||\nab u_{\bA_j}|-\nab u'_{\bA_{m}}\cdot\nab u_{\bA_j}=&c_{\bA_m}\left(|\nab u_{\bA_{m}}||\nab u_{\bA_j}|-\nab u_{\bA_{m}}\cdot\nab u_{\bA_j}\right)\\
\geq &2\left(|\nab u'_{\bA_{i}}||\nab u_{\bA_j}|-\nab u'_{\bA_{i}}\cdot\nab u_{\bA_j}\right),
\end{split}
\end{equation*}
i.e., \eqref{well.prepared} holds for $i_2=m$, $i_1-i<m$ and $j>m$.

{\bf Spatially adapted{.}}
By the induction hypothesis, we only need to verify the conditions in Definition \ref{def.spatial} when one of the eikonal functions involved is $u_{\bA_m}$. 

First, it is easy to see that if 
\begin{equation}\label{c.cond.2.0}
c_{\bA_m}>1,
\end{equation}
Then $|\nab u'_{\bA_m}|>\f 14$.

Next, we compute that for $i<m$
\begin{equation}\label{spatial.compute}
\begin{split}
|\nab(u_{\bA_i}'\pm u'_{\bA_m})|^2 \geq & c_{\bA_m}^2|\nab u_{\bA_m}|^2-2c_{\bA_m} |\nab u_{\bA_m}||\nab u'_{\bA_i}|,\\
|\nab(u_{\bA_i}'\pm 2 u'_{\bA_m})|^2 \geq & 4 c_{\bA_m}^2|\nab u_{\bA_m}|^2-4c_{\bA_m} |\nab u_{\bA_m}||\nab u'_{\bA_i}|,\\
|\nab(2u_{\bA_i}'\pm u'_{\bA_m})|^2 \geq & c_{\bA_m}^2|\nab u_{\bA_m}|^2-4c_{\bA_m} |\nab u_{\bA_m}||\nab u'_{\bA_i}|,\\
|\nab(u_{\bA_i}'\pm 3 u'_{\bA_m})|^2 \geq & 9 c_{\bA_m}^2|\nab u_{\bA_m}|^2-6c_{\bA_m} |\nab u_{\bA_m}||\nab u'_{\bA_i}|,\\
|\nab(3u_{\bA_i}'\pm u'_{\bA_m})|^2 \geq & c_{\bA_m}^2|\nab u_{\bA_m}|^2-6c_{\bA_m} |\nab u_{\bA_m}||\nab u'_{\bA_i}|.
\end{split}
\end{equation}
Each of these is required to be $>\f 1{16}$. Using the fact that $\inf_{\m R^2}|\nab u_{\bA_m}|\geq \f 14$, we can therefore choose
\begin{equation}\label{c.cond.2}
c_{\bA_m}\geq \f{12\max_{i<m} \sup_{\m R^2}|\nab u'_{\bA_i}|}{\inf_{\m R^2}|\nab u_{\bA_m}|}+{2}
\end{equation}
such that for the {RHS} of each line in \eqref{spatial.compute}
\begin{equation}
\begin{split}
& RHS\\
\geq & c_{\bA_m}^2|\nab u_{\bA_m}|^2-6c_{\bA_m} |\nab u_{\bA_m}||\nab u_{\bA_i}|\\
{>} & c_{\bA_m}^2|\nab u_{\bA_m}|^2-\frac{{c_{\bA_m}^2}}{2}|\nabla u_{\bA m}|\inf_{\m R^2}|\nab u_{\bA_m}|\\
\geq &\frac{1}{2}c_{\bA_m}^2|\nabla u_{\bA_m}|^2 \geq \frac{1}{16}.
\end{split}
\end{equation}

{\bf Null adapted: the easy cases{.}} By the induction hypothesis, in order to verify that $\cup_{i=1}^m \{u'_{\bA_i}\}$ is null adapted, we only need to check the conditions in Definition \ref{def.null} where one of the eikonal functions is $u_{\bA_m}$. We first check the first two conditions in Definition \ref{def.null}, which are easier and they hold for any choice of $c_{\bA_m}>0$.

As a preliminary step, note that for any $\bA$, since $u_{\bA}$ is an eikonal function, by \eqref{g.form}, we have $e_0 u_{\bA}=N e^{-\gamma}|\nab u_{\bA}|$. Therefore, we compute for the sum (or difference) of two eikonal functions with $i<m$
\begin{equation*}
\begin{split}
&(g^{-1})^{\alp\bt} \rd_\alp (u'_{\bA_i}\pm u'_{\bA_m})\rd_\beta(u'_{\bA_i}\pm u'_{\bA_m})\\
=&-\f 1{N^2}(e_0 (u'_{\bA_i}\pm u'_{\bA_m}))(e_0 (u'_{\bA_i}\pm u'_{\bA_m}))+e^{-2\gamma}\nab (u'_{\bA_i}\pm u'_{\bA_m})\cdot \nab (u'_{\bA_i}\pm u'_{\bA_m})\\
=&2 e^{-2\gamma}\left(\mp |\nab u'_{\bA_i}||\nab u'_{\bA_m}|\pm \nab u'_{\bA_i}\cdot \nab u'_{\bA_m}\right).
\end{split}
\end{equation*}
By Cauchy--Schwarz, Definition \ref{ang.sep} and the assumption $|\nab u'_{\bA_i}|\geq \f 14$, $|\nab u_{\bA_m}|\geq \f 14$,
\begin{equation*}
\begin{split}
&\left|{\mp}|\nab u'_{\bA_i}||\nab u'_{\bA_m}|\pm \nab u'_{\bA_i}\cdot \nab u'_{\bA_m}\right|\\
\geq &|\nab u'_{\bA_i}||\nab u'_{\bA_m}|- \nab u'_{\bA_i}\cdot \nab u'_{\bA_m}\\
\geq &c_{\bA_m}\eta'|\nab u'_{\bA_i}||\nab u_{\bA_m}|\geq \f{c_{\bA_m}\eta'}{16}.
\end{split}
\end{equation*}
This verifies the first condition in Definition \ref{def.null} for any $c_{\bA_m}\geq 1$. The second condition in Definition \ref{def.null} can be checked in an identical manner.

{\bf Null adapted: the hard case.} It remains to check the condition in Definition~\ref{def.null} involving three distinct eikonal function. By the induction hypothesis, it suffices to consider $i<k<m$, in which case we compute
\begin{equation}\label{eikonal.3}
\begin{split}
&(g^{-1})^{\alp\bt} \rd_\alp (u'_{\bA_i}\pm_1 u'_{\bA_k}\pm_2 u'_{\bA_m})\rd_\beta(u'_{\bA_i}\pm_1 u'_{\bA_k} \pm_2 u'_{\bA_m}))\\
=&-\f 1{N^2}(e_0 (u'_{\bA_i}\pm_1 u'_{\bA_k}\pm_2 u'_{\bA_m}))(e_0 (u'_{\bA_i}\pm_1 u'_{\bA_k}\pm_2 u'_{\bA_m}))\\
&+e^{-2\gamma}\nab (u'_{\bA_i}\pm_1 u'_{\bA_k}\pm_2  c_{\bA_m} u_{\bA_m})\cdot \nab (u'_{\bA_i}\pm_1 u'_{\bA_k}\pm_2 c_{\bA_m} u_{\bA_m})\\
=&2 e^{-2\gamma}\underbrace{\left(\mp_1 |\nab u'_{\bA_i}||\nab u_{\bA'_k}|\pm_1 \nab u'_{\bA_i}\cdot \nab u'_{\bA_k}\right)}_{=:I}\\
&+2 c_{\bA_m} e^{-2\gamma}\underbrace{\left(\mp_2 |\nab u'_{\bA_i}||\nab u_{\bA_m}|\pm_2 \nab u'_{\bA_i}\cdot \nab u_{\bA_m}\right)}_{=:II}\\
&+2 c_{\bA_m} e^{-2\gamma}\underbrace{\left(\mp_1\mp_2 |\nab u'_{\bA_k}||\nab u_{\bA_m}|\mp_1\pm_2 \nab u'_{\bA_k}\cdot \nab u_{\bA_m}\right)}_{=:III}.
\end{split}
\end{equation}
Let us note that by the angular separation assumption, each of $I$, $II$ and $III$ are bounded below. It remains to check that not too much cancellations can occur. The key is to note that by the induction hypothesis (more precisely, by \eqref{well.prepared} with $m$ replaced by $m-1$), we have
$$|II+III|\geq |\nab u'_{\bA_i}||\nab u_{\bA_m}|- \nab u'_{\bA_i}\cdot \nab u_{\bA_m},$$
which, by the angular separation assumption and the lower bound on $|\nab u_{\bA}|$, is uniformly bounded below for any $i<m$. Therefore, by choosing
\begin{equation}\label{c.cond.3}
c_{\bA_m}\geq \f{2\max_{i<k<m}\sup_{\m R^2}|\nab u'_{\bA_i}||\nab u'_{\bA_k}|}{\min_{i< m}\inf_{\m R^2}|\nab u'_{\bA_i}||\nab u_{\bA_m}|- \nab u'_{\bA_i}\cdot \nab u_{\bA_m}},
\end{equation}
\eqref{eikonal.3} is bounded below as follows
$$\eqref{eikonal.3}\geq 2{e^{{-}2\gamma}}\max_{i<k<m}\sup_{\m R^2}|\nab u'_{\bA_i}||\nab u'_{\bA_k}|\geq \f 18 e^{-2\gamma},$$
and therefore satisfies the required condition.

{\bf End of proof.}
To conclude the induction step, it is easy to see that $c_{\bA_m}$ can be chosen such that \eqref{c.cond.1}, \eqref{c.cond.2.0}, \eqref{c.cond.2} and \eqref{c.cond.3} are simultaneously verified. This then gives the desired result.
\end{proof}

{\bf From now on, we assume that the eikonal functions $u_{\bA}$ with respect to the background solution are both spatially adapted and null adapted.} 

\subsection{Construction of the free initial data for $(g_{\lambda},\phi_\lambda)$}

{We construct the one parameter family $(g_{\lambda},\phi_\lambda)$ by specifying appropriate initial data to \eqref{sys}. In the remainder of the paper, we will prove that the solutions, guaranteed to exist locally by Theorem~\ref{lwp}, in fact exist throughout the time interval $[0,1]$ and have the desired convergence properties as $\lambda \to 0$. The following lemma constructs the initial data:}

\begin{lm}\label{lmini}Under the assumptions of Theorem \ref{main.thm.2}, there exists $\lambda_0>0$ sufficiently small such that for every $\lambda\in (0,\lambda_0]$, there exist initial functions $\phi_{{\lambda}} \restriction_{\Sigma_0}$ and $\dot{\phi}_{{\lambda}}\restriction_{\Sigma_0}$ (for \eqref{sys}) which are compactly supported in $B(0,R)$, such that
\begin{itemize}
\item $$\left\|\phi_{{\lambda}} \restriction_{\Sigma_0}- \phi_0\restriction_{\Sigma_0} - \sum_\bA  \lambda F_\bA \cos\left(\frac{ u_\bA}{\lambda}\right)\restriction_{\Sigma_0}\right\|_{H^5}\leq \sqrt{2}\ep\lambda^2,$$
$$\left\|\dot{\phi}_{{\lambda}}\restriction_{\Sigma_0}- \dot{\phi_0}\restriction_{\Sigma_0} + \sum_\bA e^{\gamma_0}|\nabla u_\bA|  F_\bA \sin\left(\frac{ u_\bA}{\lambda}\right)\restriction_{\Sigma_0}\right\|_{H^4}\leq \sqrt{2}\ep\lambda^2;$$
\item The following condition holds initially
$$\int_{\Sigma_0} \dot{\phi}_{{\lambda}}\,\partial_j\phi_{{\lambda}}\, dx=0.$$
\end{itemize}
\end{lm}
\begin{proof}
We will construct $\phi_{{\lambda}} \restriction_{\Sigma_0}$ and $\dot{\phi}_{{\lambda}}\restriction_{\Sigma_0}$ of the following form:
\begin{equation}\label{phi.i.def}
\phi_{{\lambda}} \restriction_{\Sigma_0}= \phi_0\restriction_{\Sigma_0} + \sum_\bA  \lambda F_\bA \cos\left(\frac{ u_\bA}{\lambda}\right)\restriction_{\Sigma_0} + \Omega_0 r_0,
\end{equation}
\begin{equation}\label{phid.i.def}
\dot{\phi}_{{\lambda}}\restriction_{\Sigma_0}= \dot{\phi_0}\restriction_{\Sigma_0} - \sum_\bA e^{\gamma_0}|\nabla u_\bA|  F_\bA \sin\left(\frac{ u_\bA}{\lambda}\right)\restriction_{\Sigma_0}
+\Omega_1 r_1,
\end{equation}
where $r_0, r_1:\m R^2\to \m R$ are smooth functions compactly supported in $B(0,R)$ {independent of $\lambda$ to be chosen below}, with
\begin{equation}\label{r.small}	
\|r_0\|_{H^{10}}+\|r_1\|_{H^{9}}\leq \ep,
\end{equation}
and $\Omega_0, \,\Omega_1\in \mathbb R$ {(depending on $\lambda$)}.

By the genericity assumption \eqref{genericity}, there exists $r_0$ and $r_1$ (smooth and compactly supported in $B(0,R)$) such that for some $c_p>0$,
\begin{equation}\label{gen.quan}
\left|\det\left(\begin{array}{ll}
-\int_{\Sigma_0} r_0 \partial_1 \dot{\phi_0}\, dx &\int_{\Sigma_0} r_1 \partial_1 \phi_0\, dx\\
-\int_{\Sigma_0} r_0 \partial_2 \dot{\phi_0}\, dx&  \int_{\Sigma_0} r_1 \partial_2 \phi_0\, dx\\
\end{array}\right)\right|> c_p>0.
\end{equation}
This can for instance be achieved by taking $r_0$ and $r_1$ to be supported in a very small neighborhood of $p$ (as in \eqref{genericity}). 

Without loss of generality, we can rescale the functions $r_0$ and $r_1$ so that \eqref{r.small} holds. Notice that we must have $c_p\lesssim \ep^4$, but in general, it is possible that $c_p\ll \ep$.
	
Our goal now is to find $\Omega_0$ and $\Omega_1$ (satisfying appropriate bounds) such that
\begin{equation}\label{goal.Om}
\int_{\Sigma_0} \dot{\phi}_{{\lambda}}\,\partial_j \phi_{{\lambda}}\, dx= 0.
\end{equation}
This condition can be written as\footnote{where for simplicity we have suppressed the measure $dx$.}
\begin{equation}\label{main.cond.data}
\begin{split}
0=&\int_{\Sigma_0} \underbrace{\Omega_0 \partial_j r_0  \dot{\phi_0}+ \Omega_1 r_1 \partial_j \phi_0}_{=:Main\;term}
+  \underbrace{\dot{\phi_0}\partial_j \phi_0 +  \sum_\bA e^{\gamma_0}F_\bA^2 |\nabla u_{\bA}|\partial_j u_{\bA}\sin^2\left(\frac{ u_\bA}{\lambda}\right)}_{Background\; terms} \\
&+ \int_{\Sigma_0} \underbrace{-\Omega_0 \partial_j r_0 \sum_\bA  e^{\gamma_0}|\nabla u_\bA|  F_\bA \sin\left(\frac{ u_\bA}{\lambda}\right)}_{=:HF_1}+\underbrace{\Omega_1r_1\partial_j\left(\sum_\bA  \lambda F_\bA \cos\left(\frac{ u_\bA}{\lambda}\right)\right)}_{=:HF_2}\\
&+\int_{\Sigma_0} \underbrace{-\partial_j \phi_0  \sum_\bA  e^{\gamma_0}|\nabla u_\bA|  F_\bA \sin\left(\frac{ u_\bA}{\lambda}\right)}_{=:HF_3}+\underbrace{\dot{\phi_0}\partial_j\left(\sum_{{\bA}}  \lambda F_\bA \cos\left(\frac{ u_\bA}{\lambda}\right)\right)}_{=:HF_4}\\
&-\int_{\Sigma_0} \underbrace{\left(\sum_\bA  e^{\gamma_0}|\nabla u_\bA|  F_\bA \sin\left(\frac{ u_\bA}{\lambda}\right)\right)  \left(\sum_\bB  \lambda(\partial_j F_\bB) \cos\left(\frac{ u_\bB}{\lambda}\right)\right)}_{=:HF_5}\\
&+\int_{\Sigma_0} \underbrace{\sum_{\bA}\sum_{\bB \neq  \bA} e^{\gamma_0}F_\bA F_{\gra{B}} |\nabla u_{\bA}|\partial_j u_{\gra B}\sin\left(\frac{ u_\bA}{\lambda}\right)\sin\left(\frac{ u_{\gra B}}{\lambda}\right)}_{=:HF_6}+\underbrace{\Omega_0 \Omega_1\partial_j r_0 r_1}_{=:Small\;term}.
\end{split}
\end{equation}
We first consider the ``background terms'' above. We have 
$$\sum_\bA e^{\gamma_0}F_\bA^2 |\nabla u_{\bA}|\partial_j u_{\bA}\sin^2\left(\frac{ u_\bA}{\lambda}\right) 
=\f 12\sum_\bA \breve{F}_\bA^2 |\nabla u_{\bA}|\partial_j u_\bA\left(1+\cos\left(\frac{ 2 u_\bA}{\lambda}\right) \right).$$
Let $v_\bA=2u_\bA$. Since $\{u_\bA\}_{\bA\in \mathcal A}$ is spatially adapted (see Definition \ref{def.spatial}), we have $|\nabla u_{\bA} |> \f 14$ and consequently $|\nabla v_{\bA} |> \f 12$. Therefore, we can iteratively apply $\si{\f{v_{\bA}}{\lambda}}=-\f{\lambda^2\Delta\big(\si{\f{v_{\bA}}{\lambda}}\big)}{|\nab v_{\bA}|^2}+\f{\lambda(\Delta v_{\bA})\co{\f{v_{\bA}}{\lambda}}}{|\nab v_{\bA}|^2}$ and $\co{\f{v_{\bA}}{\lambda}}=-\f{\lambda^2\Delta\big(\co{\f{v_{\bA}}{\lambda}}\big)}{|\nab v_{\bA}|^2}-\f{\lambda(\Delta v_{\bA})\si{\f{v_{\bA}}{\lambda}}}{|\nab v_{\bA}|^2}$ and integrate by parts to obtain\footnote{Here, and below in this proof, the constant depends on the initial data for $u_\bA$.}
$$\int_{\Sigma_0} \breve{F}_\bA^2 |\nabla u_{\bA}|\cos\left(\frac{2 u_\bA}{\lambda}\right)\, dx= O(\ep^2\lambda^3).$$
Notice that we indeed have (more than) sufficient regularity for the background to justify this calculation. By the hypothesis on the initial data for the background metric we have
$$\int_{\Sigma_0}\left( 2\dot{\phi_0}\partial_j \phi_0+ \sum \breve{F}_\bA^2 |\nabla u_{\bA}|\partial_j u_\bA\right)\, dx= 0.$$
Therefore, we have
$$|\mbox{``Background term''}|= O(\ep^2\lambda^3).$$

Next, note that all the terms that we have denoted as ``high frequency'' (HF), i.e., $HF_1$, $HF_2$, $HF_3$, $HF_4$, $HF_5$ and $HF_6$, contain an oscillating factor $\si{\f{v}{\lambda}}$ or $\co{\f{v}{\lambda}}$ with $|\nab v|>\f 14$ since $\{u_\bA\}_{\bA\in \mathcal A}$ is spatially adapted. Here, we have used the simple trigonometric identities
$$\sin\left(\frac{ u_\bA}{\lambda}\right)\sin\left(\frac{ u_{\gra B}}{\lambda}\right)
=\f 12\sum_{\pm}(\mp 1)\cos\left(\frac{u_\bA\pm u_{\gra B}}{\lambda}\right),\quad \sin\left(\frac{ u_\bA}{\lambda}\right)\cos\left(\frac{ u_{\gra B}}{\lambda}\right)
=\f 12\sum_{\pm}\sin\left(\frac{u_\bA\pm u_{\gra B}}{\lambda}\right).$$
Therefore, by the argument as above, we also have
$$HF_1 =O(\ep^2\lambda^3)\Omega_0,\quad HF_2 =O(\ep^2\lambda^3)\Omega_1.$$
and for $i=3,4,5, 6$,
$$|HF_i |=O(\ep^2\lambda^3).$$

Therefore, to solve for \eqref{main.cond.data} for $j=1,2$ is equivalent to finding $\Omega_0,\Omega_1$ solving the following problem
\begin{equation}\label{fd.prob}
{\bf A}\left(\begin{array}{l}\Omega_0\\\Omega_1
\end{array}\right)= {\bf a}+{\bf B}\left(\begin{array}{l}\Omega_0\\\Omega_1
\end{array}\right)+\Omega_0\Omega_1{\bf c},
\end{equation}
where ${\bf A}=\left(\begin{array}{ll}
-\int_{\Sigma_0} r_0 \partial_1 \dot{\phi_0} \, dx &\int_{\Sigma_0} r_1 \partial_1 \phi_0\, dx\\
-\int_{\Sigma_0} r_0 \partial_2 \dot{\phi_0} \, dx&  \int_{\Sigma_0} r_1 \partial_2 \phi_0\, dx\\
\end{array}\right)$, ${\bf a}$ is a $2$D vector with $\|{\bf a}\|_2\lesssim \ep^2\lambda^3$, ${\bf B}$ is a $2\times 2$ matrix with $\|{\bf B}\|_2\lesssim \ep^2\lambda^3$ and ${\bf c}$ is a $2$D vector with $\|{\bf c}\|_2\lesssim \ep^2$. Since the determinant is a continuous function, by choosing $\lambda$ sufficiently small (depending on $c_p$), ${\bf A}-{\bf B}$ is invertible. We can thus define the (continuous) map $\Phi:\m R^2 \to \m R^2$ by
$$\Phi\left(\begin{array}{l}x_0\\ x_1
\end{array}\right)
=\left({\bf A}-{\bf B}\right)^{-1}\left({\bf a}+x_0 x_1 {\bf c}\right).$$
Choosing $\lambda$ sufficiently small (say, $\lambda^{\f 12}\ll c_p^{-1}$), it is easy to verify that $\Phi$ maps the closed ball $\overline{B(0,\lambda^2)}$ to itself. Hence, by Brouwer's fixed point theorem, $\Phi$ has a fixed point $\left(\begin{array}{l}\Omega_0\\\Omega_1
\end{array}\right)$, which is a solution to \eqref{fd.prob} satisfying
\begin{equation}\label{Om.small}
|\Omega_0|+|\Omega_1|\leq \sqrt{2}\sqrt{\Omega_0^2+\Omega_1^2}\leq \sqrt{2}\lambda^2.
\end{equation}
Therefore, we have now found $\phi_{{\lambda}}\restriction_{\Sigma_0}$ and $\dot{\phi}_{{\lambda}}\restriction_{\Sigma_0}$ such that \eqref{goal.Om} holds. Moreover, combining \eqref{phi.i.def}, \eqref{phid.i.def}, \eqref{r.small} and \eqref{Om.small}, it is easy to check that the required bounds for $\phi_{{\lambda}}\restriction_{\Sigma_0}$ and $\dot{\phi}_{{\lambda}}\restriction_{\Sigma_0}$ hold.
\end{proof}

\begin{rk}
We note from the proof above that the choice of $\lambda_0$ not only depends on the initial $u_\bA\restriction_{\Sigma_0}$, but also depends on the profile of the background solution. This is because the choice of $\lambda_0$ depends on $c_p$ in \eqref{gen.quan}, which can be viewed as a more quantitative version of \eqref{genericity}.
\end{rk}

\subsection{Construction of the parametrix}\label{sec.parametrix}

By Theorem~\ref{lwp}, the one-parameter family of free admissible initial data to \eqref{sys} constructed in Lemma~\ref{lmini} gives rise to a unique one-parameter family of solutions in the {elliptic} gauge. To analy{z}e the solution arising from data given in Lemma~\ref{lmini}, we decompose the scalar field $\phi_{\lambda}$ and the metric $g_{\lambda}$. {\bf It will be convenient notationally to suppress the subscripts ${ }_{\lambda}$ when there is no danger of confusion. We will however keep the subscripts ${ }_0$ for the background $\phi_0$ and $g_0$.}

The definitions of the parametrix for $\phi_\lambda$ and $g_{\lambda}$ are coupled in the sense that their mains terms are defined not only in terms of the background solution, but are defined to satisfy a coupled system of PDEs. {A}t this point, it is already useful to keep in mind that we will decompose $\phi_\lambda$ as
\begin{align*}
\phi_\lambda= &\phi_0+\sum_\bA \lambda F_\bA \cos\left(\frac{u_\bA}{\lambda}\right) +\sum_\bA \lambda^2 \wht F_\bA \si{\frac{u_\bA}{\lambda}}\\
&+\sum_\bA \lambda^2 \wht F_\bA^{(2)}\co{\frac{2u_\bA}{\lambda}}+\sum_\bA  \lambda^2\wht F_\bA^{(3)} \si{\frac{3u_\bA}{\lambda }} + \mathcal E_\lambda
\end{align*}
and decompose each metric component $\mfg$ (i.e., $\mfg\in \{\gamma,N,\beta^i\}$) as
$$\mfg=\mfg_0+\mfg_1+\mfg_2+\mfg_3,$$
where $\mfg_0$ is the corresponding metric component of the background solution. We will define all of these terms below in Sections \ref{sec.para.phi} and \ref{sec.para.g}.

\subsubsection{Parametrix for $\phi_\lambda$}\label{sec.para.phi}
The parametrix for $\phi_\lambda$ is constructed as follows:
\begin{equation}\label{phi.para}
\begin{split}
\phi_\lambda= &\phi_0+\sum_\bA \lambda F_\bA \cos\left(\frac{u_\bA}{\lambda}\right) +\sum_\bA \lambda^2 \wht F_\bA \si{\frac{u_\bA}{\lambda}}\\
&+\sum_\bA \lambda^2 \wht F_\bA^{(2)}\co{\frac{2u_\bA}{\lambda}}+\sum_\bA  \lambda^2\wht F_\bA^{(3)} \si{\frac{3u_\bA}{\lambda }} + \mathcal E_\lambda.
\end{split}
\end{equation}
Here, $\phi_0$, $F_\bA$ and $u_\bA$ are the background quantities; $\mathcal E_\lambda$ is the ``error term\footnote{As we will see later in Section \ref{secelambda}, we will further decompose $\mathcal E_\lambda$ into extra terms and some of which can be viewed as ``main terms'' of $\mathcal E_\lambda$. Nevertheless, at this point of the discussion, it suffices to note that each of these decomposed terms of $\mathcal E_\lambda$ are indeed ``smaller'' than the terms $\phi_0$, $\sum_\bA \lambda F_\bA \cos\left(\frac{u_\bA}{\lambda}\right)$, $\sum_\bA \lambda^2 \wht F_\bA \si{\frac{u_\bA}{\lambda}}$, $\sum_\bA \lambda^2 \wht F_\bA^{(2)}\co{\frac{2u_\bA}{\lambda}}$ and $\sum_\bA  \lambda^2\wht F_\bA^{(3)} \si{\frac{3u_\bA}{\lambda }}$}''; and  $\wht F_\bA$, $\wht F_\bA^{(2)}$ and $\wht F_\bA^{(3)}$ are defined to be solutions to the following transport equations {\bf with zero initial data}:
\begin{equation}
\label{whtfi}2(g_0^{-1})^{\alpha \beta}\partial_\alpha u_\bA \partial_\beta \wht F_\bA +(\Box_{g_0} u_\bA) \wht F_\bA
= \frac{1}{ \lambda^2} F_\bA (g_3')^{\alpha \beta}\partial_\alpha u_\bA \partial_\beta u_\bA- \Box_{g_0} F_\bA +A_\bA^{(1)},
\end{equation}
\begin{equation}
\label{whtgi2}2(g_0^{-1})^{\alpha \beta}\partial_\alpha u_\bA \partial_\beta \wht F^{(2)}_\bA +(\Box_{g_0} u_\bA )\wht F_\bA^{(2)} 
=A_\bA^{(2)},
\end{equation}
\begin{equation}
\label{whtgi3}2(g_0^{-1})^{\alpha \beta}\partial_\alpha u_\bA \partial_\beta \wht F^{(3)}_\bA +(\Box_{g_0} u_\bA) \wht F_\bA^{(3)} 
=A_\bA^{(3)},
\end{equation}
where
\begin{itemize}
\item $A_\bA^{(1)}$, $A_\bA^{(2)}$ and $A_\bA^{(3)}$ are compactly {supported} expressions which depend only on the background (and can be explicitly written down\footnote{We however will not do that as the explicit expressions are rather tedious.}) and obey the bounds
\begin{equation}\label{A.bound.1}
\|A^{(a)}_\bA\|_{H^8\cap {C^{8}}}+\|\rd_tA^{(a)}_\bA\|_{H^7\cap {C^{7}}}\leq C(C_0)\ep^3\quad\mbox{  for $a=1,2,3$}  ;
\end{equation}
and 
\item for $a={1,3}$, we define
\begin{equation}\label{gp.def}
\begin{split}
g_a'=&\left(\begin{array}{ccc}\f{2N_a}{N_0^3} & 0 & 0\\
0 & e^{-2\gamma_0}(-1+e^{-2\gamma_a}) & 0\\
0 & 0 & e^{-2\gamma_0}(-1+e^{-2\gamma_a})
\end{array}
\right)\\
&-\frac{2 N_a}{N_0^3}\left(\begin{array}{ccc}0 & \beta_0^1 & \beta_0^2\\
(\beta_0)^1 & -(\beta_0)^1(\beta_0)^1 & -(\beta_0)^1 (\beta_0)^2\\
(\beta_0)^2 & -(\beta_0)^1 (\beta_0)^2 & -(\beta_0)^2(\beta_0)^2
\end{array}
\right)\\
&+\frac{1}{N_0^2}\left(\begin{array}{ccc}0 & (\beta_a)^1 & (\beta_a)^2\\
(\beta_a)^1 & -2(\beta_0)^1(\beta_a)^1 & -(\beta_0)^1 (\beta_a)^2-(\beta_a)^1 (\beta_0)^2\\
(\beta_a)^2 & -(\beta_0)^1 (\beta_a)^2-(\beta_a)^1 (\beta_0)^2 & -2(\beta_0)^2(\beta_a)^2
\end{array}
\right).
\end{split}
\end{equation}
\end{itemize}
We will not define until Section \ref{sec.sf.main} the expressions $A_\bA^{(1)}$, $A_\bA^{(2)}$ and $A_\bA^{(3)}$ (see Definition \ref{A.def}). They are of course chosen so that {we} indeed {have a} good parametrix, i.e., so that the main terms in the expression $\Box_g\phi_\lambda$ cancel.

We remark that $g_a'$ is constructed as follows: Let $g_0+g_a$ be the matrix with components given by $\mfg_0+\mfg_a$. Then $g_a'$ is { the linear in $\mfg_a$ term of the expression $(g_0+g_a)^{-1}-(g_0)^{-1}$.}

Here, two important remarks are {in order}:
\begin{itemize}
\item In the definition above, $F_\bA$ and $u_\bA$ manifestly depend only on the background. The functions $\wht F^{(2)}_\bA$ and $\wht F^{(3)}_\bA$ are also chosen to depend only on the background. Therefore, the only part that interacts with the metric $g_\lambda$ are $\wht F_\bA$ and the error term $\mathcal E_\lambda$.
\item The parametrix as defined above captures up to order\footnote{i.e., terms that are $\sim \lambda^2$ with a constant depending on $\ep$.} $O_\ep(\lambda^2)$ of terms that are oscillating along a null direction of the background. The d'Alembertian of these terms is a $O_\ep(\lambda)$.
 The multilinear interactions also give rise to terms of order $O_\ep(\lambda)$ with oscillating factors such as $\si{\f{u_\bA\pm u_\bB}\lambda}$ (with $\bA\neq \bB$){.} Since the $\{u_\bA\}$ are null adapted, these terms oscillate in a non-null direction and therefore behave better. Hence, we treat these terms as part of the error term.
\end{itemize}
\subsubsection{Parametrix for $g_\lambda$}\label{sec.para.g}

For the metric $g_\lambda$, we construct a parametrix for each metric component.  More precisely, we will define
$$N=N_0+N_1+N_2+N_3,\quad \gamma=\gamma_0+\gamma_1+\gamma_2+\gamma_3,\quad \beta^{i}=(\beta^{i})_0+(\beta^{i})_1+(\beta^{i})_2+(\beta^{i})_3.$$
We will use the convention that $\mfg$ denotes one of these metric components. In this notation, the decomposition above reads
\begin{equation}\label{g.para.def}
\mfg=\mfg_0+\mfg_1+\mfg_2+\mfg_3.
\end{equation}
In order to treat these metric components in a unified manner, we consider the Poisson-type equations satisfied by them. We consider the equations \eqref{elliptic.1}-\eqref{elliptic.4} (in the vacuum case, i.e., when $F_{\bA}=0$). Taking the divergence of \eqref{elliptic.4}, using \eqref{elliptic.1} and eliminating $H$ (using \eqref{elliptic.4}), we thus obtain the following system of elliptic equations for each of the metric components

\begin{align}
&\Delta \gamma = -|\nabla \phi|^2-\frac{e^{2\gamma}}{N^2} (e_0 \phi)^2-\frac{e^{2\gamma}}{8N^2}|L\beta|^2,\label{elliptic.g.1}\\
&\Delta N =\f{e^{2\gamma}}{4N}|L\beta|^2+ \frac{2e^{2\gamma}}{N}(e_0 \phi)^2,\label{elliptic.g.2}\\
& \Delta \beta^j={\delta^{ik}}\delta^{j\ell}\rd_k\left(\log(Ne^{-2\gamma})\right)(L\beta)_{i\ell}-4 \delta^{ij}(e_0 \phi)( \partial_i \phi). \label{elliptic.g.3}
\end{align}

We now introduce some notations which will allow us to deal with \eqref{elliptic.g.1}, \eqref{elliptic.g.2} and \eqref{elliptic.g.3} simultaneously.

For $\gamma$, $N$, $\beta^i$, define the following {matrices}. Notice that these correspond to the quadratic terms in $\phi$ in \eqref{elliptic.g.1}-\eqref{elliptic.g.3} {(cf. \eqref{g.elliptic})}:
\begin{align}
{\bf \Gamma}(\gamma)^{\mu\nu}=& 
\left(\begin{array}{ccc}-\f{e^{2\gamma}}{N^2} & \f{e^{2\gamma}}{N^2}\beta^1 & \f{e^{2\gamma}}{N^2}\beta^2\\
\f{e^{2\gamma}}{N^2}\beta^1 & -1-\f{e^{2\gamma}}{N^2}(\beta^1)^2 & -\f{e^{2\gamma}}{N^2}\beta^1 \beta^2\\
\f{e^{2\gamma}}{N^2}\beta^2 & -\f{e^{2\gamma}}{N^2}\beta^1 \beta^2 & -1-\f{e^{2\gamma}}{N^2}(\beta^2)^2
\end{array}
\right),\label{G.g}\\
{\bf \Gamma}(N)^{\mu\nu}=& 
\left(\begin{array}{ccc}-\f{2e^{2\gamma}}{N} & \f{2e^{2\gamma}}{N}\beta^1 & \f{2e^{2\gamma}}{N}\beta^2\\
\f{2e^{2\gamma}}{N}\beta^1 & -\f{2e^{2\gamma}}{N}(\beta^1)^2 & -\f{2e^{2\gamma}}{N}\beta^1 \beta^2\\
\f{2e^{2\gamma}}{N}\beta^2 & -\f{2e^{2\gamma}}{N}\beta^1 \beta^2 & -\f{2e^{2\gamma}}{N}(\beta^2)^2
\end{array}
\right),\label{G.N}\\
{\bf \Gamma}(\beta^1)^{\mu\nu}=& 
\left(\begin{array}{ccc} 0 & -2 & 0\\
-2 & 4\beta^1 & 2\beta^2\\
0 & 2\beta^2 & -0
\end{array}
\right),\label{G.b1}\\
{\bf \Gamma}(\beta^2)^{\mu\nu}=& 
\left(\begin{array}{ccc} 0 & 0 & -2\\
0 & 0 & 2\beta^1\\
-2 & 2\beta^1 & 4\beta^2
\end{array}
\right)\label{G.b2}.
\end{align}

Define also the matrices $\bfG_0(\gamma)$, $\bfG_0(N)$ and $\bfG_0(\beta^i)$ in a similar manner as above except that all of the metric components $\gamma$, $N$ and $\beta^i$ are replaced by their background value $\gamma_0$, $N_0$ and $\beta_0^i$.

Next, define the following functions 
\begin{align}
\Upsilon(\gamma):=&-\frac{e^{2\gamma}}{8N^2}|L\beta|^2,\label{up1}\\
\Upsilon(N):=&\f{e^{2\gamma}}{4N}|L\beta|^2,\label{up2}\\
\Upsilon(\beta^i):=&{\delta^{{j}k}}\delta^{{i}\ell}\left((\rd_k\log N)(L\beta)_{{j}\ell}-2(\rd_k\gamma)(L\beta)_{{j}\ell}\right).\label{up3}
\end{align}
Similarly as above, we use $\Upsilon_0(\gamma)$, $\Upsilon_0(N)$ and $\Upsilon_0(\beta^i)$ to denote the above expressions when all of the metric components $\gamma$, $N$ and $\beta^i$ are replaced by their background value $\gamma_0$, $N_0$ and $\beta_0^i$ (including both the ones that are differentiated and those that are not). 

Therefore, by the above notations, \eqref{elliptic.g.1}, \eqref{elliptic.g.2} and \eqref{elliptic.g.3} now take the form
\begin{equation}\label{g.elliptic}
\Delta \mfg= {\bf \Gamma}(\mfg)^{\mu\nu}\rd_\mu\phi\,\rd_\nu\phi+\Upsilon(\mfg).
\end{equation}
The background metric components, which are denoted by $\mfg_0$, satisfies a similar equation with an extra term
\begin{equation}\label{g0.elliptic}
\Delta \mfg_0= {\bf \Gamma}_0(\mfg)^{\mu\nu}\rd_\mu\phi_0\,\rd_\nu\phi_0+\f 12 \sum_{\bA} F^2_{\bA}{\bf \Gamma}_0(\mfg)^{\mu\nu}(\rd_\mu u_{\bA})(\rd_\nu u_{\bA})+\Upsilon_0(\mfg).
\end{equation}

{\bf Definition of $\mfg_1$}

Define $\mfg_1$ by
\begin{equation}\label{g1.def}
\begin{split}
\mfg_1=&-\f 18{\bf \Gamma}_0(\mfg)^{\mu\nu}\sum_\bA \frac{\lambda^2 F_\bA^2}{|\nabla u_\bA|^2}  (\partial_\mu u_\bA)( \partial_\nu u_\bA) \co{\frac{2u_\bA}{\lambda}}\\
&-2{\bf \Gamma}_0(\mfg)^{\mu\nu}\sum_\bA\frac{\lambda^2 F_\bA}{|\nabla u_\bA|^2} (\partial_\mu \phi_0 )(\partial_\nu u_\bA)  \si{\frac{u_\bA}{\lambda}}\\
&-{\f 12}{\bf \Gamma}_0(\mfg)^{\mu\nu}\sum_{\pm} \sum_{\bA}\sum_{\bB \neq \bA}\frac{(\mp 1)\cdot \lambda^2 F_\bA F_{\gra B}}{|\nabla (u_\bA \pm u_{\gra B})|^2} (\partial_\mu u_\bA)( \partial_\nu u_{\gra B}) \co{\frac{ u_\bA \pm u_{\gra B}}{\lambda}}.
\end{split}
\end{equation}
Recall that the (background) collection $\{u_{\bA}\}_{\bA\in \mathcal A}$ is spatially adapted and therefore the above expression is well-defined and bounded. Notice that the definition of $\mfg_1$ depends \underline{only} on the \underline{background}. One can think of $\mfg_1$ as the part of the parametrix that is constructed to cancel all the\footnote{i.e., these terms may be small in terms of $\ep$ but are not small in terms of $\lambda$.} $O_\ep(1)$ high-frequency terms in $\Delta(\mfg-\mfg_0)$. 

{\bf Definition of $\mfg_2$}

{Define $\mfg_2$ by}
\begin{equation}\label{g2.def}
\begin{split}
\mfg_2:= &-\sum_{\bA}\f{\lambda^3 \mathcal G_{1,1,\bA}(\mfg)}{|\nabla u_\bA|^2}\co{\frac{u_\bA}{ \lambda}}-\sum_{\bA}\f{\lambda^3\mathcal G_{1,2,\bA}(\mfg)}{4|\nab u_\bA|^2} \si{\frac{2u_\bA}{\lambda}}\\
&-\sum_{\bA}\f{\lambda^3\mathcal G_{1,3,\bA}(\mfg)}{9|\nab u_\bA|^2} \co{\frac{3u_\bA}{\lambda}}-\sum_{\pm}\sum_{\bA}\sum_{\bB\neq \bA}\f{\lambda^3 \mathcal G_{2,1,\bA,\bB,\pm}(\mfg)}{|\nab(u_\bA\pm u_{\bB})|^2}\si{\frac{ u_\bA \pm u_{\gra B}}{\lambda}}\\
&-\sum_{\pm}\sum_{\bA}\sum_{\bB\neq \bA}\f{\lambda^3 \mathcal G_{2,2,\bA,\bB,\pm}(\mfg)}{|\nab(u_\bA\pm 2 u_{\bB})|^2}\co{\frac{ u_\bA \pm 2 u_{\gra B}}{\lambda}}\\
&-\sum_{\pm}\sum_{\bA,\bB} \f{\lambda^3 \mathcal G_{2,3,\bA,\bB,\pm}(\mfg)}{|\nab(u_\bA\pm 3u_\bB)|^2}\si{\f{u_\bA\pm 3u_\bB}{\lambda}},
\end{split}
\end{equation}
where for each $\mfg$, the functions $\mathcal G_{1,1,\bA}(\mfg)$, $\mathcal G_{1,2,\bA}(\mfg)$, $\mathcal G_{1,3,\bA}(\mfg)$, $\mathcal G_{2,1,\bA,\bB,\pm}(\mfg)$, $\mathcal G_{2,2,\bA,\bB,\pm}(\mfg)$ and $\mathcal G_{2,3,\bA,\bB,\pm}(\mfg)$ are all compactly supported. All these functions will be defined precisely in Section \ref{sec.g2.def} and can in principle be written down expl{i}citly. At this point, let us only not{e} that these functions \underline{not} only depend on the background solution, but also depends of $\wht F_\bA$ defined earlier in \eqref{whtfi}. One can think of $\mfg_2$ as the part of the parametrix that is constructed to cancel all the $O_\ep(\lambda)$ high-frequency terms in $\Delta(\mfg-\mfg_0-\mfg_1)$.

{\bf Definition of $\mfg_3$}

Finally, of course the term $\mfg_3$ is determined by the evolution equation as well as \eqref{g.para.def}.

\section{The bootstrap assumptions}\label{sec.bootstrap}

In this section, we will describe all the bootstrap assumptions. {Our goal will be to use a bootstrap argument to prove that for $\ep$ and $\lambda$ sufficiently small (with $\lambda\ll \ep$, consistent with Theorem~\ref{main.thm.2}) the solution arising from initial data given in Lemma~\ref{lmini} exists for time $[0,1]$ and that the parametrix we introduced in Section~\ref{sec.parametrix} is indeed a good approximation of the solution. This then also allows us to show convergence as $\lambda \to 0$ (cf. Section~\ref{sec.concl}).} {\bf We will use $C_1$ as a bootstrap constant, which will be determined later and will depend only on $C_0$, which is given by Corollary \ref{lwp.small}}. Our goal will be to improve all the bounds made in the bootstrap assumption. {\bf In order to emphasize which constants depend on $C_1$, we will use $C(C_0)$ to denote a constant depending only on $C_0$ and $C(C_1)$ to denote a constant depending on both $C_0$ and $C_1$. {We will also use $C$ (or $\ls$) for constants which are independent of both $C_0$ and $C_1$.}}

Recall that in Section \ref{sec.parametrix}, we have introduced the decomposition of $\phi_\lambda$ and $g_\lambda$ into various pieces. We will therefore need to obtain estimates for all of them. Notice however that $\phi_0$, $F_\bA$, $\wht F^{(2)}_\bA$, $\wht F^{(3)}_\bA$, $\mfg_0$ and $\mfg_1$ are defined in a way that \underline{only} depend on the \underline{background} solution. We therefore do not need bootstrap assumptions on them. Nevertheless, for $\wht F^{(2)}_\bA$, $\wht F^{(3)}_\bA$ and $\mfg_1$, it is convenient to already state the bounds that they satisfy\footnote{We will however not restate the bounds for $\phi_0$, $F_\bA$ and $\mfg_0$ but will simply refer the readers to {Corollary~\ref{lwp.small}.}}. Notice that all these bounds depend only on $C_0$ and are independent of the bootstrap constant.
\begin{align}
\|\wht F_\bA^{(a)}\|_{H^{9}}+ \|\partial_t \wht F_{\bA}^{(a)}\|_{H^8}+\|\rd_t^2 \wht F_\bA^{(a)}\|_{H^7}\leq C(C_0) \ep^2,\quad a=2,3,\label{B1}\tag{B1}\\
\sum_{k\leq 8}\lambda^{k}\|\mfg_1\|_{H^k\cap {C^{k}}}+\sum_{k\leq 7}\lambda^{k+1}\|\rd_t\mfg_1\|_{H^k\cap {C^{k}}}+\sum_{k\leq 6}\lambda^{k+2}\|\rd_t^2\mfg_1\|_{H^k\cap {C^{k}}}\leq C(C_0)\lambda^2\ep^2.\label{B2}\tag{B2}
\end{align}
\eqref{B1} and \eqref{B2} will be proven in Propositions \ref{whtG.prop} and \ref{mfg1.prop} respectively.

In view of the above discussions, we only need to introduce bootstrap assumptions for $\wht F_\bA$, $\mathcal E_\lambda$, $\mfg_2$ and $\mfg_3$. Introduce the following bootstrap assumptions\footnote{We will later decompose $\mathcal E_\lambda$ into three different pieces (see Section \ref{secelambda} and \eqref{E.decompose}). We will need the bootstrap assumption \eqref{BA2} to hold for each of these pieces. We do not write this out explicitly at this point to simplify notations, but it will be clear in Section \ref{secelambda} that we indeed improve the bootstrap assumptions for each of those decomposed pieces.}:
\begin{align}
&\sum_{k\leq 3} \left(\lambda^k \|\wht F_\bA\|_{H^{2+k}}+\lambda^k\|\partial_t \wht F_\bA\|_{H^{1+k}}+\lambda^{k+1}\|\rd_t^2\wht F_\bA\|_{H^k}\right)\leq C_1\ep,\label{BA1}\tag{BA1}\\
&\sum_{k\leq 3}\lambda^k\|\partial \mathcal E_\lambda\|_{H^k}
+\sum_{k\leq 2}\lambda^{k+1} \|\partial^2 \mathcal E_\lambda\|_{H^k}\leq C_1\ep\lambda^2 ,\label{BA2}\tag{BA2}\\
&\sum_{k \leq 5}
\lambda^k\|\mfg_2\|_{H^k}+\sum_{k\leq 4}\lambda^{k+1} \|\partial_t \mfg_2\|_{H^k}\leq C_1 \ep \lambda^3 .\label{BA3}\tag{BA3}
\end{align}
{F}or $\mfg_3$, each of the metric component can be decomposed as 
$$\mfg_3=(\mfg_3)_{asymp}(t)\chi({|x|})\log({|x|})+\wht{\mfg}_3,$$
where
\begin{align}
|(\mfg_3)_{asymp}|+|\rd_t(\mfg_3)_{asymp}|+\sum_{k\leq 3}\lambda^k\|\wht{\mfg}_3\|_{H^{2+k}_{{\delta}}}+\sum_{k\leq 2}\lambda^{k+1}\|\rd_t\wht{\mfg}_3\|_{H^{{2}+k}_{{\delta}}}\leq C_1\ep \lambda^2{.}\label{BA4}\tag{BA4}
\end{align}
{The existence of such a decomposition is a consequence of the local existence result (Theorem~\ref{lwp}). Notice that it also implies the following conditions for the full metric,}
{$$\mfg = (\mfg_0)_c + ((\mfg_0)_{asymp}(t) + (\mfg_3)_{asymp}(t))\chi({|x|})\log({|x|})+\wht\mfg_0 + \mfg_1 + \mfg_2 +\wht{\mfg}_3$$}
{\begin{equation}\label{g.asymp.sign.1}
(\mfg_0)_c = \begin{cases}
1\quad \mbox{if $\mfg = N$}\\
0\quad \mbox{if $\mfg = \gamma, \bt^i$},
\end{cases}
\end{equation}}
{and
\begin{equation}\label{g.asymp.sign.2}
\mfg_{asymp}(t):= (\mfg_0)_{asymp}(t) + (\mfg_3)_{asymp}(t) \begin{cases}
\geq 0\quad \mbox{if $\mfg = N$}\\
\leq 0\quad \mbox{if $\mfg = \gamma$}\\
= 0\quad \mbox{if $\mfg = \bt^i$},
\end{cases},\quad
(\mfg_0)_{asymp}(t) \begin{cases}
\geq 0\quad \mbox{if $\mfg = N$}\\
\leq 0\quad \mbox{if $\mfg = \gamma$}\\
= 0\quad \mbox{if $\mfg = \bt^i$}.
\end{cases}
\end{equation}
}

{At this point it is useful to note the following support properties, which follows immediately from the estimates for the metric in Corollary~\ref{lwp.small}, \eqref{B2}, \eqref{BA3}, \eqref{BA4} and the support statement in Theorem~\ref{lwp}:}
{
\begin{lm}\label{lm:Rsupp}
For $\ep$ and $\lambda$ sufficiently small, as long as the bootstrap assumptions \eqref{BA3} and \eqref{BA4} hold on $[0,T]$ ($T\leq 1$), there exists $R_{supp}>R$ (independent of $\lambda$ and $T$) such that\footnote{Here, $J_0^+$ is the causal future with respect to the metric $g_0$ and $J_{\lambda}^+$ is the causal future with respect to the metric $g_\lambda$.} 
$$J_0^+(\Sigma_0{\cap} B(0,R)),\, J_\lambda^+(\Sigma_0{\cap} B(0,R))\subset [0,T]\times B(0,R_{supp}).$$
In particular, $\phi_0$ and $F_{\bA}$ associated to the background solution to \eqref{back}, as well as $\phi_\lambda$ associated to the one-parameter family of solutions to \eqref{sys}, are all supported in $B(0,R_{supp})$ for all $t\in [0,T]$.
\end{lm}
}
{In fact, we will choose $A_\bA^{(1)}$, $A_\bA^{(2)}$ and $A_\bA^{(3)}$ (cf. \eqref{whtfi}, \eqref{whtgi2}, \eqref{whtgi3}) and $\mathcal G_{1,1,\bA}(\mfg)$, $\mathcal G_{1,2,\bA}(\mfg)$, $\mathcal G_{1,3,\bA}(\mfg)$, $\mathcal G_{2,1,\bA,\bB,\pm}(\mfg)$, $\mathcal G_{2,2,\bA,\bB,\pm}(\mfg)$ and $\mathcal G_{2,3,\bA,\bB,\pm}(\mfg)$ (cf. \eqref{g2.def}) so that they are all supported in $J^+_0(\Sigma_0{\cap} B(0,R))$. Hence, by \eqref{phi.para}, \eqref{whtfi}, \eqref{whtgi2}, \eqref{whtgi3}, \eqref{g1.def}, \eqref{g2.def}, 
\begin{equation}\label{eq:supp}
supp(F_{\bA}),\, supp(\wht F_{\bA}),\, supp(\wht F^{(2)}_{\bA}),\, supp(\wht F^{(3)}_{\bA}),\, supp(\mathcal E_{\lambda}),\, supp(\mfg_1),\, supp(\mfg_2)\subset B(0,R_{supp}),
\end{equation}
as long as the bootstrap assumption holds.
}

{Finally, we need to introduce one more bootstrap assumption.} Notice that the above bootstrap assumptions distinguish between the spatial $\nab$ and the $\rd_t$ derivative. On the other hand, it will be crucial to our argument that $\rd_t\gamma_3$ is in fact better that the $\rd_t$ derivative of a general $\mfg_3$ component. Namely, we make the following bootstrap assumption for $\rd_t \wht \gamma_3$ on the compact set $B(0,R_{supp}{+1})$ {(where $R_{supp}$ is as in Lemma~\ref{lm:Rsupp})}:
\begin{align}
\|\partial_t \wht \gamma_3\|_{L^2(B(0,R_{supp}{+1}))}
\leq C_1\ep\lambda^2.\label{BA5}\tag{BA5}
\end{align}
Our goal will be to improve all the constants in \eqref{BA1}, \eqref{BA2}, \eqref{BA3}, \eqref{BA4} and \eqref{BA5}{. The proof of 
 which will occupy the rest of the paper. More precisely, we will prove}
\begin{thm}[Bootstrap theorem]\label{thm:BS}
{Let $\ep$ and $\lambda$ be sufficiently small (with $\lambda\ll \ep$) and $C_1$ be sufficiently large (depending on $C_0$ but independent of $\ep$ and $\lambda$). Suppose that the bootstrap assumptions \eqref{BA1}-\eqref{BA5} hold for $t\in [0,T]$ for the unique solution arising from the initial data given in Lemma~\ref{lmini}. Then in fact all of \eqref{BA1}-\eqref{BA5} hold for $t\in [0,T]$ with the constant $C_1$ replaced by $\f{C_1}{2}$.}
\end{thm}
{In view of Theorem~\ref{thm:BS},} {\bf we will assume below that \eqref{BA1}-\eqref{BA5} hold and that we have the following hierarchy of constants:
$$\lambda\ll \ep\ll C_0\ll C_1. $$}
In particular, we will freely take $C(C_1)\ep<10^{-2}$, $C(C_1)\lambda<10^{-2}$ for constants $C(C_1)$ depending on $C_1$.

{Before we end this section, let us note that our bootstrap assumptions imply immediate the following pointwise estimates via the Sobolev embbeding theorems in Appendix~\ref{weightedsobolev}:}

\begin{prp}\label{prp:Linfty.BA}
Under the bootstrap assumptions, the following estimates hold (for some $C(C_0)$ depending on $C_0$ and some universal constant $C$):
\begin{align}
\sum_{k\leq 3}\lambda^k\|\wht F_{\bA}\|_{C^k} + \sum_{k\leq 2} \lambda^{k+\f 12}\|\rd_t \wht F_{\bA}\|_{C^k} + \sum_{k\leq 1}\lambda^{k+2}\|\partial^2_t\wht F_\bA\|_{C^k} \leq &C C_1 \ep, \label{F.Linfty}\\
\left(\|\wht F^{(a)}_\bA\|_{H^7} + \|\rd_t\wht F^{(a)}_\bA\|_{H^6} + \|\rd_t^2\wht F^{(a)}_\bA\|_{H^5}\right)\leq &C(C_0)\ep^2,\quad a=2,3, \label{Fa.Linfty}\\
\sum_{k\leq 1} \lambda^k\|\rd \mathcal E_\lambda\|_{C^k} + \lambda\|\rd^2 \mathcal E_{\lambda}\|_{C^0}\leq &CC_1\ep \lambda,\label{E.Linfty}\\
\sum_{k\leq 3}\lambda^k\|\mfg_2\|_{C^k} + \sum_{k\leq 2} \lambda^{k+1}\|\partial_t \mfg_2\|_{C^k}\leq &C C_1\ep \lambda^2, \label{g2.Linfty}\\
\sum_{k\leq 3}\lambda^k\|\wht\mfg_3\|_{C^k_{\de+1}}+ \sum_{k\leq 2} \lambda^{k+1}\|\partial_t \wht\mfg_3\|_{C^k_{\de+1}}\leq &C C_1\ep \lambda^2.\label{g3.Linfty}
\end{align}
\end{prp}

\begin{proof}

Estimates for the last two terms in \eqref{F.Linfty}, \eqref{E.Linfty} and \eqref{g2.Linfty} follow from \eqref{BA1}, \eqref{BA2}, \eqref{BA3} and Proposition~\ref{scaled.Sobolev}; while estimates for the first term in \eqref{F.Linfty}, as well as \eqref{Fa.Linfty} and \eqref{g3.Linfty}{,} follo{w} from {\eqref{B1}, \eqref{BA1},} \eqref{BA4} and Proposition~\ref{holder} {(in the particular case $H^{2+k}\subset C^k$)}.

\end{proof}

\section{Estimates for the scalar field}\label{sec.scalar}

\subsection{Decomposition of the inverse metric}

For the purpose of controlling the scalar field, we need to decompose $g^{-1}$. Our strategy is to write it as a sum of the following four terms
\begin{equation}\label{g.inverse.decomp}
g^{-1}= (g_0)^{-1}+g_1' +g_3' +\mbox{Error},
\end{equation}
where $(g_0)^{-1}$ is the inverse of the background metric and for $a=1,3$, $g_a'$ is defined as in \eqref{gp.def}. The ``Error'' term is then defined so that \eqref{g.inverse.decomp} holds. 

In the following propositions, we will give precise estimates for $g_1'$, $g_3'$ and the ``Error'' term. Let us note that we will only need to control the inverse metric when it is multiplied with (the decomposed pieces of) the scalar field. Therefore, it suffices to obtain estimates on the \underline{compact} set $B(0,R_{supp})$ and we need not be concerned with the growth near infinity\footnote{which, of course, is important for the estimates of the metric itself!}.

\begin{proposition}\label{g1p.prop}
$g_1'$ depends only on the background and can be decomposed as follows:
\begin{equation*}
\begin{split}
(g_1')^{\mu\nu}=&\sum_\bA \lambda^2 (  G_{1,1,\bA})^{\mu\nu}\si{\f{u_\bA}{\lambda}}+\sum_{\bA} \lambda^2 (G_{1,2,\bA})^{\mu\nu} \co{\f{2u_\bA}{\lambda}}\\
&+\sum_{\pm}\sum_{\bA}\sum_{\bB\neq \bA} \lambda^2 (G_{1,bil,\bA,\bB,\pm})^{\mu\nu}\co{\f{u_\bA\pm u_\bB}{\lambda}},
\end{split}
\end{equation*}
where each of $G_{1,1,\bA}$, $G_{1,2,\bA}$ and $G_{1,bil,\bA,\bB,\pm}$ is compactly supported in $B(0,R_{supp})$ and they satisfy the estimate
\begin{equation*}
\begin{split}
\|(G_{1,1,\bA},&G_{1,2,\bA},G_{1,bil,\bA,\bB,\pm})\|_{((H^8\cap {C^{8}})(B(0,R_{supp})))^3}\\
&+\|(\rd_tG_{1,1,\bA},\rd_tG_{1,2,\bA},\rd_tG_{1,bil,\bA,\bB,\pm})\|_{((H^7\cap {C^{7}})(B(0,R_{supp})))^3}\leq C(C_0)\ep^2.
\end{split}
\end{equation*}
In particular, we have
$$\sum_{k\leq 8}\lambda^k\|g_1' \|_{(H^k\cap {C^{k}})(B(0,R_{supp}))}+\sum_{k\leq 7}\lambda^{k+1}\|\rd_t g_1' \|_{(H^k\cap {C^{k}})(B(0,R_{supp}))}\leq C(C_0)\ep^2\lambda^2.$$
\end{proposition}
\begin{proof}
This is an immediate consequence of the definitions \eqref{gp.def}, \eqref{g1.def} and the estimates on the background solution (see Corollary \ref{lwp.small}).
\end{proof}

\begin{proposition}\label{g3p.prop}
On the compact set $B(0,R_{supp}{+1})$, $g_3'$ satisfies {the following $H^k$ estimates}
$$\sum_{k\leq 3}\lambda^k\|g_3'\|_{H^{2+k}(B(0,R_{supp}{+1}))}+\sum_{k\leq 2}\lambda^{k+1}\|\rd_t g_3'\|_{H^{1+k}(B(0,R_{supp}{+1}))}\leq C(C_1)\ep \lambda^2{,}$$
{and the following $C^k$ estimates}
{$$\sum_{k\leq 3} \lambda^k \|g_3'\|_{C^k(B(0,R_{supp}+1))} + \sum_{k\leq 2} \lambda^{k+1} \|\rd_t g_3'\|_{C^k(B(0,R_{supp}+1))}\leq C(C_1)\ep\lambda^2. $$
}
\end{proposition}
\begin{proof}
This is an immediate consequence of the definition of $g_3'$ in \eqref{gp.def} together with the bootstrap assumption \eqref{BA4} {and the estimate \eqref{g3.Linfty}}.
\end{proof}

\begin{proposition}\label{gp.error.prop}
On the compact set $B(0,R_{supp}{+1})$, the error term in \eqref{g.inverse.decomp} can be controlled {in $H^k$} as follows:
\begin{equation*}
\begin{split}
\sum_{k\leq 5}\lambda^k &\|g^{-1}- (g_0)^{-1} - g_1' - g_3'\|_{H^k(B(0,R_{supp}{+1}))}\\
&+\sum_{k\leq 4}\lambda^{k+1} \|\rd_t(g^{-1}- (g_0)^{-1} - g_1' - g_3')\|_{H^k(B(0,R_{supp}{+1}))}\leq C(C_1) \ep \lambda^3,
\end{split}
\end{equation*}
{and in $C^k$ as follows:}
\begin{equation*}
\begin{split}
 {\sum_{k\leq 3}\lambda^k} &\|g^{-1}- (g_0)^{-1} - g_1' - g_3'\|_{{C^k(B(0,R_{supp}+1))}}\\
&+ {\sum_{k\leq 2}}\lambda^{{k+1}}\|\partial(g^{-1}- (g_0)^{-1} - g_1' - g_3')\|_{{C^k(B(0,R_{supp}+1))}}\leq C(C_1) \ep \lambda^2. 
\end{split}
\end{equation*}

\end{proposition}
\begin{proof}
$g_1'$ and $g_3'$ are defined so that $g^{-1}- (g_0)^{-1} - g_1' - g_3'$ can be expanded into terms which are either linear in $\mfg_2$ or at least quadratic in $\mfg_1$, $\mfg_2$ and $\mfg_3$. It is then easy to check that all these terms obey the desired bounds thanks to \eqref{B2}, \eqref{BA3}, {\eqref{BA4}, \eqref{g2.Linfty} and \eqref{g3.Linfty}}.
\end{proof}

As a consequence of the decomposition \eqref{g.inverse.decomp}, we also obtain the following estimate for $\Box_g-\Box_{g_0}$:
\begin{proposition}\label{box.diff.prop}
Let $f$ be compactly supported on $B(0,R_{supp})$. Then, for $k\leq 4$,
\begin{equation}\label{box.diff}
\begin{split}
&\left\|(\Box_g-\Box_{g_0})f-\left(\rd_\alpha (g'_1)^{\alpha\beta}+(g_0^{-1})^{\alpha\beta}(2\rd_\alpha\gamma_1+\f{\rd_\alpha N_1}{N})+\rd_t (g'_3)^{t\beta}+(g_0^{-1})^{t\beta}\f{\rd_t N_3}{N_0}\right)(\rd_\beta f)\right\|_{H^k}\\
\leq &C(C_1)\ep\left(\lambda\|\partial f\|_{H^{k}}+\lambda^2\|\partial^2 f\|_{H^{{k}}}+ \lambda^{2-k}\|\partial f\|_{L^\infty}+ \lambda^{2-k}\|\partial^2 f\|_{L^\infty}\right)
.\end{split}
\end{equation}
\end{proposition}
\begin{proof}
The difference of the wave operators takes the form
\begin{equation*}
\begin{split}
(\Box_g-\Box_{g_0})f=&\underbrace{(\gi^{\alpha\beta}-(g^{-1}_0)^{\alpha\beta})\rd^2_{\alpha\beta}f}_{=:Error_1}+\underbrace{\left(\rd_\alpha(\gi^{\alpha\beta}-(g^{-1}_0)^{\alpha\beta})\right)\rd_\beta f}_{=:I}\\
&+\underbrace{\f 12 \left(\gi^{\alpha\beta}\rd_\alpha\log |\det g|-(g_0^{-1})^{\alpha\beta}\rd_\alpha\log |\det g_0|\right)\rd_\beta f}_{=:II}.
\end{split}
\end{equation*}
First, notice that $Error_1$ is an acceptable error\footnote{in the sense that it can be bounded by the {RHS} of \eqref{box.diff}. We will also use the same language below without further comment.} by Propositions \ref{g1p.prop}, \ref{g3p.prop}{,} \ref{gp.error.prop} {and \ref{product.local}:}
\begin{equation*}
\begin{split}
& \|(\gi^{\alpha\beta}-(g^{-1}_0)^{\alpha\beta})\rd^2_{\alpha\beta}f\|_{H^k}\\
\lesssim & \|\gi - g^{-1}_0\|_{L^\infty(B(0,R_{supp}+1))}\|\partial^2 f\|_{H^k}+ \|\partial^2 f\|_{L^\infty}
\|\gi - g^{-1}_0\|_{H^k(B(0,R_{supp}+1))}\\
\leq & C(C_1) \ep\lambda^2 \|\partial^2 f\|_{H^k}+C(C_1)\ep\lambda^{2-k} \|\partial^2 f\|_{L^\infty}.
\end{split}
\end{equation*}
Next, we consider the term $I$. By Proposition \ref{gp.error.prop}, it suffices to consider the contributions from $g_1'$ and $g_3'$ as the remainder is an acceptable error. The contribution from $g_1'$ is precisely $(\rd_\alpha (g'_1)^{\alpha\beta})(\rd_\beta f)$. For the contribution from $g_3'$, we have $(\rd_\alpha (g'_3)^{\alpha\beta})(\rd_\beta f)$. Then notice that by Proposition \ref{g3p.prop}, if $\alpha=1$ or $\alpha=2$, this term is an acceptable error. 
More precisely, we estimate using Proposition~\ref{product.local}
\begin{equation}\label{nablag3}
\begin{split}
\|\nabla g'_3\partial f\|_{H^k}
\lesssim & \|\nabla g'_3\|_{L^\infty{(B(0,R_{supp}+1)}}\|\partial f\|_{H^k}+ \|\partial f\|_{L^\infty}
\|\nabla g'_3\|_{H^k{(B(0,R_{supp}+1)}}\\
\leq & {C(C_1)}\ep\lambda \|\partial f\|_{H^k} + {C(C_1)}\ep\lambda^{2-k} \|\partial f\|_{L^\infty},
\end{split}
\end{equation}
where we have used {\eqref{BA4} and \eqref{g3.Linfty}.}
We are thus left with $(\rd_t (g'_3)^{t\beta})(\rd_\beta f)$, which is one of the main terms.

Finally, we compute the term $II$ using \eqref{g.det} and suppress all the terms which according to Propositions \ref{g1p.prop}, \ref{g3p.prop} and \ref{gp.error.prop} can be treated as acceptable errors:
\begin{equation*}
\begin{split}
&\f 12 \left(\gi^{\alpha\beta}\rd_\alpha\log |\det g|-(g_0^{-1})^{\alpha\beta}\rd_\alpha\log |\det g_0|\right)\\
=& \gi^{\alpha\beta}\rd_\alpha (2\gamma+\log N)-(g_0^{-1})^{\alpha\beta}\rd_\alpha (2\gamma_0+\log N_0)\\
=& (g_0^{-1})^{\alpha\beta}\rd_\alpha(2\gamma_1+2 \gamma_3+\log(1+\f{N_1+N_3}{N}))+\left((g_1')^{\alpha\beta}+(g_3')^{\alpha\beta}\right)\rd_\alpha (2\gamma_0+\log N_0)+\dots\\
=&(g_0^{-1})^{\alpha\beta}\rd_\alpha(2\gamma_1+2 \gamma_3)+(g_0^{-1})^{\alpha\beta}\f{\rd_\alpha(N_1+N_3)}{N+N_1+N_3}+\dots\\
=&(g_0^{-1})^{\alpha\beta}(2\rd_\alpha\gamma_1+\f{\rd_\alpha N_1}{N})+(g_0^{-1})^{t\beta}\f{\rd_t N_3}{N_0}+\dots,
\end{split}
\end{equation*}
where in the last line we have used the fact that (according to \eqref{BA4}) we only need to treat the $\rd_t$ derivative of $\mfg_3$ and moreover that $\rd_t\gamma_3$ satisfies better bounds according to \eqref{BA5}. More precisely{,} we {estimate using Proposition~\ref{product.local}} 
\begin{equation}\label{rdtgammardf}
\|\partial_t \gamma_3 \partial f\|_{H^k}
\lesssim \|\partial_t  \gamma_3\|_{L^\infty}\|\partial f\|_{H^k}
+\|\partial f \|_{L^\infty}\|\partial_t \gamma_3\|_{H^k}
{\leq} {C(C_1)}\ep \lambda \|\partial f\|_{H^k}
+ {C(C_1)} \ep\lambda^{2-k}\|\partial f\|_{L^\infty},
\end{equation}
where {for $\|\partial_t  \gamma_3\|_{L^\infty}$, we have used \eqref{g3.Linfty}; and for $\|\partial_t  \gamma_3\|_{H^k}$, we have used} \eqref{BA5} for $k=0$ and \eqref{BA4} for $k\geq 1$.
\end{proof}

{We also have another variant of} {Proposition~\ref{box.diff.prop}, which requires more integrability for derivatives of $f$:}

\begin{prp}
\label{box.diff.bis.prop}
Let $f$ be compactly supported on $B(0,R_{supp})$. Then, for $k\leq 4$,
\begin{equation}\label{box.diff.bis}
\begin{split}
&\left\|(\Box_g-\Box_{g_0})f-\left(\rd_\alpha (g'_1)^{\alpha\beta}+(g_0^{-1})^{\alpha\beta}(2\rd_\alpha\gamma_1+\f{\rd_\alpha N_1}{N})+\rd_t (g'_3)^{t\beta}+(g_0^{-1})^{t\beta}\f{\rd_t N_3}{N_0}\right)(\rd_\beta f) \right\|_{H^k}\\
\leq &C(C_1)\ep
\sum_{{\ell} \leq k} \lambda^{2-{\ell}}\left(\|\partial^2 f\|_{{C^{k-{\ell}}}}{+ \|\partial f\|_{C^{k-{\ell}}}}\right).
\end{split}
\end{equation}
\end{prp}
\begin{proof}

Since $supp(f)\subset B(0,R_{supp})$, $\|\rd f\|_{H^k}\ls \|\rd f\|_{C^{k}}$ and $\|\rd^2 f\|_{H^k}\ls \|\rd^2 f\|_{C^{k}}$. We can therefore proceed in the same way {as} in the proof of {Proposition~\ref{box.diff.prop}}, except that we need to improve the estimate \eqref{nablag3} and \eqref{rdtgammardf} in terms of the power of $\lambda$.
 
{For \eqref{nablag3}, we have 
$$
\|\nabla g'_3\partial f\|_{H^k}
\lesssim \sum_{\ell\leq k} \|\nabla g'_3\|_{H^{\ell}}\|\partial f\|_{C^{k-\ell}}\leq  C(C_1)\ep\sum_{\ell\leq k}\lambda^{2-\ell} \|\partial f\|_{C^{k-\ell}},
$$
where we have used \eqref{BA4}.} For {\eqref{rdtgammardf},} we have
$$\|\partial_t \gamma_3 \partial f\|_{H^k}
\lesssim \sum_{\ell\leq k} \|\partial_t  \gamma_3\|_{H^{\ell}}\|\partial f\|_{C^{k-\ell}}
\leq C(C_1)\ep \sum_{\ell\leq k} \lambda^{2-\ell} \|\partial f\|_{C^{k-\ell}},$$
where we used \eqref{BA4} for $\ell\geq 1$ and \eqref{BA5} for $\ell=0$.
\end{proof}

\subsection{Computations for the main terms (and definitions of $A^{(1)}_\bA$, $A^{(2)}_\bA$ , $A^{(3)}_\bA$)}\label{sec.sf.main}
In this subsection, we compute the $\Box_g$ of the main part of $\phi_\lambda$, i.e., of $\phi_\lambda-\mathcal E_\lambda$. We will show that it consists either of (1) terms of size $C(C_1)\ep\lambda^2$ or (2) terms of size $C(C_0)\ep^2\lambda$ but oscillating in a non-null direction or (3) terms of size $C(C_1)\ep\lambda$ and roughly of the form $\rd_t \mfg_3$ multiplied by a regular function of size $C(C_0)\ep$. The main result of the section is given in Proposition \ref{main.term}.

In the process of obtaining Proposition \ref{main.term}, we will define $A^{(1)}_\bA$, $A^{(2)}_\bA$ , $A^{(3)}_\bA$, which appeared in the definitions \eqref{whtfi}, \eqref{whtgi2} and \eqref{whtgi3} (see Definition \ref{A.def}). As we will see, they are defined so that certain cancellations take place.

\begin{proposition}\label{box.phi0.prop}
There exist functions $A^{(1)}_\bA[\phi_0]$, $A^{(2)}_\bA[\phi_0]$ and $B^{(2,\pm)}_{\bA\bB}[\phi_0]$, each of which depends only on the background solution, is compactly supported in $B(0,R_{supp})$ and obeys the bound
$$\|\cdot\|_{H^8\cap {C^{8}}}+\|\rd_t(\cdot)\|_{H^7\cap {C^{7}}}+\|\rd_t^2(\cdot)\|_{H^6\cap {C^{6}}}\leq C(C_0)\ep^3$$
such that 
\begin{equation*}
\begin{split}
\Box_g\phi_0=&\left(\rd_t (g'_3)^{t\beta}+(g_0^{-1})^{t\beta}\f{\rd_t N_3}{N_0}\right)(\rd_\beta \phi_0)+\lambda\sum_{\bA}\left(A^{(1)}_\bA[\phi_0] \co{\f{u_\bA}{\lambda}}+A^{(2)}_\bA[\phi_0]\si{\f{2u_\bA}{\lambda}}\right)\\
&+\lambda\sum_{\pm}\sum_{\bA}\sum_{\bB\neq \bA} B^{(2,\pm)}_{\bA\bB}[\phi_0] \si{\f{u_\bA\pm u_\bB}{\lambda}}+R[\phi_0],
\end{split}
\end{equation*}
where 
$R[\phi_0]$ is compactly supported in $B(0,R_{supp})$ and {satisfies the following estimate}:
$$\sum_{k\leq 3}\lambda^k\|R[\phi_0]\|_{H^k}\leq C(C_1)\ep \lambda^2.$$
\end{proposition}
\begin{proof}
Since $\Box_{g_0}\phi_0=0$, it suffices to compute $(\Box_g-\Box_{g_0})\phi_0$. Using {Corollary~\ref{lwp.small} and} Proposition~\ref{box.diff.bis.prop},
\begin{equation}\label{box.phi0.1}
\begin{split}
&\Box_g\phi_0=(\Box_g-\Box_{g_0})\phi_0\\
=&\left(\rd_\alpha (g'_1)^{\alpha\beta}+(g_0^{-1})^{\alpha\beta}(2\rd_\alpha\gamma_1+\f{\rd_\alpha N_1}{N})+\rd_t (g'_3)^{t\beta}+(g_0^{-1})^{t\beta}\f{\rd_t N_3}{N_0}\right)(\rd_\beta \phi_0)+\dots,
\end{split}
\end{equation}
where $\dots$ denotes terms that in the $\sum_{k\leq 3}\lambda^k\|\cdot\|_{H^k}$ norm is bounded by $C(C_1)\ep \lambda^2$. These terms are then treated as $R[\phi_0]$.

Notice that in \eqref{box.phi0.1}, the terms involving $\rd_t g'_3$ and $\rd_t N_3$ are exactly of the form in the statement of the proposition. It thus remains to consider the terms that involve the derivative of $g_1'$, $\gamma_1$ and $N_1$. Now by \eqref{g1.def}, \eqref{B2} and Proposition \ref{g1p.prop}, $g_1'$, $\gamma_1$ and $N_1$ can be written as sum of terms with oscillating factors $\si{\f{u_\bA}{\lambda}}$, $\co{\f{2u_\bA}{\lambda}}$ and $\co{\f{u_\bA\pm u_\bB}{\lambda}}$ (with $\bA\neq \bB$) with coefficients which in the $\|\cdot\|_{H^8\cap {C^{8}}}+\|\rd_t(\cdot)\|_{H^7\cap {C^{7}}}$ norm are of size $C(C_0)\ep^2\lambda^2$. Therefore, their derivatives, when multiplied by the derivatives of $\phi_0$, takes the form 
$$\lambda\sum_{\bA}\left(A^{(1)}_\bA[\phi_0] \co{\f{u_\bA}{\lambda}}+A^{(2)}_\bA[\phi_0]\si{\f{{2u_\bA}}{\lambda}}\right)+\lambda\sum_{\pm}\sum_{\bA}\sum_{\bB\neq \bA} B^{(2,\pm)}_{\bA\bB}[\phi_0] \si{\f{u_\bA\pm u_\bB}{\lambda}}+\dots$$
This concludes the proof of the proposition.
\end{proof}

\begin{proposition}\label{box.F.prop} There exist functions $A^{(1)}_\bA[F]$, $A^{(2)}_\bA[F]$ , $A^{(3)}_\bA[F]$, $A^{(cross)}_{\bA\bB}[F]$, $B^{(2,\pm)}_{\bA\bB}[F]$, $B^{(3,\pm)}_{\bA\bB}[F]$ and $D_{\bA\bB{\bf C}}^{(\pm_1,\pm_2)}[F]$, each of which depends only on the background solution, is compactly supported in $B(0,R_{supp})$ and obeys the bound
$$\|\cdot\|_{H^8\cap {C^{8}}}+\|\rd_t(\cdot)\|_{H^7\cap {C^{7}}}+\|\rd_t^2(\cdot)\|_{H^6\cap {C^{6}}}\leq C(C_0)\ep^3$$
such that for every $\bA\in \mathcal A$,
\begin{equation*}
\begin{split}
&\Box_g \left(\lambda F_\bA \co{\frac{u_\bA}{\lambda}}\right)\\
=&\left(\lambda\Box_{g_0} F_\bA-\frac{1}{\lambda} (g_3')^{\alpha\beta}\partial_\alpha u_\bA \partial_\beta u_\bA F_\bA \right)\co{\frac{u_\bA}{\lambda}}-\left(\rd_t (g'_3)^{t\beta}+(g_0^{-1})^{t\beta}\f{\rd_t N_3}{N_0}\right)(\rd_\beta u_\bA) F_\bA \si{\frac{u_\bA}{\lambda}}\\
&+\lambda A^{(1)}_\bA[F] \co{\f{u_\bA}{\lambda}}+\lambda A^{(2)}_\bA[F] \si{\f{2 u_\bA}{\lambda}}+\lambda A^{(3)}_\bA[F] \co{\f{3 u_\bA}{\lambda}}+\sum_{\bB\neq \bA} \lambda A^{(cross)}_{\bA\bB}[F] \co{\f{u_\bB}{\lambda}}\\
&+\sum_{\pm}\sum_{\bB\neq \bA} \lambda B^{(2,\pm)}_{\bA\bB}[F]\si{\f{u_\bA\pm u_\bB}{\lambda}}+\sum_{\pm}\sum_{\bB}  \lambda B^{(3,\pm)}_{\bA\bB}[F]\co{\f{u_\bA\pm 2u_\bB}{\lambda}}\\
&+\sum_{\pm_1}\sum_{\pm_2}\sum_{\bB\neq \bA}\sum_{\substack{{\bf C}\neq \bB\\{\bf C}\neq \bA}}\lambda D_{\bA\bB{\bf C}}^{(\pm_1,\pm_2)}[F]\co{\f{u_\bA\pm_1 u_\bB\pm_2 u_{\bf C}}{\lambda}}+R_\bA[F],
\end{split}
\end{equation*}
where $R_\bA[F]$ is compactly supported in $B(0,R_{supp})$ and is small in the following sense:
$$\sum_{k\leq 3}\lambda^k\|R_\bA[F]\|_{H^k}\leq C(C_1)\ep \lambda^2.$$
\end{proposition}
\begin{proof}
Using the expansion \eqref{g.inverse.decomp}, the estimate in Proposition \ref{gp.error.prop} and the fact that $(g_0^{-1})^{\alpha\beta}\partial_\alpha u_\bA \partial_\beta u_\bA=0$ and $(\Box_{{g_0}} u_\bA ) F_\bA + 2 (g_{{0}}^{-1})^{\alpha \beta}\partial_\alpha u_\bA \partial_\beta F_\bA=0$ (recall that $u_\bA$ and $F_\bA$ satisfy \eqref{back}), we get
\begin{align*}
&\Box_g \left(\lambda F_\bA \co{\frac{u_\bA}{\lambda}}\right)\\
=& \lambda (\Box_g F_\bA) \co{\frac{u_\bA}{\lambda}}
-2\gi^{\alpha \beta}\partial_\alpha u_\bA \partial_\beta F_\bA \si{\frac{u_\bA}{\lambda}}\\
 &-(\Box_g u_\bA) F_\bA \si{\frac{u_\bA}{\lambda}}
-\frac{1}{\lambda} \gi^{\alpha \beta}\partial_\alpha u_\bA \partial_\beta u_\bA F_\bA \co{\frac{u_\bA}{\lambda}}\\
=& \underbrace{\lambda (\Box_{g_0} F_\bA) \co{\frac{u_\bA}{\lambda}}}_{=:\mathfrak F_1} \underbrace{-\frac{1}{\lambda} (g_1')^{\alpha\beta}\partial_\alpha u_\bA \partial_\beta u_\bA F_\bA \co{\frac{u_\bA}{\lambda}}}_{=:\mathfrak F_2} \underbrace{ -{\f{1}{\lambda}} (g_3')^{\alpha\beta}\partial_\alpha u_\bA \partial_\beta u_\bA F_\bA \co{\frac{u_\bA}{\lambda}}}_{=:\mathfrak F_3}\\
&\underbrace{-((\Box_{g}-\Box_{g_0}) u_\bA ) F_\bA \si{\frac{u_\bA}{\lambda}}}_{=:\mathfrak F_4}+\dots,
\end{align*}
where we have used {the} convention {(as} in the proof of Proposition \ref{box.phi0.prop}{) that} $\dots$ denotes terms which are compactly supported in $B(0,R_{supp})$ that by Propositions \ref{gp.error.prop} and \ref{box.diff.prop} are bounded in the $\sum_{k\leq 3} \lambda^k\|\cdot\|_{H^k}$ norm by $ C(C_1)\ep\lambda^2$. {These terms can then be grouped into $R_{\bA}[F]$.}

This computation already shows that the potentially most harmful term of order $O(\f{\ep}{\lambda})$ is not present since $u_\bA$ is an eikonal function of the background solution. Also, an $O(\ep)$ term is not present thanks {to} $(\Box_{{g_0}} u_\bA ) F_\bA+2 (g_{{0}}^{-1})^{\alpha \beta}\partial_\alpha u_\bA \partial_\beta F_\bA=0$. The remaining terms {are} of size at most $O(\ep\lambda)$. This smallness will however not be sufficient to close the argument and we will need to analyze all the $O(\ep\lambda)$ terms, only allowing $O({\ep}\lambda^2)$ terms to be treated as error terms.

{\bf The terms $\mathfrak F_1$ and $\mathfrak F_3$.} First, there are the two main terms $\mathfrak F_1$ and $\mathfrak F_3$, which are simply included into the main terms in the statement of the proposition.

{\bf The term $\mathfrak F_2$.} Using Proposition \ref{g1p.prop}, we expand $\mathfrak F_2$ as follows\footnote{Note that we have relabeled the indices in Proposition \ref{g1p.prop}.}:
\begin{equation*}
\begin{split}
\mathfrak F_2=&-\lambda \sum_{\bB}(G_{1,1,\bB})^{\alpha\beta} \partial_\alpha u_\bA \partial_\beta u_\bA F_\bA \co{\frac{u_\bA}{\lambda}}\si{\f{u_\bB}{\lambda}}\\
&-\lambda \sum_{\bB} (G_{1,2,\bB})^{\alpha\beta}\partial_\alpha u_\bA \partial_\beta u_\bA F_\bA \co{\frac{u_\bA}{\lambda}}\co{\f{2u_\bB}{\lambda}}\\
&-\sum_{\pm}\sum_{\bB}\sum_{{\bf C}\neq \bB}\lambda (G_{1,bil,\bB,{\bf C},\pm})^{\alpha\beta}\partial_\alpha u_\bA \partial_\beta u_\bA F_\bA \co{\frac{u_\bA}{\lambda}}\co{\f{u_\bB\pm u_\bC}{\lambda}} {.}
\end{split}
\end{equation*}
We can expand this using standard trigonometric identities. First, notice that there are \underline{no} low frequency terms. We then separate terms which according to whether $\bB\neq \bA$ (and $\bB\neq {\bf C}$). It is easy to check that each of the terms takes the form of one of the terms in the statement of the proposition.

{\bf The term $\mathfrak F_4$.} Finally, for the term $\mathfrak F_4$, we apply Proposition \ref{box.diff.bis.prop} to get
\begin{equation*}
\begin{split}
\mathfrak F_4=-\left(\rd_\alpha (g'_1)^{\alpha\beta}+(g_0^{-1})^{\alpha\beta}(2\rd_\alpha\gamma_1+\f{\rd_\alpha N_1}{N})+\rd_t (g'_3)^{t\beta}+(g_0^{-1})^{t\beta}\f{\rd_t N_3}{N_0}\right)(\rd_\beta u_\bA)F_\bA \si{\frac{u_\bA}{\lambda}}+{\dots}
\end{split}
\end{equation*}
The term $-\left(\rd_t (g'_3)^{t\beta}+(g_0^{-1})^{t\beta}\f{\rd_t N_3}{N_0}\right)(\rd_\beta u_\bA)F_\bA \si{\frac{u_\bA}{\lambda}}$ is one of the main terms in the statement of the proposition. The remaining terms involving the derivatives of $g_1'$, $\gamma_1$ and $N_1$ can, according to \eqref{g1.def} and Proposition \ref{g1p.prop}, be written as a sum of terms with oscillating factors $ \si{\frac{u_\bA}{\lambda}}\co{\f{u_\bB}{\lambda}}$, $\si{\frac{u_\bA}{\lambda}}\si{\f{2u_\bB}{\lambda}}$ and $\si{\frac{u_\bA}{\lambda}}\si{\f{u_\bB\pm u_{\bf C}}{\lambda}}$ (with $ \bB\neq {\bf C}$) with coefficients which in the $\|\cdot\|_{H^8\cap {C^{8}}}+\|\rd_t(\cdot)\|_{H^7\cap {C^{7}}}$ norm are of size $C(C_0)\ep^3\lambda$. Expanding again with standard trigonometric identities\footnote{Let us note that compared to the terms in $\mathfrak F_2$, the roles of all the $\sin$ and $\cos$ are flipped and therefore the products of these trigonometric functions still take the same form.}, it is easy to check that each of the terms takes the form of one of the terms in the statement of the proposition.

This concludes the proof of the proposition.
\end{proof}

We can now define $A^{(1)}_\bA$, $A^{(2)}_\bA$ and $A^{(3)}_\bA$, which appeared in the definitions \eqref{whtfi}, \eqref{whtgi2} and \eqref{whtgi3}:
\begin{df}\label{A.def}
Let $A^{(1)}_\bA$, $A^{(2)}_\bA$ and $A^{(3)}_\bA$ be defined as\footnote{Note that the swapping of the indices in $A^{(cross)}$ as compared to Proposition \ref{box.F.prop} is intentional.}
\begin{equation*}
\begin{split}
A^{(1)}_\bA :=& -\left(A^{(1)}_\bA[\phi_0]+A^{(1)}_\bA[F]+\sum_{\bB\neq \bA} A^{(cross)}_{\bB \bA}[F]\right),\\
A^{(2)}_\bA :=& \f 12 \left(A^{(2)}_\bA[\phi_0]+A^{(2)}_\bA[F]\right) {,}\\
A^{(3)}_\bA :=& -\f 13 A^{(3)}_\bA[F],
\end{split}
\end{equation*}
where all the terms on the {RHS} are as in Propositions \ref{box.phi0.prop} and \ref{box.F.prop}.
\end{df}

It is obvious that $A^{(1)}_\bA$, $A^{(2)}_\bA$ and $A^{(3)}_\bA$ obey the bounds as asserted in \eqref{A.bound.1}. We collect them in the following proposition:
\begin{proposition}\label{A.prop}
$A^{(1)}_\bA$, $A^{(2)}_\bA$ and $A^{(3)}_\bA$ obey the bound
$$\|A^{(a)}_\bA\|_{H^8\cap {C^{8}}}+\|\rd_tA^{(a)}_\bA\|_{H^7\cap {C^{7}}}\leq C(C_0)\ep^3\quad\mbox{  for $a=1,2,3$}.$$
\end{proposition}
\begin{proof}
This is immediate from Propositions \ref{box.phi0.prop} and \ref{box.F.prop}.
\end{proof}

\begin{proposition}\label{box.whtF.prop}
For every $\bA\in \mathcal A$, $\Box_g \left(\lambda^2 \wht F_{\bA}\si{\f{u_\bA}{\lambda}}\right)$ can be written as follows:
\begin{equation*}
\begin{split}
&\Box_g \left(\lambda^2 \wht F_{\bA}\si{\f{u_\bA}{\lambda}}\right)-\left(\frac{1}{ \lambda} F_\bA (g_3')^{\alpha \beta}\partial_\alpha u_\bA \partial_\beta u_\bA- \lambda\Box_{g_0} F_\bA +\lambda A_\bA^{(1)}\right) \co{\f{u_\bA}{\lambda}}\\
=&\lambda^2(g_0^{-1})^{tt}\rd_t^2\wht F_{\bA}\si{\f{u_\bA}{\lambda}}+R_\bA[\wht F],
\end{split}
\end{equation*}
where $R_\bA[\wht F]$ denotes terms that are compactly supported in $B(0,R_{supp})$ and satisfy the following estimate:
$$\sum_{k\leq 3}\lambda^k\|R_\bA [\wht F]\|_{H^k}\leq C(C_1)\ep \lambda^2.$$
\end{proposition}
\begin{proof}
Using $(g_0^{-1})^{\alpha\beta}\rd_\alpha u_\bA \rd_\beta u_\bA=0$ and \eqref{whtfi}, we have
\begin{equation*}
\begin{split}
&\Box_g \left(\lambda^2 \wht F_{\bA}\si{\f{u_\bA}{\lambda}}\right)\\
=&\lambda^2(\Box_{g_0}\wht F_{\bA})\si{\f{u_\bA}{\lambda}}+\lambda\left(2 (g_0^{-1})^{\alpha\beta}\rd_\alpha u_\bA\rd_\beta\wht F_{\bA}+(\Box_{g_0} u_\bA)\wht F_{\bA}\right) \co{\f{u_\bA}{\lambda}}\\
&+ \left(\Box_g-\Box_{g_0}\right)\left(\lambda^2 \wht F_{\bA}\si{\f{u_\bA}{\lambda}}\right)\\
=&\underbrace{\lambda^2(\Box_{g_0}\wht F_{\bA})\si{\f{u_\bA}{\lambda}}}_{=:\wht {\mathfrak F}_1}+\underbrace{\left(\frac{1}{ \lambda} {F_\bA} (g_3')^{\alpha \beta}\partial_\alpha u_\bA \partial_\beta u_\bA- \lambda\Box_{g_0} {F_\bA} +\lambda A_\bA^{(1)}\right) \co{\f{u_\bA}{\lambda}}}_{=:\wht {\mathfrak F}_2}\\
&+\underbrace{ \left(\Box_g-\Box_{g_0}\right)\left(\lambda^2 \wht F_{\bA}\si{\f{u_\bA}{\lambda}}\right)}_{=:\wht {\mathfrak F}_3}.
\end{split}
\end{equation*}
For the term $\wht {\mathfrak F}_1$, notice that according to the bootstrap assumption \eqref{BA1}, this is an acceptable error unless both derivatives on $\wht F_\bA$ are $\rd_t$ derivatives, i.e.,
$$\wht {\mathfrak F}_1=\lambda^2(g_0^{-1})^{tt}\rd_t^2\wht F_{\bA}\si{\f{u_\bA}{\lambda}}+\dots,$$
which gives the main term on the {RHS} in the statement of the proposition.

The term $\wht {\mathfrak F}_2$ contributes to the main terms on the {LHS} in the statement of the proposition.

Finally, we claim that according to Proposition~\ref{box.diff.prop}, the bootstrap assumptions and Proposition~\ref{prp:Linfty.BA}, $\wht {\mathfrak F}_3$ can be considered as part of $R_\bA[\wht F]$.
To see this, we first estimate the term on the RHS of \eqref{box.diff} in Proposition~\ref{box.diff.prop} (with $f= \lambda^2 \wht F_{\bA}\si{\f{u_\bA}{\lambda}}$),
\begin{align*}
&C(C_1)\ep\left(\lambda\|\partial f\|_{H^{k}}+\lambda^2\|\partial^2 f\|_{H^{{k}}}+ \lambda^{2-k}\|\partial f\|_{L^\infty}+ \lambda^{2-k}\|\partial^2 f\|_{L^\infty}\right)\\
\leq &C(C_1)\ep\sum_{k_1+k_2= k+1}\lambda^{3-k_1}\|\wht F_\bA\|_{H^{k_2}}
+\sum_{k_1+k_2=k} {(}\lambda^{3-k_1}\|\partial_t \wht F_\bA\|_{H^{k_2}}+ \lambda^{4-k_1}\|\partial^2_t \wht F_\bA\|_{H^{k_2}}{)}\\
&+  \lambda^{2-k}(\|\wht F_\bA \|_{L^\infty}+ \lambda\|\partial \wht F_\bA\|_{L^\infty}+ \lambda^2 \|\partial^{{2}} \wht F_\bA\|_{L^\infty})
{\leq} {C(C_1)\ep^2} \lambda^{2-k},
\end{align*}
where we have used \eqref{BA1}{, \eqref{F.Linfty}} {and the bounds of $u_{\bA}$ in Corollary~\ref{lwp.small}}.
We {then} estimate the {terms on the LHS} in Proposition \ref{box.diff.prop} {(with $f= \lambda^2 \wht F_{\bA}\si{\f{u_\bA}{\lambda}}$)}. {By Corollary~\ref{lwp.small}, \eqref{B2} and \eqref{BA1}, w}e have
$$\|\partial g_1' \partial f \|_{H^k}{\leq C(C_0)\ep^2} \sum_{k_1+k_2=k}\lambda^{1-k_1}\|\partial f\|_{H^{k_2}}
{\leq C(C_1)\ep^4} \sum_{k_1+k_2=k}\lambda^{2-k_1-k_2}{\leq C(C_1)\ep^4} \lambda^{2-k},$$
and {by \eqref{BA1}, \eqref{BA4}, \eqref{F.Linfty}, \eqref{g3.Linfty} and Proposition~\ref{product.local},}
\begin{equation*}
\begin{split}
\|\partial_t g_3'\partial f\|_{H^k}
{\leq} & {C}\left( \|\partial_t g_3'\|_{L^\infty{(B(0,R_{supp}+1)}}\|\partial f\|_{H^k}
+ \|\partial_t g_3'\|_{H^k{(B(0,R_{supp}+1)}}\|\partial f\|_{L^\infty}\right)\\
{\leq} & {C(C_1)} \left(\ep\lambda \|\partial f\|_{H^k}+ \ep\lambda^{1-k}\|\partial f\|_{L^\infty}\right)
{\leq C(C_1)} \ep^2 \lambda^{2-k}.
\end{split}
\end{equation*}
\end{proof}

\begin{proposition}\label{box.whtG.prop}
For every $\bA\in \mathcal A$, $\Box_g \left(\lambda^2 \wht F^{(2)}_{\bA}\co{\f{2u_\bA}{\lambda}}\right)$ and $\Box_g \left(\lambda^2 \wht F^{(3)}_{\bA}\si{\f{3u_\bA}{\lambda}}\right)$ can be written as follows:
$$\Box_g \left(\lambda^2 \wht F^{(2)}_{\bA}\co{\f{2u_\bA}{\lambda}}\right)+2\lambda A^{(2)}_\bA\si{\frac{2u_\bA}{ \lambda}}=R_\bA[\wht F^{(2)}],$$
and
$$\Box_g \left(\lambda^2 \wht F^{(3)}_{\bA}\si{\f{3u_\bA}{\lambda}}\right)-3\lambda A^{(3)}_\bA\co{\frac{3u_\bA}{ \lambda}}=R_\bA[\wht F^{(3)}],$$
where for $a=2,3$, $R_\bA[{\wht F}^{(a)}]$ denotes terms that are compactly supported in $B(0,R_{supp})$ and satisfy the following estimate:
$$\sum_{k\leq 3}\lambda^k\|R_\bA [{\wht F}^{(a)}]\|_{H^k}\leq C(C_1)\ep \lambda^2.$$
\end{proposition}

\begin{proof}
Using $(g_0^{-1})^{\alpha\beta}\rd_\alpha u_\bA \rd_\beta u_\bA=0$ and \eqref{whtgi2}, we have
\begin{equation*}
\begin{split}
&\Box_{g} \left( \lambda^2 \wht F^{(2)}_\bA \co{\frac{{2 u_\bA}}{ \lambda}}\right)\\
= &\lambda^2 \Box_{g_0} \wht F^{(2)}_\bA \co{\frac{{2u_\bA}}{ \lambda}}
-4 \lambda (g_0^{-1})^{\alpha \beta}\partial_\alpha u_\bA \partial_\beta \wht F^{(2)}_\bA \si{\frac{2u_\bA}{ \lambda}}\\
&-2 \lambda \Box_{g_0} u_\bA \wht F^{(3)}_\bA \si{\frac{2u_\bA}{ \lambda}}+(\Box_g-\Box_{g_0})\left( \lambda^2 \wht F^{(2)}_\bA \co{2\frac{u_\bA}{ \lambda}}\right)\\
=& -2\lambda A^{(2)}_\bA\si{\frac{2u_\bA}{ \lambda}}+\dots,
\end{split}
\end{equation*}
where in the last line we have used \eqref{B1} and Proposition \ref{box.diff.prop}. In a similar manner, we have
\begin{equation*}
\begin{split}
\Box_{g} \left( \lambda^2 \wht F^{(3)}_\bA \si{\frac{{3 u_\bA}}{ \lambda}}\right)
= 3\lambda A^{(3)}_\bA \co{\frac{3u_\bA}{ \lambda}}+\dots
\end{split}
\end{equation*}\qedhere

\end{proof}

Using Propositions \ref{box.phi0.prop}, \ref{box.F.prop}, \ref{box.whtF.prop} and \ref{box.whtG.prop}, we see that by the choice of $\wht F_\bA$, $\wht F^{(2)}_\bA$ and $\wht F^{(3)}_\bA$, many of the terms cancel. However, Propositions \ref{box.F.prop} and \ref{box.whtF.prop} each has a term which we cannot control. Moreover, we cannot introduce additional terms in the parametrix to cancel these terms as we otherwise would not have sufficient regularity. Fortunately, there is the following additional cancellation!
\begin{proposition}\label{cancellation}
For every $\bA\in \mathcal A$, the following estimate holds:
$$\sum_{k\leq 3}\lambda^k \left\|\lambda^2(g_0^{-1})^{tt}\rd_t^2 \wht F_\bA-\left(\rd_t (g'_3)^{t\beta}+(g_0^{-1})^{t\beta}\f{\rd_t N_3}{N_0}\right)\rd_\beta u_\bA F_\bA \right\|_{H^k}\leq C(C_1)\ep\lambda^2.$$
\end{proposition}
\begin{proof}
By \eqref{g.inverse}, 
\begin{equation}\label{g0.inverse.recall}
(g_0^{-1})^{tt}=-\f 1{N_0^2},\quad (g_0^{-1})^{ti}=\f{(\beta_0)^i}{N_0^2}.
\end{equation}
{By} \eqref{gp.def}, we have 
\begin{equation}\label{g3p.recall}
(g_3')^{tt}=\f{2N_3}{N_0^3},\quad (g_3')^{ti}=-\f{2N_3 \beta_0^i}{N_0^3}+\f{\beta_3^i}{N_0^2},\quad (g_3')^{ij}=\f{2N_3}{N_0^3}\beta_0^i\beta_0^j-\f 1{N_0^2}(\beta_0^i\beta_3^j+\beta_3^i\beta_0^j){+ e^{-2\gamma_0}(-1+e^{-2\gamma})}.
\end{equation}
Therefore, up to terms which are acceptable error terms, we have
\begin{equation}\label{boxdiff.u.F}
\begin{split}
&\left(\rd_t (g'_3)^{t\beta}+(g_0^{-1})^{t\beta}\f{\rd_t N_3}{N_0}\right)\rd_\beta u_\bA F_\bA\\
=&\left(\f{\rd_t N_3}{N_0^3}\rd_t u_\bA-\f{(\rd_t N_3)\beta_0^i}{N_0^3}\rd_i u_\bA+\f{\rd_t\beta_3^i}{N_0^2}\rd_i u_\bA\right)F_\bA +\dots
\end{split}
\end{equation}
Notice that we have used in the above computation the observation that whenever $\rd_t$ acts on a background $\mfg_0$, the term is an acceptable error. This observation will also be used later without further comments.

To compute $\rd_t^2 \wht F_\bA$, we begin with \eqref{whtfi}. Differentiating \eqref{whtfi} by $\rd_t$, we notice that by the estimates on the background solution and the bounds on $A^{(1)}_\bA$, the only terms which are not acceptable error are those where $\rd_t$ acts on $\wht F_\bA$ or $g_3'$. Therefore, suppressing acceptable error terms, we have
\begin{equation}\label{rdt2.whtF.1}
2(g_0^{-1})^{\alpha \beta}\partial_\alpha u_\bA \partial_\beta \rd_t\wht F_\bA 
= \frac{1}{ \lambda^2} F_\bA \rd_t(g_3')^{\alpha \beta}\partial_\alpha u_\bA \partial_\beta u_\bA+\dots
\end{equation}
Notice now that the {LHS} of \eqref{rdt2.whtF.1} consists of the term where $\beta=t$ and the other terms are acceptable error terms. {On the other hand, for the RHS of \eqref{rdt2.whtF.1}, we compute using \eqref{g3p.recall}. Notice that the only terms that are not acceptable error terms are when $\rd_t$ acts on $\bt^i_3$ or $N_3$ - this is because by \eqref{BA5} (and \eqref{BA4}), $\rd_t\gamma_3$ obeys sufficiently good estimates.} Therefore,
\begin{equation}\label{rdt2.whtF.2}
\begin{split}
&-\f{2}{N_0^2}(\rd_t u_\bA-\beta_0^i\rd_i u_\bA)\rd_t^2 \wht F_\bA\\
=&\frac{1}{ \lambda^2} F_\bA \left(\f{2\rd_t N_3}{N_0^3}(\partial_t u_\bA-\beta_0^i\rd_i u_\bA)^2+\f{2(\rd_t \beta_3^i)}{N_0^2} (\partial_t u_\bA-\beta_0^j \rd_i u_\bA)( \partial_i u_\bA)\right)+\dots
\end{split}
\end{equation}
Since $|\nab u_\bA|>\f 14$ and $u_\bA$ is null with respect to $g_0$, $(\rd_t u_\bA-\beta_0^i\rd_i u_\bA)$ is bounded away from $0$ on the compact set $B(0,R_{supp})${, on which $F_{\bA}$ is supported}. Therefore, after dividing \eqref{rdt2.whtF.2} by $2{\lambda^{-2}}(\rd_t u_\bA-\beta_0^i\rd_i u_\bA)$, we obtain
\begin{equation}\label{rdt2.whtF.3}
\begin{split}
&\lambda^2(g_0^{-1})^{tt} \rd_t^2\wht F_\bA= -\f{\lambda^2}{N_0^2} \rd_t^2 \wht F_\bA\\
=&F_\bA \left(\f{\rd_t N_3}{N_0^3}(\partial_t u_\bA-\beta_0^i\rd_i u_\bA)+\f{\rd_t \beta_3^i}{N_0^2}\rd_i u_\bA \right)+\dots
\end{split}
\end{equation}
Finally, notice that there is an exact cancellation in the main terms in \eqref{boxdiff.u.F} and \eqref{rdt2.whtF.3}, which yields the proposition.
\end{proof}

We now summarize all the calculations and arrive at the following main result of the subsection:
\begin{proposition}\label{main.term}
There exist functions $B^{(2,\pm)}_{\bA\bB}$, $B^{(3,\pm)}_{\bA\bB}$ and $D_{\bA\bB{\bf C}}^{(\pm_1,\pm_2)}$, each of which depends only on the background solution, is compactly supported in $B(0,R_{supp})$ and obeys the bound
$$\|\cdot\|_{H^8\cap {C^{8}}}+\|\rd_t(\cdot)\|_{H^7\cap {C^{7}}}+\|\rd_t^2(\cdot)\|_{H^6\cap {C^{6}}}\leq C(C_0)\ep^3$$
such that
\begin{equation*}
\begin{split}
\Box_g(\phi_\lambda-\mathcal E_\lambda)=&\left(\rd_t (g'_3)^{t\beta}+(g_0^{-1})^{t\beta}\f{\rd_t N_3}{N_0}\right)(\rd_\beta \phi_0)\\
&+\sum_{\pm}\sum_\bA\sum_{\bB\neq \bA} \lambda B^{(2,\pm)}_{\bA\bB}\si{\f{u_\bA\pm u_\bB}{\lambda}}+\sum_{\pm}\sum_\bA\sum_{\bB}  \lambda B^{(3,\pm)}_{\bA\bB}\co{\f{u_\bA\pm 2u_\bB}{\lambda}}\\
&+\sum_{\pm_1}\sum_{\pm_2}\sum_\bA\sum_{\bB\neq \bA}\sum_{\substack{{\bf C}\neq \bB\\{\bf C}\neq \bA}}\lambda D_{\bA\bB{\bf C}}^{(\pm_1,\pm_2)}\co{\f{u_\bA\pm_1 u_\bB\pm_2 u_{\bf C}}{\lambda}}+R,
\end{split}
\end{equation*} 
where $R$ is compactly supported in $B(0,R_{supp})$ and satisfies the estimate
\begin{equation}\label{R.bound}
\sum_{k\leq 3}\lambda^k \| R\|_{H^k}\leq C(C_1)\ep\lambda^2.
\end{equation}
\end{proposition}
\begin{proof}
This is achieved by combining Propositions \ref{box.phi0.prop}, \ref{box.F.prop}, \ref{box.whtF.prop}, \ref{box.whtG.prop} and \ref{cancellation}, where we have defined 
$$B^{(2,\pm)}_{\bA\bB}:=B^{(2,\pm)}_{\bA\bB}[\phi_0]+B^{(2,\pm)}_{\bA\bB}[F],\quad B^{(3,\pm)}_{\bA\bB}:=B^{(3,\pm)}_{\bA\bB}[F],\quad D_{\bA\bB{\bf C}}^{(\pm_1,\pm_2)}:=D_{\bA\bB{\bf C}}^{(\pm_1,\pm_2)}[F],$$
and 
\begin{equation*}
\begin{split}
R:=&R[\phi_0]+\sum_\bA\left(R_\bA[F]+R_\bA[\wht F]+R_\bA[\wht F^{(2)}]+R_\bA[\wht F^{(3)}]\right)\\
&+\sum_\bA \left(\lambda^2(g_0^{-1})^{tt}\rd_t^2 \wht F_\bA-\left(\rd_t (g'_3)^{t\beta}+(g_0^{-1})^{t\beta}\f{\rd_t N_3}{N_0}\right)\rd_\beta u_\bA F_\bA \right).
\end{split}
\end{equation*}
\end{proof}

\subsection{Estimates for $\wht F_\bA$}\label{secwhtf}

We now begin to estimate the terms in the parametrix for $\phi$. Since $\wht F_\bA$, $\wht F^{(2)}_\bA$ and $\wht F^{(3)}_\bA$ are defined via transport equations, it is convenient to have a general lemma for obtaining estimates:
\begin{lemma}\label{lemma.transport}
Suppose $f$ satisfies the transport equation
$$2(g_0^{-1})^{\alpha\beta} (\partial_\alpha u_\bA) \partial_\beta f +(\Box_{g_0}u_\bA) f= h,$$
with $f$ and $h$ both\footnote{Note the trivial fact that just assuming that $supp(h)\subset B(0,R_{supp})$ does \underline{not} imply $supp(f)\subset B(0,R_{supp})$. This lemma however applies when the support properties are assumed for \underline{both} $f$ and $h$.} compactly supported in $B(0,R_{supp})$.
Then, the following estimate holds for $k\leq 9$:
\begin{equation}\label{lemma.transport.0}
\sup_{t\in [0,1]}\|f\|_{H^k}(t)\leq C(C_0)\left(\|f\|_{H^k}(0)+\int_0^1 \|h\|_{H^k}(s)\, ds\right),
\end{equation}
the following estimate holds for $k\leq 8$:
\begin{equation}\label{lemma.transport.t}
\sup_{t\in [0,1]}\|\rd_t f\|_{H^k}(t)\leq C(C_0)\left(\|f\|_{H^{k+1}}(0)+\int_0^1 \|h\|_{H^{k+1}}(s)\, ds\right){,}
\end{equation}
and the following estimate holds for $k\leq 7$:
\begin{equation}\label{lemma.transport.tt}
\sup_{t\in [0,1]}\|\rd_t^2 f\|_{H^k}(t)\leq C(C_0)\left(\|\rd_t f\|_{H^{k+1}}(0)+\int_0^1 \left(\|h\|_{H^{k+1}}+\|\rd_t h\|_{H^{k+1}}\right)(s)\, ds\right).
\end{equation}
\end{lemma}

\begin{proof}
In coordinates, we can write
\begin{equation}\label{transport.coord}
2(g_0^{-1})^{\alpha t} \partial_\alpha u_\bA \partial_t f +2 (g_0^{-1})^{i\alpha}\partial_{\alpha} u_\bA \partial_i f+(\Box_{g_0}u_\bA) f= h.
\end{equation}
To derive the desired estimate in the $k=0$ case, simply multiply \eqref{transport.coord} by $f$, integrate in a spacetime region $[0,t]\times \m R^2$ for $t\in [0,1]$ with respect to the measure $dx\, dt$, integrate by parts and apply the obvious estimates. Notice that the compact support of $f$ and $h$ allows us not to track the behavior at infinity. For higher $k$, one simply differentiate the equation $k$ times in $x$. Notice that we need $k\leq 9$ due to the regularity\footnote{In fact, if we carefully track the regularity of $\Box_{g_0} u_\bA$ in the proof of Theorem \ref{lwp}, we can have $k\leq 10$. Such improvement is of course completely irrelevant.} of the background solution given in Corollary \ref{lwp.small}.

Next, for the estimates on $\rd_t f$ (i.e., \eqref{lemma.transport.t}), one uses the equation \eqref{transport.coord} to rewrite $\rd_t f$ as $\rd_i f$ plus lower order terms. This is possible since we have a lower bound on $\left|(g_0^{-1})^{\alpha t} \partial_\alpha u_\bA\right|$ thanks to the bound $|\nab u_\bA|>\f 14$ (see Definition \ref{def.spatial}) and the fact that $u_\bA$ is an eikonal function (with respect to $g_0$).

Finally, for \eqref{lemma.transport.tt}, we simply differentiate \eqref{transport.coord} by $\rd_t$, apply \eqref{lemma.transport.t} to $\tilde{f}=\rd_t f$ (with $\tilde{h}=\rd_t h+[2(g_0^{-1})^{\alpha\beta} (\partial_\alpha u_\bA) \partial_\beta, \rd_t] f-(\rd_t\Box_{g_0} u_\bA)f$) and then use \eqref{lemma.transport.0} and \eqref{lemma.transport.t} to control the error terms.
\end{proof}

Consequently, we have
\begin{proposition}\label{whtF.est}
$\wht F_\bA$ satisfies the following estimates:
$$\sum_{k \leq 3} \left(\lambda^k \|\wht F_\bA\|_{H^{2+k}} + \lambda^k \|\partial_t \wht F_\bA\|_{H^{1+k}} +\lambda^{k+1}\|\rd_t^2 \wht F_\bA\|_{H^k}\right)\leq C(C_0)\ep + C(C_1)\ep^2.$$
\end{proposition}
\begin{proof}
The idea is to apply Lemma \ref{lemma.transport}. In the transport equation \eqref{whtfi} for $\wht F_\bA$, except for the term $\frac{1}{ \lambda^2} F_\bA (g_3')^{\alpha \beta}\partial_\alpha u_\bA \partial_\beta u_\bA$, all the terms on the {RHS} depend only on the background and can be controlled in the norm $\sum_{k\leq 3}\lambda^k\|\cdot\|_{H^{2+k}}$ by $\leq C(C_0)\ep$. (In particular, the $A^{(1)}_\bA$ term can be controlled by Proposition \ref{A.prop}.) It thus remains to control the term $\frac{1}{ \lambda^2} F_\bA (g_3')^{\alpha \beta}\partial_\alpha u_\bA \partial_\beta u_\bA$. The main contribution of this term comes from $g_3'$, which can be estimated by Proposition \ref{g3p.prop}.
\end{proof}

Notice that we have in particular improved the bootstrap assumptio{n \eqref{BA1}.}

\subsection{Estimates for the terms $\wht F^{(2)}_\bA$ and $\wht F^{(3)}_\bA$}\label{secwhtg}

\begin{proposition}\label{whtG.prop}
For $a=2,3$, ${\wht F}^{(a)}_\bA$ (which depends only on the background solution) obeys the following estimates
$$\|{\wht F}_\bA^{(a)}\|_{H^{8}}+ \|\partial_t {\wht F}_{\bA}^{(a)}\|_{H^{7}}
+\|\rd_t^2 {\wht F}_\bA^{(a)}\|_{H^{6}}
\leq C(C_0) \ep^2.$$
\end{proposition}
\begin{proof}
By Lemma \ref{lemma.transport} and the transport equations \eqref{whtgi2} and \eqref{whtgi3}, it suffices to control the terms $A_\bA^{(a)}$ and $\rd_t A_\bA^{(a)}$. This then follows from Proposition \ref{A.prop}.

\end{proof}

\subsection{Estimates for the remaining term $\mathcal E_\lambda$}\label{secelambda}

Given the estimates derived in Sections \ref{secwhtf} and \ref{secwhtg} above, in order to control the scalar field $\phi_\lambda$, it remains to estimate $\mathcal E_\lambda$. Since $\Box_g\phi_\lambda=0$ by \eqref{sys}, $\mathcal E_\lambda$ obeys the equation 
\begin{equation}\label{box.E.final}
\Box_g\mathcal E_\lambda=-\Box_g(\phi_\lambda-\mathcal E_\lambda)=-(\mbox{terms in Proposition \ref{main.term}}).
\end{equation}

We further decompose $\mathcal E_\lambda$ and the remainder of this section will be dedicated to estimating each of the decomposed piece. More precisely, define
\begin{equation}\label{E.decompose}
\mathcal E_\lambda:=\mathcal E_\lambda^{(elliptic)}+\mathcal E_\lambda^{(\rd_t\mfg_3)}+\mathcal E_\lambda^{(Error)},
\end{equation}
where 
\begin{equation}\label{E1.def}
\begin{split}
&\mathcal E_\lambda^{(elliptic)}\\
:=&{\sum_{\pm}}\sum_\bA\sum_{\bB\neq \bA}  \f{\lambda^3 B^{(2,\pm)}_{\bA\bB}}{(g_0^{-1})^{\alpha\beta}(\rd_\alpha (u_\bA\pm  u_\bB))(\rd_\beta (u_\bA\pm  u_\bB))}\si{\f{u_\bA\pm u_\bB}{\lambda}}\\
&{+}\sum_{\pm}\sum_\bA\sum_{\bB}  \f{\lambda^3 B^{(3,\pm)}_{\bA\bB}}{(g_0^{-1})^{\alpha\beta}(\rd_\alpha (u_\bA\pm  2u_\bB))(\rd_\beta (u_\bA\pm  2u_\bB))}\co{\f{u_\bA\pm 2u_\bB}{\lambda}}\\
&{+}\sum_{\pm_1}\sum_{\pm_2}\sum_\bA\sum_{\bB\neq \bA}\sum_{\substack{{\bf C}\neq \bB\\{\bf C}\neq \bA}}\f{\lambda^3 D_{\bA\bB{\bf C}}^{(\pm_1,\pm_2)}}{(g_0^{-1})^{\alpha\beta}(\rd_\alpha (u_\bA\pm_1  u_\bB \pm_2 u_{\bf C}))(\rd_\beta (u_\bA \pm_1  u_\bB\pm_2 u_{\bf C}))}\co{\f{u_\bA\pm_1 u_\bB\pm_2 u_{\bf C}}{\lambda}},
\end{split}
\end{equation}
with $B^{(2,\pm)}_{\bA\bB}$, $B^{(3,\pm)}_{\bA\bB}$ and $D_{\bA\bB{\bf C}}^{(\pm_1,\pm_2)}$ as in Proposition \ref{main.term}; $\mathcal E_\lambda^{(\rd_t\mfg_3)}$ is the unique solution to the wave equation
\begin{equation}\label{E2.def}
\begin{cases}
\Box_g \mathcal E_\lambda^{(\rd_t\mfg_3)}={-}\left(\rd_t (g'_3)^{t\beta}+(g_0^{-1})^{t\beta}\f{\rd_t N_3}{N_0}\right)(\rd_\beta \phi_0)\\
(\mathcal E_\lambda^{(\rd_t\mfg_3)}, \rd_t\mathcal E_\lambda^{(\rd_t\mfg_3)})\restriction_{\Sigma_0}=(0,0);
\end{cases}
\end{equation}
and $\mathcal E_\lambda^{(Error)}$ is defined so that \eqref{E.decompose} holds. According to Proposition \ref{main.term}, this means that $\mathcal E_\lambda^{(Error)}$ satisfies
\begin{equation}\label{E3.eqn}
\begin{split}
{-}\Box_g \mathcal E_\lambda^{(Error)}=&R+\sum_{\pm}\sum_\bA\sum_{\bB\neq \bA} \lambda B^{(2,\pm)}_{\bA\bB}\si{\f{u_\bA\pm u_\bB}{\lambda}}+\sum_{\pm}\sum_\bA\sum_{\bB}  \lambda B^{(3,\pm)}_{\bA\bB}\co{\f{u_\bA\pm 2u_\bB}{\lambda}}\\
&+\sum_{\pm_1}\sum_{\pm_2}\sum_\bA\sum_{\bB\neq \bA}\sum_{\substack{{\bf C}\neq \bB\\{\bf C}\neq \bA}}\lambda D_{\bA\bB{\bf C}}^{(\pm_1,\pm_2)}\co{\f{u_\bA\pm_1 u_\bB\pm_2 u_{\bf C}}{\lambda}}{+}\Box_g\mathcal E_\lambda^{(elliptic)}{,}
\end{split}
\end{equation}
with initial conditions
\begin{equation}\label{E3.IC}
\begin{split}
\mathcal E_\lambda^{(Error)}\restriction_{\Sigma_0}=&
\phi-\phi_0-\sum_\bA \lambda F_\bA \cos\left(\frac{u_\bA}{\lambda}\right)-\mathcal E_\lambda^{(elliptic)}\\
\rd_t\mathcal E_\lambda^{(Error)}\restriction_{\Sigma_0}
=&\rd_t\Big(\phi-\phi_0-
\sum_\bA \lambda F_\bA \cos\left(\frac{u_\bA}{\lambda}\right) -\sum_\bA \lambda^2 \wht F_\bA \si{\frac{u_\bA}{\lambda}}\\
&\quad-\sum_\bA \lambda^2 \wht F_\bA^{(2)}\co{\frac{2u_\bA}{\lambda}}-\sum_\bA  \lambda^2\wht F_\bA^{(3)} \si{\frac{3u_\bA}{\lambda }}
-\mathcal E_\lambda^{(elliptic)}\Big)\restriction_{\Sigma_0}.
\end{split}
\end{equation}

Notice that $\mathcal E_\lambda^{(elliptic)}$ is well-defined thanks to the fact that the $u_\bA$'s are null adapted (i.e., by Definition \ref{def.null}, we can divide by $(g_0^{-1})^{\alpha\beta}(\rd_\alpha (u_\bA\pm  u_\bB))(\rd_\beta (u_\bA\pm  u_\bB))$, etc.) The choice of $\mathcal E_\lambda^{(elliptic)}$ (and $\mathcal E_\lambda^{(\rd_t\mfg_3)}$ ) is such that the {RHS} of \eqref{E3.eqn} and the initial conditions in \eqref{E3.IC} are at most of size $C(C_1)\ep\lambda^2$ and so that $\mathcal E_\lambda^{(elliptic)}$ itself is also at most $C(C_0)\ep^3 \lambda^3$. That this can be achieved is precisely because the $u_\bA$'s are null adapted so that order $C(C_0)\ep^2\lambda$ oscillating terms are oscillating in a \underline{non-null} direction. Consequently, for the high frequency part, $\Box_g$ heuristically ``looks like an elliptic operator''.

In the next two propositions, we make precise the smallness of $\mathcal E_\lambda^{(elliptic)}$, of the {RHS} of \eqref{E3.eqn} and of \eqref{E3.IC}.

\begin{proposition}\label{E1.prop}
{$\mathcal E_\lambda^{(elliptic)}$ satisfies}
$$\sum_{k\leq 8}\lambda^k\|\mathcal E_\lambda^{(elliptic)}\|_{H^k}+\sum_{k\leq 7}\lambda^{k+1}\|\rd_t\mathcal E_\lambda^{(elliptic)}\|_{H^k}+\sum_{k\leq 6}\lambda^{k+2}\|\rd_t^2\mathcal E_\lambda^{(elliptic)}\|_{H^k}\leq C(C_0)\ep^3 \lambda^3.$$
\end{proposition}
\begin{proof}
This is an immediate consequence of the definition \eqref{E1.def} and the estimates in Proposition \ref{main.term} for $B^{(2,\pm)}_{\bA\bB}$, $B^{(3,\pm)}_{\bA\bB}$ and $D_{\bA\bB{\bf C}}^{(\pm_1,\pm_2)}$.
\end{proof}

\begin{proposition}\label{est.E3.inho}
The {RHS} of \eqref{E3.eqn} (for $\mathcal E_\lambda^{(Error)}$) can be estimated by
\begin{equation}\label{est.E3.inho.1}
\sum_{k\leq 3}\lambda^k\|\cdot\|_{H^k}\leq C(C_1)\ep\lambda^2.
\end{equation}
Moreover, the initial data for $\mathcal E_\lambda^{(Error)}$ can be bounded as follows:
\begin{equation}\label{est.E3.inho.2}
\sum_{k\leq {5}}\lambda^k\left\|\mathcal E_\lambda^{(Error)}\restriction_{\Sigma_0}\right\|_{H^k}+\sum_{k\leq {4}}\lambda^{k{+1}}\left\| (\rd_t\mathcal E_\lambda^{(Error)})\restriction_{\Sigma_0}\right\|_{H^k}\leq  C(C_0)\ep\lambda^2.
\end{equation}
\end{proposition}
\begin{proof}
To estimate the {RHS} of \eqref{E3.eqn}, we first note that the term $R$ satisfies the desired estimate by \eqref{R.bound} in Proposition \ref{main.term}. Then, note that  $\mathcal E_\lambda^{(elliptic)}$ is constructed so that $\Box_g\mathcal E_\lambda^{(elliptic)}$ cancels with the remaining terms up to error terms of size $C(C_0)\ep^3\lambda^2$, i.e.,
\begin{equation*}
\begin{split}
\sum_{k\leq 6}\lambda^k\|&\sum_{\pm}\sum_\bA\sum_{\bB\neq \bA} \lambda B^{(2,\pm)}_{\bA\bB}\si{\f{u_\bA\pm u_\bB}{\lambda}}+\sum_{\pm}\sum_\bA\sum_{\bB}  \lambda B^{(3,\pm)}_{\bA\bB}\co{\f{u_\bA\pm 2u_\bB}{\lambda}}\\
&+\sum_{\pm_1}\sum_{\pm_2}\sum_\bA\sum_{\bB\neq \bA}\sum_{\substack{{\bf C}\neq \bB\\{\bf C}\neq \bA}}\lambda D_{\bA\bB{\bf C}}^{(\pm_1,\pm_2)}\co{\f{u_\bA\pm_1 u_\bB\pm_2 u_{\bf C}}{\lambda}}-\Box_g\mathcal E_\lambda^{(elliptic)}\|_{H^k} \leq C(C_0)\ep^3\lambda^2.
\end{split}
\end{equation*}
We hence conclude \eqref{est.E3.inho.1}. To prove \eqref{est.E3.inho.2}, we recall \eqref{E3.IC}{. By Lemma~\ref{lmini},} 
$$\left\|\left(\phi-\phi_0-\sum_\bA \lambda F_\bA \cos\left(\frac{u_\bA}{\lambda}\right),\dot{\phi}-\dot{\phi_0}+\sum_\bA F_\bA |\nabla u_\bA| e^{\gamma_0} \sin\left(\frac{u_\bA}{\lambda}\right)\right)\restriction_{\Sigma_0}\right\|_{{H^5\times H^4}}{\leq C\ep\lambda^2.}$$ 
{By Proposition~\ref{E1.prop},}
$${\sum_{k\leq 8}\lambda^k\|\mathcal E_\lambda^{(elliptic)}\|_{H^k}+\sum_{k\leq 7}\lambda^{k+1}\|\rd_t\mathcal E_\lambda^{(elliptic)}\|_{H^k}\leq C(C_0)\ep^3 \lambda^3.}$$ 
{T}herefore the initial data for ${(}\mathcal E_\lambda^{(Error)}{, \rd_t\mathcal E_\lambda^{(Error)})}$ {obey the stated estimates}. ({Here, we have used} that the initial data for $\wht F_{\bA}$, $\wht F^{(2)}_{\bA}$ and $\wht F^{(3)}_{\bA}$ are zero, and if a time derivative falls on these term {(instead of the oscillating factors)}, this give a contribution {that can be} controlled by a higher power of $\lambda$).
\end{proof}

Proposition \ref{est.E3.inho} is already a good indication that $\mathcal E_\lambda^{(Error)}$ can be appropriately controlled. To proceed, we consider general energy estimates for solutions to general wave equations so that we can treat both $\mathcal E_\lambda^{(\rd_t\mfg_3)}$ and $\mathcal E_\lambda^{(Error)}$.

\begin{proposition}\label{EE.prop}
Let $f$ be a solution to
$$\Box_g f=h$$
with both $f$ and $h$ being compactly supported in $B(0,R_{supp})$. Then there exists a constant $C(C_0)>0$ such that $f$ obeys the estimate
\begin{equation}\label{EE.prop.main}
\sup_{t\in [0,1]}\|\rd f\|_{L^2}^2(t)\leq C(C_0) \left(\|\rd f\|_{L^2}^2(0)+\sup_{t\in [0,1]}\left|\int_0^t \int_{\Sigma_{t'}} (\rd_t f) h(t',x)\, \sqrt{|\det g|}\, dx\, dt'\right|\right).
\end{equation}
Consequently, it also holds that
\begin{equation}\label{EE.prop.2}
\sup_{t\in [0,1]}\|\rd f\|_{L^2}(t)\leq C(C_0) \left(\|\rd f\|_{L^2}(0)+\|h\|_{L^2}\right).
\end{equation}
\end{proposition}
\begin{proof}
Define a $2$-tensor $Q$ as follows:
$$Q_{\alpha \beta}[f] :=\partial_\alpha  f \partial_\beta  f -\frac{1}{2}g_{\alpha \beta}\gi^{\rho\sigma}\partial_\sigma 
 f \partial_\rho  f .$$
By the equation $\Box_g f=h$, we have
$$\gi^{\mu\alpha} D_{\mu} Q_{\alpha \beta}[f]=h \rd_\beta f,$$
where $D$ is the Levi-Civita connection associated to $g$. Define also the \emph{deformation tensor} associated to $\rd_t$
$$^{(\rd_t)}\pi_{\alpha\beta}=D_\alpha (\rd_t)_\beta+D_\beta (\rd_t)_\alpha.$$
By the Stoke's theorem, we have that for every $t\in (0,1]$,
\begin{equation*}
\begin{split}
&\int_{\Sigma_t} Q[f](\rd_t, \f{1}{{N}} e_0)(t,x)\, \sqrt{|\det {\bar{g}}|}\, dx\\
=&\int_{\Sigma_0} Q[f](\rd_t, \f{1}{{N}} e_0)(0,x)\, \sqrt{|\det {\bar{g}}|}\, dx\\
&-\int_0^t \int_{\Sigma_{t'}} \left((\rd_t f) h+\f 12 Q_{\alpha\beta}[f] { }^{(\rd_t)}\pi^{\alpha\beta}\right)(t',x)\, \sqrt{|\det g|}\, dx\, dt'{,}
\end{split}
\end{equation*}
where $\bar{g}$ is as in \eqref{g.form.0}. 
It follows from our assumptions on the background and \eqref{B2}, \eqref{BA3}, \eqref{BA4} that on the compact set $B(0,R_{supp})$,
$$Q[f](\rd_t, \f{1}{{N}} e_0)\sim_{C_0} |\rd f|^2,\quad {\sqrt{|\det \bar{g}|}\sim_{C_0}}\sqrt{|\det g|}\sim_{C_0} 1,\quad |Q_{\alpha\beta}[f] { }^{(\rd_t)}\pi^{\alpha\beta}|\leq C(C_0) |\rd f|^2,$$
where we denote $A\sim_{C_0} B$ if $A\leq C(C_0) B$ and $B\leq C(C_0) A$. \eqref{EE.prop.main} hence follows from the above observations and Gr\"onwall's inequality. To get from \eqref{EE.prop.main} to \eqref{EE.prop.2}, one simply applies the {Cauchy--Schwarz} inequality.
\end{proof}
From Proposition \ref{EE.prop}, we also derive the following energy estimates for higher derivatives:
\begin{cor}\label{EE.cor}
Let $f$ and $h$ be as in Proposition \ref{EE.prop}. For $1\leq k\leq 3$, it holds that
\begin{equation}
\begin{split}
&\sup_{t\in [0,1]}\left(\|\rd f\|_{H^k}(t)+\|\rd^2 f\|_{H^{k-1}}(t)\right)\\
\leq &C(C_0) \left(\|\rd f\|_{H^k}(0)+\|h\|_{H^k}+\sum_{k'\leq k-1}\lambda^{k'-k}\left(\|\rd f\|_{H^{k'}}+\|\rd^2 f\|_{H^{k'-1}}\right)\right).
\end{split}
\end{equation}
\end{cor}
\begin{proof}
We first estimate $\|\rd f\|_{H^k}$. The idea is to apply {\eqref{EE.prop.2}} for $\tilde{f}=\rd_{x^1}^{k_1} \rd_{x^2}^{k_2}f$ and then sum over all the estimates for $k_1+k_2\leq k$. For such $\tilde{f}$, we have
\begin{equation}\label{th}
\Box_g \tilde{f}=\tilde{h}:= \rd_{x^1}^{k_1} \rd_{x^2}^{k_2} h+[\Box_g, \rd_{x^1}^{k_1} \rd_{x^2}^{k_2}] f.
\end{equation}
Now for $k_1+k_2\leq k$, we {use Corollary~\ref{lwp.small} to control the background metric and obtain}
\begin{equation}\label{box.commute}
\begin{split}
&|[\Box_g, \rd_{x^1}^{k_1} \rd_{x^2}^{k_2}] f|\\
\leq &C(C_0)\left(\underbrace{\sum_{k'_1\leq k}\ep|\rd \nab^{k_1'} f|}_{=:I_k}+\underbrace{\sum_{k'_1\leq k-1}\ep|\rd^2 \nab^{k_1'} f|}_{=:II_k}\right.\\
&\qquad\qquad\left.+\underbrace{\sum_{i=1,2,3} \sum_{k_1'+k_2'\leq k} |\rd \nab^{k'_1} \mfg_i||\rd \nab^{k'_2} f|}_{=:III_k}+\underbrace{\sum_{i=1,2,3} \sum_{\substack{k_1'+k_2'\leq k\\k_1'\geq 1}} |\nab^{k'_1} \mfg_i||\rd^2 \nab^{k'_2} f|}_{=:IV_k}\right).
\end{split}
\end{equation}
Therefore, applying {\eqref{EE.prop.2}} to $\tilde{f}$ (for every $k_1+k_2\leq k$) and using \eqref{th} and \eqref{box.commute}, we get
\begin{equation}\label{EE.higher}
\begin{split}
\sup_{t\in [0,1]}\|\rd f\|_{H^k}(t)
\leq &C(C_0) \left(\|\rd f\|_{H^{k}}(0)+\|h\|_{H^k}+\|I_k\|_{L^2}+\|II_k\|_{L^2}+\|III_k\|_{L^2}+\|IV_k\|_{L^2}\right)\\
\end{split}
\end{equation}
It therefore remains to control the terms $I_k$, $II_k$, $III_k$ and $IV_k$. {It is immediate that}
\begin{align}
\|I_k\|_{L^2}\leq C(C_0) \|\rd f\|_{H^k}\label{I.bound}\\
\|II_k\|_{L^2}\leq C(C_0)\|\rd^2 f\|_{H^{k-1}}\label{II.bound}
\end{align}
For $III_k$, we split into the case where $k\leq 2$ and $k=3$. If \underline{$k\leq 2$}, then by {H\"older's inequality, \eqref{B2}, \eqref{g2.Linfty} and \eqref{g3.Linfty},}
\begin{equation}\label{III.bound.1}
\|III_k\|_{L^2}\leq C(C_1)\ep\sum_{\substack{k_1'+k_2'\leq k\\k_1'\leq 2}} \lambda^{-k_1'}\|\rd f\|_{H^{k_2'}}.
\end{equation}
If \underline{$k=3$} {and $i=1$, it can be treated as in the $k\leq 2$ case. If \underline{$k=3$} and} $i\neq 1$, we treat separately the cases $k_1'\leq 2$ and $k_1'=3$ {using H\"older's inequality, \eqref{BA3}, \eqref{BA4}, \eqref{g2.Linfty} and \eqref{g3.Linfty}} to get
\begin{equation}\label{III.bound.2}
\begin{split}
\|III_k\|_{L^2}\leq &C(C_1)\ep\sum_{\substack{k_1'+k_2'\leq k\\k_1'\leq 2}} \lambda^{-k_1'}\|\rd f\|_{H^{k_2'}}+\sum_{i=2,3}\|\rd \mfg_i\|_{H^3}\|\rd f\|_{L^\infty}\\
\leq &C(C_1)\ep\sum_{k_1'+k_2'\leq 3} \lambda^{-k_1'}\|\rd f\|_{H^{k_2'}}+C(C_1)\ep{\lambda^{-1}}\|\rd f\|_{H^{{2}}}\leq C(C_1)\ep\sum_{k_1'+k_2'\leq k} \lambda^{-k_1'}\|\rd f\|_{H^{k_2'}},
\end{split}
\end{equation}
where we have {also} used {Proposition~\ref{holder}} {(more speci{fi}cally $H^2\subset L^\infty$)} for $\rd f$. Hence, in both cases, $\|{III_k}\|_{L^2}$ {obeys the desired estimate}. Finally, for {$IV_k$}, using {\eqref{B2}, \eqref{g2.Linfty} and \eqref{g3.Linfty}}, we get
\begin{equation}\label{IV.bound}
\|IV_k\|_{L^2}\leq C(C_1) \ep \sum_{\substack{k_1'+k_2'\leq k\\ k_1'\geq 1}}\lambda^{1-k_1'} \|\rd^2 f\|_{H^{k_2'}}.
\end{equation}
Combining \eqref{EE.higher}, \eqref{I.bound}, \eqref{II.bound}, \eqref{III.bound.1}, \eqref{III.bound.2} {and} \eqref{IV.bound}, we obtain
\begin{equation}\label{EE.higher.2}
\begin{split}
&\sup_{t\in [0,1]}\|\rd f\|_{H^k}(t)\\
\leq &C(C_0) \left(\|\rd f\|_{H^{k}}(0)+\|h\|_{H^k}+C(C_1) \ep \left(\sum_{k_1'+k_2'\leq k} \lambda^{-k_1'}\|\rd f\|_{H^{k_2'}}+\sum_{\substack{k_1'+k_2'\leq k\\ k_1'\geq 1}}\lambda^{1-k_1'} \|\rd^2 f\|_{H^{k_2'}}\right)\right).
\end{split}
\end{equation}
The estimate \eqref{EE.higher.2} almost controls {all components of} $\|\rd^2 f\|_{H^{k-1}}(t)$, except for the term $\|\rd_t^2 f\|_{H^{k-1}}(t)$. Now using the wave equation and noticing that $(g^{-1})^{tt}$ is bounded away from $0$ on $B(0,R_{supp})$, we can write 
$$\rd_t^2 f=-\f{1}{(g^{-1})^{tt}}\left(2\gi^{ti}\rd^2_{ti} f+\gi^{ij}\rd^2_{ij} f-{\gi}^{\alpha \beta}\Gamma^\rho_{\alpha \beta} \partial_\rho f-h\right)$$
to get 
\begin{equation*}
\begin{split}
&\sup_{t\in [0,1]}\left(\|\rd f\|_{H^k}(t)+\|\rd^2 f\|_{H^{k-1}}(t)\right)\\
\leq &C(C_0) \left(\|\rd f\|_{H^{k}}(0)+\|h\|_{H^k}+C(C_1) \ep \left(\sum_{k_1'+k_2'\leq k} \lambda^{-k_1'}\|\rd f\|_{H^{k_2'}}+\sum_{\substack{k_1'+k_2'\leq k\\ k_1'\geq 1}}\lambda^{1-k_1'} \|\rd^2 f\|_{H^{k_2'}}\right)\right).
\end{split}
\end{equation*}
The conclusion follows after absorbing $C(C_1)\ep \left(\|\rd f\|_{H^k}(t)+\|\rd^2 f\|_{H^{k-1}}(t)\right)$ to the {LHS}.
\end{proof}

\begin{proposition}\label{E.rdtg3.est}
{$\mathcal E_\lambda^{(\rd_t\mfg_3)}$ satisfies}
$$\sum_{k\leq 3} \lambda^k\|\rd \mathcal E_\lambda^{(\rd_t\mfg_3)}\|_{H^k}+\sum_{k\leq 2} \lambda^{k+1}\|\rd^2 \mathcal E_\lambda^{(\rd_t\mfg_3)}\|_{H^k}\leq C(C_1)\ep \lambda^2.$$
\end{proposition}
\begin{proof}
The idea is of course to apply Proposition \ref{EE.prop} (more precisely, \eqref{EE.prop.main}) and Corollary \ref{EE.cor} for 
$$f=\mathcal E_\lambda^{(\rd_t\mfg_3)}\mbox{ and }h={-}\left(\rd_t (g'_3)^{t\beta}+(g_0^{-1})^{t\beta}\f{\rd_t N_3}{N_0}\right)(\rd_\beta \phi_0).$$ 
We will do this in the following steps:

{\bf Estimates for $\|\rd \mathcal E_\lambda^{(\rd_t\mfg_3)}\|_{L^2}$. } It turns out that this is the hardest case in this proposition. By the assumptions on the background solution and the bootstrap assumption \eqref{BA4}, we have $\|h\|_{L^2}\leq C(C_1)\ep\lambda$, which if we directly apply \eqref{EE.prop.2} gives an estimate that is off by one factor of $\lambda^{-1}$. Instead, we use \eqref{EE.prop.main}. To simplify the exposition, let us define $h=h_1+h_2$ with $h_1:=(\rd_t (g'_3)^{t\beta}){(\rd_\bt \phi_0)}$ and $h_2:=(g_0^{-1})^{t\beta}\f{\rd_t N_3}{N_0}(\rd_\beta \phi_0)$. Since they can be treated similarly, we will only consider $h_1$. More precisely, we consider the following term that arises in \eqref{EE.prop.main}:
\begin{equation*}
\begin{split}
\sup_{t\in [0,1]}\left|\int_0^t \int_{\Sigma_{t'}} (\rd_t \mathcal E_\lambda^{(\rd_t\mfg_3)}) (\rd_t (g'_3)^{t\beta})(\rd_\beta \phi_0)(t',x)\, \sqrt{|\det g|}\, dx\, dt'\right|.
\end{split}
\end{equation*}
The idea is to integrate by parts in $t$, use the wave equation for $\mathcal E_\lambda^{(\rd_t\mfg_3)}$ and integrate by parts in $x$. We integrate by parts below, but only writing down the main terms. The lower order terms (for instance when the derivatives fall on either $\phi_0$ or $\sqrt{|\det g|}$ or from the $\rd \mathcal E_\lambda^{(\rd_t\mfg_3)}$ term in the wave equation) are bounded by $C(C_1)\ep^2\lambda^4$ using \eqref{B2}, \eqref{BA2}, \eqref{BA3}, \eqref{BA4}{, \eqref{g2.Linfty}, \eqref{g3.Linfty}} and the {estimates for} the background solution\footnote{In particular, these lower order terms can be bounded with a factor of $C(C_1)\ep$ from $\phi_0$ (and its derivatives), a factor of $C(C_1)\ep\lambda^2$ from $g_3'$ and a factor of $C(C_1)\lambda^2$ from $\rd \mathcal E$.} {in Corollary~\ref{lwp.small}}.
\begin{equation*}
\begin{split}
&\sup_{t\in [0,1]}\left|\int_0^t \int_{\Sigma_{t'}} (\rd_t \mathcal E_\lambda^{(\rd_t\mfg_3)}) (\rd_t (g'_3)^{t\beta})(\rd_\beta \phi_0)(t',x)\, \sqrt{|\det g|}\, dx\, dt'\right|\\
\leq &\sup_{t\in [0,1]}\left|\int_0^t \int_{\Sigma_{t'}} (\rd_t^2 \mathcal E_\lambda^{(\rd_t\mfg_3)})  (g'_3)^{t\beta}(\rd_\beta \phi_0)(t',x)\, \sqrt{|\det g|}\, dx\, dt'\right|\\
&+2\sup_{t\in [0,1]}\left|\int_{\Sigma_{t}} (\rd_t \mathcal E_\lambda^{(\rd_t\mfg_3)}) (g'_3)^{t\beta}(\rd_\beta \phi_0)(t,x) \, \sqrt{|\det g|}\, dx\right|+C(C_1)\ep^2\lambda^4\\
\leq &\sup_{t\in [0,1]}\left|\int_0^t \int_{\Sigma_{t'}} \f{1}{\gi^{tt}}\left(2\gi^{ti}\rd^2_{it} \mathcal E_\lambda^{(\rd_t\mfg_3)}+\gi^{ij}\rd^2_{ij}\mathcal E_\lambda^{(\rd_t\mfg_3)}+ h \right)  (g'_3)^{t\beta}(\rd_\beta \phi_0)(t',x)\, \sqrt{|\det g|}\, dx\, dt'\right|\\
&+2\sup_{t\in [0,1]}\left|\int_{\Sigma_{t}} (\rd_t \mathcal E_\lambda^{(\rd_t\mfg_3)}) (g'_3)^{t\beta}(\rd_\beta \phi_0)(t,x) \, \sqrt{|\det g|}\, dx\right|+C(C_1)\ep^2\lambda^4\\
\leq &C(C_0)\ep\|\rd \mathcal E_\lambda^{(\rd_t\mfg_3)}\|_{L^2}\left(\|\nab \mfg_3\|_{L^2}+\|\mfg_3\|_{L^2}\right)+C(C_1)\ep^2\lambda^4\\
\leq &C(C_1)\ep^2\lambda^4,
\end{split}
\end{equation*}
where in the last line we have again used \eqref{BA2} and \eqref{BA4}.

Now, $h_2$ can be bounded in a similar fashion By \eqref{EE.prop.main} in Proposition \ref{EE.prop}, since the initial data vanish, we have
$$\|\rd \mathcal E_\lambda^{(\rd_t\mfg_3)}\|_{L^2}^2 \leq C(C_1)\ep^2\lambda^4,$$
which then implies after taking square root that
\begin{equation}\label{E2.est.1}
\|\rd \mathcal E_\lambda^{(\rd_t\mfg_3)}\|_{L^2} \leq C(C_1)\ep\lambda^2.
\end{equation}

{\bf Estimates for higher derivatives.} By the assumptions on the background solution and the bootstrap assumption \eqref{BA4}, we have
$$\sum_{k\leq 2}\|h\|_{H^k}\leq C(C_1)\ep\lambda,\quad \|h\|_{H^3}\leq C(C_1)\ep.$$
Since the initial data vanish, applying Corollary \ref{EE.cor} inductively in $k$ and using \eqref{E2.est.1}, we obtain the bound
\begin{equation}\label{E2.est.2}
\sum_{1\leq k\leq 3} \lambda^k\|\rd \mathcal E_\lambda^{(\rd_t\mfg_3)}\|_{H^k}+\sum_{k\leq 2} \lambda^{k+1}\|\rd^2 \mathcal E_\lambda^{(\rd_t\mfg_3)}\|_{H^k}\leq C(C_1)\ep \lambda^2.
\end{equation}
This concludes the proof of the proposition.
\end{proof}

Finally, we prove the following bounds for $\mathcal E_\lambda^{(Error)}$:
\begin{proposition}\label{E.error.est}
{$\mathcal E_{\lambda}^{(Error)}$ satisfies}
$$\sum_{k\leq 3}\lambda^k \|\rd \mathcal E_\lambda^{(Error)}\|_{H^k}+\sum_{k\leq 2}\lambda^{k+1} \|\rd^2 \mathcal E_\lambda^{(Error)}\|_{H^k}\leq C(C_0)\ep\lambda^2.$$
\end{proposition}
\begin{proof}
This is an immediate consequence of Proposition \ref{EE.prop} (more precisely, \eqref{EE.prop.2}), Corollary \ref{EE.cor} and Proposition \ref{est.E3.inho}.
\end{proof}

Finally, we combine all the estimates in this subsection in the following corollary:
\begin{cor}\label{E.cor}
Each of $\mathcal E_\lambda^{(elliptic)}$, $\mathcal E_\lambda^{(\rd_t\mfg_3)}$ and $\mathcal E_\lambda^{(Error)}$ satisfies
$$\sum_{k\leq 3}\lambda^k \|\rd \mathcal E_\lambda^{(\cdot)}\|_{H^k}+\sum_{k\leq 2}\lambda^{k+1} \|\rd^2 \mathcal E_\lambda^{(\cdot)}\|_{H^k}\leq C(C_1)\ep\lambda^2.$$
Consequently, we also have
$$\sum_{k\leq 3}\lambda^k \|\rd \mathcal E_\lambda\|_{H^k}+\sum_{k\leq 2}\lambda^{k+1} \|\rd^2 \mathcal E_\lambda\|_{H^k}\leq C(C_1)\ep\lambda^2.$$
\end{cor}
\begin{proof}
This is a consequence of {Propositions}~\ref{E1.prop}, \ref{E.rdtg3.est} and \ref{E.error.est}.
\end{proof}
In particular, we have improved the bootstrap assumption \eqref{BA2}.

\section{Estimates for the metric components}\label{secelliptic}

\subsection{Calculations of the main terms}\label{sec.g.main.terms}

 To begin, we consider the difference of the squares of the derivatives of the scalar field and collect the $O_\ep(1)$ terms as well as the $O_\ep(\lambda)$ terms:
\begin{proposition}\label{phi.phi.diff}
The difference $(\partial_\mu \phi)(\rd_\nu\phi)- (\partial_\mu \phi_0)(\rd_\nu \phi_0) $ can be expanded as follows:
\begin{equation}\label{g.expand}
\begin{split}
(\partial_\mu \phi)(\rd_\nu\phi)- (\partial_\mu \phi_0)(\rd_\nu \phi_0) 
=& \Theta_{\mu\nu}^{(main)}+\Theta_{\mu\nu}^{(1)} + \Theta_{\mu\nu}^{(2)} + \Theta_{\mu\nu}^{(remainder)},
\end{split}
\end{equation}
where
\begin{equation}\label{g.main.expand}
\Theta_{\mu\nu}^{(main)}:=\f 12\sum_{\bA} F_\bA^2 (\partial_\mu u_\bA)( \partial_\nu u_\bA ),
\end{equation}
and
\begin{equation}\label{g.O1.expand}
\begin{split}
\Theta_{\mu\nu}^{(1)}
:=&\underbrace{-2\sum_{\bA}F_\bA \partial_{(\mu} \phi_0\partial_{\nu)} u_\bA \si{\frac{u_\bA}{ \lambda}}}_{\mathfrak G_{1,1}}+ \underbrace{ \f 12 \sum_{\bA} F_\bA^2 \partial_{{\mu}} u_\bA \partial_{{\nu}} u_\bA \co{\frac{2u_\bA}{ \lambda}}}_{\mathfrak G_{1,2}}\\
&+ \underbrace{ \f 12\sum_{\pm}\sum_{\bA}\sum_{\bB \neq \bA}(\mp 1)\cdot F_\bA F_\bB {\rd_\mu u_{\bA} \rd_\nu u_{\bB} }\co{\frac{u_\bA {\pm u_\bB}}{ \lambda} }}_{\mathfrak G_{1,3}},
\end{split}
\end{equation}
and
\begin{equation}\label{g.Ol.expand}
\begin{split}
\Theta_{\mu\nu}^{(2)}
:=&\underbrace{- \lambda\sum_{\pm}\sum_{\bA}\sum_{\bB \neq \bA} \rd_{(\mu} F_{\bA} \rd_{\nu)} u_\bB F_\bB \si{\frac{u_\bA\pm u_\bB}{ \lambda}}}_{=:\mathfrak G_{2,1}}\\
&\underbrace{- \lambda\sum_{\bA} \rd_{(\mu} F_{\bA} \rd_{\nu)} u_\bA F_\bA \si{\frac{2 u_\bA}{ \lambda}}}_{=:\mathfrak G_{2,2}}\\
&{+} \underbrace{2\lambda \sum_{\bA}   \partial_{(\mu} \phi_0\partial_{\nu)} u_\bA \left(\wht F_\bA\co{\frac{u_\bA}{ \lambda}}-2 \wht F^{(2)}_\bA\si{\f{2u_\bA}{\lambda}}+3\wht F^{(3)}_\bA\co{\f{3u_\bA}{\lambda}}\right)}_{=:\mathfrak G_{2,3}}\\
&+\underbrace{2\lambda \sum_{\bA} \rd_{(\mu}\phi_0 \rd_{\nu)} F_\bA\co{\frac{u_\bA}{ \lambda}}}_{=:\mathfrak G_{2,4}}\\
&\underbrace{-2\lambda \sum_{\bA,\bB} \rd_{{(\mu}} u_\bA\rd_{{\nu)}} u_\bB F_\bA \si{\f{u_\bA}{\lambda}}\left(\wht F_\bB \co{\f{u_\bB}{\lambda}}-2\wht F^{(2)}_\bB\si{\f{2u_\bB}{\lambda}}+3\wht F^{(3)}_\bB\co{\f{3 u_\bB}{\lambda}}\right)}_{=:\mathfrak G_{2,5}}.
\end{split}
\end{equation}
and the remainder term satisfies
\begin{equation}\label{remainder.est}
\sum_{k\leq 3}\lambda^k\|\Theta^{(remainder)}\|_{H^k}+\sum_{k\leq 2}\lambda^{k+1}\|\rd_t\Theta^{(remainder)}\|_{H^k} \leq C(C_1) \ep^2 \lambda^2.
\end{equation}
\end{proposition}

\begin{proof}
This is a direct computation using the parametrix \eqref{phi.para} for $\phi$ and estimating the resulting terms with \eqref{B1}, \eqref{BA1} and Corollary \ref{E.cor}. {\bf In order to simplify the exposition, let us only check the bound for $\sum_{k\leq 3}\lambda^k\|\Theta^{(remainder)}\|_{H^k}$ in \eqref{remainder.est} in the proof. The estimate for the time derivative can be verified in a completely identical manner.} We now compute
\begin{equation}\label{g.O1.expand.0}
\begin{split}
&\mbox{Cross terms between $\phi_0$ and $\displaystyle\sum_\bA \lambda F_\bA \co{\f{u_\bA}{\lambda}}$}\\
=&\underbrace{-2\sum_{\bA}F_\bA \partial_{(\mu} \phi_0\partial_{\nu)} u_\bA \si{\frac{u_\bA}{ \lambda}}}_{\mathfrak G_{1,1}}+\underbrace{2\lambda \sum_{\bA} \rd_{(\mu}\phi_0 \rd_{\nu)} F_\bA\co{\frac{u_\bA}{ \lambda}}}_{\mathfrak G_{2,4}}.
\end{split}
\end{equation}

\begin{equation}\label{g.O1.expand.1}
\begin{split}
&\mbox{The term $\rd_\mu\left(\displaystyle\sum_\bA \lambda F_\bA \co{\f{u_\bA}{\lambda}}\right)\rd_\nu\left(\displaystyle\sum_\bB \lambda F_\bB \co{\f{u_\bB}{\lambda}}\right)$}\\
=&\sum_{\bA,\bB}F_\bA F_\bB \partial_\mu u_\bA \partial_\nu u_\bB \si{\frac{u_\bA}{ \lambda}} \si{\frac{u_\bB}{\lambda}} - 2\lambda \sum_{\bA,\bB} \rd_{(\mu} F_{\bA} \rd_{\nu)} u_\bB F_\bB\co{\f{u_\bA}{\lambda}}\si{\frac{u_\bB}{\lambda}} {+ \dots}\\
=& \underbrace{\f 12\sum_{\bA} F_\bA^2 (\partial_\mu u_\bA)( \partial_\nu u_\bA )}_{\Theta^{(main)}}
+ \underbrace{ \f 12 \sum_{\bA} F_\bA^2 \partial_{{\mu}} u_\bA \partial_{{\nu}} u_\bA \co{\frac{2u_\bA}{ \lambda}}}_{\mathfrak G_{1,2}} \\
&+ \underbrace{ \f 12\sum_{\pm}\sum_{\bA}\sum_{\bB \neq \bA}(\mp 1)\cdot F_\bA F_\bB {\rd_\mu u_{\bA} \rd_\nu u_{\bB} } \co{\frac{u_\bA{\pm u_{\bB}}}{ \lambda}}}_{\mathfrak G_{1,3}} \\
&\underbrace{- \lambda\sum_{\pm}\sum_{\bA}\sum_{\bB \neq \bA} {(\pm 1)} \rd_{(\mu} F_{\bA} \rd_{\nu)} u_\bB F_\bB \si{\frac{u_\bA\pm u_\bB}{ \lambda}}}_{\mathfrak G_{2,1}}\underbrace{- \lambda\sum_{\bA} \rd_{(\mu} F_{\bA} \rd_{\nu)} u_\bA F_\bA \si{\frac{2 u_\bA}{ \lambda}}}_{\mathfrak G_{2,2}}{ + \dots,}
\end{split}
\end{equation}
{where here, and below, $\dots$ denotes terms satisfying the bound \eqref{remainder.est}. To justify that the $\dots$ terms indeed satisfy the necessary bounds, we have used here Corollary~\ref{lwp.small}.}

{Next, we use Corollary~\ref{lwp.small}, \eqref{B1} and \eqref{BA1} to obtain}
\begin{equation}\label{g.Ol.expand.1}
\begin{split}
&\mbox{Cross terms between $\phi_0$ and $\displaystyle\sum_\bA \lambda^2 \wht F_\bA \si{\frac{u_\bA}{\lambda}}+\displaystyle\sum_\bA \lambda^2 \wht F_\bA^{(2)}\co{\frac{2u_\bA}{\lambda}}+\displaystyle\sum_\bA  \lambda^2\wht F_\bA^{(3)} \si{\frac{3u_\bA}{\lambda }}$}\\
=&\underbrace{2\lambda \sum_{\bA}   \partial_{(\mu} \phi_0\partial_{\nu)} u_\bA \left(\wht F_\bA\co{\frac{u_\bA}{ \lambda}}-2 \wht F^{(2)}_\bA\si{\f{2u_\bA}{\lambda}}+3\wht F^{(3)}_\bA\co{\f{3u_\bA}{\lambda}}\right)}_{\mathfrak G_{2,3}}+\dots
\end{split}
\end{equation}
{Using again Corollary~\ref{lwp.small}, \eqref{B1} and \eqref{BA1}, we have}
\begin{equation}\label{g.Ol.expand.2}
\begin{split}
&\mbox{Cross terms between $\displaystyle\sum_\bA \lambda F_\bA \co{\f{u_\bA}{\lambda}}$}\\
&\qquad\qquad\mbox{ and $\displaystyle\sum_{{\bB}} \lambda^2 \wht F_{{\bB}} \si{\frac{u_{{\bB}}}{\lambda}}+\displaystyle\sum_{{\bB}} \lambda^2 \wht F_{{\bB}}^{(2)}\co{\frac{2u_{{\bB}}}{\lambda}}+\displaystyle\sum_{{\bB}}  \lambda^2\wht F_{{\bB}}^{(3)} \si{\frac{3u_{{\bB}}}{\lambda }}$}\\
=&\underbrace{-2\lambda \sum_{\bA,\bB} \rd_{{\mu}} u_\bA\rd_{{\nu}} u_\bB F_\bA \si{\f{u_\bA}{\lambda}}\left(\wht F_\bB \co{\f{u_\bB}{\lambda}}-2\wht F^{(2)}_\bB\si{\f{2u_\bB}{\lambda}}+3\wht F^{(3)}_\bB\co{\f{3 u_\bB}{\lambda}}\right)}_{\mathfrak G_{2,5}}+\dots
\end{split}
\end{equation}

In {\eqref{g.O1.expand.0}}, \eqref{g.O1.expand.1}, \eqref{g.Ol.expand.1} and \eqref{g.Ol.expand.2}, we have seen all of the $\Theta_{\mu\nu}^{(main)}$, $\Theta_{\mu\nu}^{(1)}$ and $\Theta_{\mu\nu}^{(2)}$ terms. Therefore, it suffices to show that all remaining terms satisfy the estimate \eqref{remainder.est}.

For the cross terms between $\displaystyle\sum_\bA \lambda^2 \wht F_\bA \si{\frac{u_\bA}{\lambda}}$, $\displaystyle\sum_\bA \lambda^2 \wht F_\bA^{(2)}\co{\frac{2u_\bA}{\lambda}}$ and $\displaystyle\sum_\bA  \lambda^2\wht F_\bA^{(3)} \si{\frac{3u_\bA}{\lambda }}$, we argue as follows. For the $\wht F \cdot \wht F$ term, we have
\begin{equation}\label{whtF.whtF}
\begin{split}
&\sum_{k\leq 3}\lambda^k\|\cdot\|_{H^k}\\
\leq &\sum_{k\leq 3}\sum_{\bA,\bB}\lambda^{k+2} \left(\sum_{k_1+k_2=k}\lambda^{-k_1}\|(\wht F_\bA)( \wht F_\bB)\|_{H^{k_2}}+\sum_{k_1+k_2=k}\lambda^{-k_1+1}\|\wht F_\bA( \rd\wht F_\bB)\|_{H^{{k_2}}}\right.\\
&\left.\qquad\qquad+\sum_{k_1+k_2=k}\lambda^{-k_1+2}\|(\rd \wht F_\bA)( \rd\wht F_\bB)\|_{H^{{k_2}}}\right)\\
\leq &C(C_0)\sum_{k\leq 3}\sum_{\bA,\bB}\lambda^{k+2} \left(\|\wht F_\bA\|_{L^\infty}\sum_{k_1+k_2=k}\left(\lambda^{-k_1}\|\wht F_\bB\|_{H^{k_2}}+\lambda^{-k_1+1}\| \rd\wht F_\bB\|_{H^{k_2}}\right)\right.\\
&\left.\qquad\qquad +\|\rd \wht F_\bA\|_{L^\infty}\sum_{k_1+k_2=k}\left(\lambda^{-k_1+1}\|\wht F_\bB\|_{H^{k_2}}+\lambda^{-k_1+2}\| \rd\wht F_\bB\|_{H^{k_2}}\right)\right)\\
\leq &C(C_1)\ep\lambda^2,
\end{split}
\end{equation}
as desired. Here, in deriving \eqref{whtF.whtF}, we have used 
\begin{itemize}
\item {the fact that (by Corollary~\ref{lwp.small}) every derivative of the phase function gives $\lambda^{-1}$,}
\item the standard product estimate {in Proposition~\ref{product}},
\item the bootstrap assumption \eqref{BA1} to control the $H^k$ norms of $\wht F_\bA$ and $\rd \wht F_\bA$,
\item the {estimates in \eqref{F.Linfty}} to control the $L^\infty$ norms of $\wht F_\bA$ and $\rd \wht F_\bA$. 
\end{itemize}
Next, we note that in the estimate \eqref{whtF.whtF}, if we replace one (or both) of the $\wht F_\bA$ by ${\wht F}^{(a)}_\bA$ for $a=2,3$, the exact same estimate applies. (In fact, it is slightly better since according to \eqref{B1}, ${\wht F}^{(a)}_\bA$ satisfies better bounds.) Therefore, all the cross terms between $\displaystyle\sum_\bA \lambda^2 \wht F_\bA \si{\frac{u_\bA}{\lambda}}$, $\displaystyle\sum_\bA \lambda^2 \wht F_\bA^{(2)}\co{\frac{2u_\bA}{\lambda}}$ and $\displaystyle\sum_\bA  \lambda^2\wht F_\bA^{(3)} \si{\frac{3u_\bA}{\lambda }}$ can be dealt with similarly.

Finally, it remains to estimate the term $2\rd_{(\mu}\mathcal E_\lambda \rd_{\nu)}\phi_\lambda$. We now consider each term in the parametrix \eqref{phi.para} and bound them. Below, we will {often} use {the} product estimate {Proposition~\ref{product} and also the simpler
$$\|uv\|_{H^k}\leq \|u\|_{H^k} \|v\|_{C^k},$$}
as well as the Sobolev embedding {(Proposition~\ref{holder}, and especially the $H^2\subset L^\infty$ embedding).} {We will use these estimates without further comments.}

Using {Corollaries~\ref{lwp.small} and} \ref{E.cor}, we have
$$\sum_{k\leq 3}\lambda^k\|\rd \mathcal E_\lambda \rd\phi_0\|_{H^k}\leq C(C_0)\sum_{k\leq 3}\lambda^k\|\rd\mathcal E_\lambda\|_{H^k}\|\rd\phi_0\|_{H^5}\leq C(C_1)\ep^2\lambda^2.$$
Using {Corollaries~\ref{lwp.small} and} \ref{E.cor}, we have
\begin{equation*}
\begin{split}
&\sum_{k\leq 3}\sum_\bA\lambda^k \left\|(\rd\mathcal E_\lambda)\left(\rd \sum_\bA \lambda F_\bA \co{\f{u_\bA}{\lambda}}\right)\right\|_{H^k}\\
\leq & C(C_0)\sum_{k\leq 3}\sum_\bA\lambda^k
\sum_{k_1+k_2=k}\|\rd\mathcal E_\lambda\|_{H^{k_1}} \left\|\partial( \lambda F_\bA\co{\f{u_\bA}{\lambda}})\right\|_{{C}^{k_2}}\\
\leq &C(C_1)\ep^2 \sum_{k\leq 3}\lambda^k \sum_{k_1+k_2= k} \lambda^{2-k_1-k_2}\leq C(C_1)\ep^2\lambda^2.
\end{split}
\end{equation*}
Using the bootstrap assumption \eqref{BA1}, the estimates \eqref{F.Linfty}, \eqref{E.Linfty}, Corollaries~\ref{lwp.small} and \ref{E.cor}, we have
\begin{equation*}
\begin{split}
&{\sum_{k\leq 3}}\lambda^k \left\|(\rd\mathcal E_\lambda)\left(\rd \sum_\bA \lambda^2 \wht F_\bA \si{\f{u_\bA}{\lambda}}\right)\right\|_{H^k}\\
\leq & C(C_0)\sum_{k\leq 3}\sum_\bA\lambda^k\left(\|\rd\mathcal E_\lambda\|_{H^k}\left(\lambda^2\| \rd \wht F_{\bA}\|_{L^\infty}+\lambda\|\wht F_\bA\|_{L^\infty}\right)\right.\\
&\left. +\|\rd\mathcal E_\lambda\|_{L^\infty}\left(\sum_{k_1+k_2=k}\lambda^{{-k_1+2}}\|\rd \wht F_\bA\|_{H^{k_2}}+
\sum_{{k_1+k_2=k}}
\lambda^{-k_1+1}\| \wht F_\bA\|_{H^{k_2}}\right)
\right)\\
\leq & C(C_1)\sum_{k\leq 3}\lambda^k \left(\ep^2\lambda^{2-k}\cdot (\lambda^2\cdot\lambda^{-1}+\lambda)+\ep^2 {\lambda} \sum_{k_1+k_2=k}\lambda^{-k_1+2}\lambda^{\min\{-k_2+1,0\}}\right.\\
&\quad\quad \left.
+\ep^2 {\lambda} \sum_{{k_1+k_2=k}}
\lambda^{-k_1+1}\lambda^{\min\{-k_2+2,0\}}
\right)\\
\leq &C(C_1)\ep^2\lambda^2.
\end{split}
\end{equation*}

Next, we note that $\wht F_\bA^{(2)}$ and $\wht F_\bA^{(3)}$ both obey better bounds than $\wht F_\bA$ according to \eqref{B1} {and \eqref{Fa.Linfty}}. Therefore, we can argue similarly as above to get
\begin{equation*}
\begin{split}
&{\sum_{k\leq 3}}\lambda^k\left(\left\|(\rd\mathcal E_\lambda)\left(\rd \sum_\bA \lambda \wht F^{(2)}_\bA \co{\f{2u_\bA}{\lambda}}\right)\right\|_{H^k} + \left\|(\rd\mathcal E_\lambda)\left(\rd \sum_\bA \lambda \wht F^{(3)}_\bA \si{\f{3u_\bA}{\lambda}}\right)\right\|_{H^k}\right)\\
\leq &C(C_1)\ep^2\lambda^2.
\end{split}
\end{equation*}
Finally, using Corollary \ref{E.cor}, we obtain
$$
\sum_{k\leq 3}\lambda^k\|(\partial  \mathcal E_\lambda)^2\|_{H^k}\leq C(C_0)\sum_{k\leq 3}\lambda^k \|\partial \mathcal E_\lambda\|_{L^\infty}
\|\partial \mathcal E_\lambda\|_{H^k}\\
\leq C(C_0) \sum_{k\leq 3}\lambda^k\|\partial \mathcal E_\lambda\|_{H^2}
\|\partial \mathcal E_\lambda\|_{H^k}\leq C(C_1)\ep^2\lambda^{2}.$$
\end{proof}

The following proposition shows that $\mfg_1$ (defined in \eqref{g1.def}) is chosen so that $\Delta \mfg_1$ cancels with the term $\Theta^{(1)}$, at the expense of some additional $O(\lambda)$ terms.

\begin{proposition}\label{g1.prop}
For $\mfg_1$ as in \eqref{g1.def} and ${{\bf \Gamma}}_0(\mfg)$ as in \eqref{G.g}, \eqref{G.N}, \eqref{G.b1} and \eqref{G.b2} (with metric components replaced by their background values) and $\Theta^{(1)}$ as in \eqref{g.O1.expand}, we have that $\Delta \mfg_1- \Gamma_0(\mfg)^{\mu\nu}\Theta^{(1)}_{\mu\nu}$ depends only on the background solution with
\begin{equation*}
\begin{split}
&\Delta \mfg_1- {{\bf \Gamma}}_0(\mfg)^{\mu\nu}\Theta^{(1)}_{\mu\nu}\\
=&\lambda\sum_{\bA}\left(\mathcal G_{1,1,\bA}^{(\Delta)}(\mfg) \co{\frac{u_\bA}{\lambda}}+\mathcal G_{1,2,\bA}^{(\Delta)}(\mfg) \si{\frac{2u_\bA}{\lambda}}\right)\\
&+\lambda\sum_{\pm}\sum_{\bA}\sum_{\bB \neq \bA}\mathcal G_{2,1,\bA,\bB,\pm}^{(\Delta)}(\mfg)\si{\frac{ u_\bA \pm u_{\gra B}}{\lambda}}+\mathcal G_{error}^{(\Delta)}(\mfg)
\end{split}
\end{equation*}
{for some $\mathcal G_{1,1,\bA}^{(\Delta)}(\mfg)$, $\mathcal G_{1,2,\bA}^{(\Delta)}(\mfg)$, $\mathcal G_{2,1,\bA,\bB,\pm}^{(\Delta)}(\mfg)$ and $\mathcal G_{error}^{(\Delta)}(\mfg)$} such that (for any $\mfg\in \{N,\gamma,\beta^i\}$)
\begin{itemize}
\item $\mathcal G_{1,1,\bA}^{(\Delta)}(\mfg)$, $\mathcal G_{1,2,\bA}^{(\Delta)}(\mfg)$ and $\mathcal G_{2,1,\bA,\bB,\pm}^{(\Delta)}(\mfg)$ are all compactly supported in $B(0,R_{supp})$ in the time interval $[0,1]$ and obey the estimates
$$\|\cdot\|_{H^7\cap {C^{7}}}+\|\rd_t(\cdot)\|_{H^6\cap {C^{6}}}\leq C(C_0)\ep^2; $$
\item $\mathcal G_{error}^{(\Delta)}(\mfg)$ is also compactly supported in $B(0,R_{supp})$ in the time interval $[0,1]$ and obeys the estimate 
$$\sum_{k\leq 3}\lambda^k \|\mathcal G_{error}^{(\Delta)}(\mfg)\|_{H^k}+\sum_{k\leq 2}\lambda^{k+1} \|\rd_t\mathcal G_{error}^{(\Delta)}(\mfg)\|_{H^k}\leq C(C_0)\ep^2\lambda^2.$$
\end{itemize}

\end{proposition}
\begin{proof}
{Note that the definition of $\mfg_1$ in \eqref{g1.def} is such that if each derivative of $\Delta$ hits the oscillating factor, the resulting term cancels exactly ${\bf \Gamma}_0^{\mu\nu}(\mfg) \Theta^{(1)}_{\mu\nu}$.}

{In the case that exactly one derivative of $\Delta$ hits the oscillating factor, then the phases become $\si{\frac{2u_\bA}{\lambda}}$, $\co{\f{u_{\bA}}{\lambda}}$ and $\si{\f{u_{\bA}\pm u_{\bB}}{\lambda}}$. This then gives rise to the three main terms in the proposition, and one checks using the definition of $\mfg_1$ from \eqref{g1.def} that $\mathcal G_{1,1,\bA}^{(\Delta)}(\mfg)$, $\mathcal G_{1,2,\bA}^{(\Delta)}(\mfg)$, $\mathcal G_{2,1,\bA,\bB,\pm}^{(\Delta)}(\mfg)$ all depend only on the background solution, are compactly supported in $[0,1]\times B(0,R_{supp})$ and obey the stated bounds.
}

{
Finally, denote by $\mathcal G_{error}^{(\Delta)}(\mfg)$ all the terms that $\Delta$ does not hit on the oscillating factor. It is easy to check using \eqref{g1.def} that it has all the desired properties.
}

\end{proof}

\begin{proposition}\label{theta2.prop}
{For ${\bf \Gamma}^{\mu\nu}_0(\mfg)$ as in \eqref{G.g}, \eqref{G.N}, \eqref{G.b1} and \eqref{G.b2} (with metric components replaced by their background values) and $\Theta^{(2)}$ as in \eqref{g.Ol.expand},} ${{\bf \Gamma}}^{\mu\nu}_0{(\mfg)}\Theta^{(2)}_{\mu\nu}$ can be written as 
\begin{equation*}
\begin{split}
{{\bf \Gamma}}^{\mu\nu}_0{(\mfg)} \Theta^{(2)}_{\mu\nu}= &\sum_{\bA}\lambda \mathcal G_{1,1,\bA}^{(\Theta)}(\mfg)\co{\frac{u_\bA}{ \lambda}}+\sum_{\bA}\lambda \mathcal G_{1,2,\bA}^{(\Theta)}(\mfg) \si{\frac{2u_\bA}{\lambda}}\\
&+\sum_{\bA}\lambda\mathcal G_{1,3,\bA}^{(\Theta)}(\mfg) \co{\frac{3u_\bA}{\lambda}}+\sum_{\pm}\sum_{\bA}\sum_{\bB\neq \bA}\lambda \mathcal G_{2,1,\bA,\bB,\pm}^{(\Theta)}(\mfg)\si{\frac{ u_\bA \pm u_{\gra B}}{\lambda}}\\
&+\sum_{\pm}\sum_{\bA}\sum_{\bB\neq \bA} \lambda \mathcal G_{2,2,\bA,\bB,\pm}^{(\Theta)}(\mfg)\co{\frac{ u_\bA \pm 2 u_{\gra B}}{\lambda}}\\
&+\sum_{\pm}\sum_{\bA,\bB} \lambda \mathcal G_{2,3,\bA,\bB,\pm}^{(\Theta)}(\mfg) \si{\f{u_\bA\pm 3u_\bB}{\lambda}},
\end{split}
\end{equation*}
{for some}
$\mathcal G_{1,1,\bA}^{(\Theta)}(\mfg)$, $\mathcal G_{1,2,\bA}^{(\Theta)}(\mfg)$, $\mathcal G_{1,3,\bA}^{(\Theta)}(\mfg)$, $\mathcal G_{2,1,\bA,\bB,\pm}^{(\Theta)}(\mfg)$, $\mathcal G_{2,2,\bA,\bB,\pm}^{(\Theta)}(\mfg)$, $\mathcal G_{2,3,\bA,\bB,\pm}^{(\Theta)}(\mfg)$ {where each of them is} compactly supported in $B(0,R_{supp})$ and satisfies the estimate
$$\sum_{k\leq 3} \lambda^k \|\cdot\|_{H^{2+k}}+ \lambda^k \|\partial_t (\cdot)\|_{H^{k+1}} \leq C(C_1)\ep^2.$$
\end{proposition}

\begin{proof}
{The factor ${\bf \Gamma}^{\mu\nu}_0(\mfg)$, which depends only on the background metric, clearly obeys much better bound than we need according to Corollary~\ref{lwp.small}. It therefore suffices to control $\Theta^{(2)}_{\mu\nu}$ from \eqref{g.Ol.expand}. According to \eqref{g.Ol.expand}, $\Theta^{(2)}_{\mu\nu}$ indeed can be expanded into a sum of terms, each with a high frequency oscillating factor as in the statement of the proposition. We thus define the $\mathcal G^{(\Theta)}$ terms accordingly.} {Most of these $\mathcal G^{(\Theta)}$ terms depend only on the background, and they obey the desired bounds according to Corollary~\ref{lwp.small} and \eqref{B1}.} The only $\mathcal G^{(\Theta)}$ {terms}
which do not depend only on the background are of the form $F_\bA \wht F_\bB$ or
$\partial \phi_0 \wht F_{\bA}${: to show that they obey the desired estimates, we use \eqref{BA1}, in addition to Corollary~\ref{lwp.small} and \eqref{B1}. We remark that it is exactly the application of \eqref{BA1} that limits the regularity of $\mathcal G^{(\Theta)}$.}
\end{proof}

\begin{proposition}\label{gamma.prop}
$$\sum_{k\leq 3}\lambda^k\|\left({{\bf \Gamma}}^{\mu\nu}(\mfg)-{{\bf \Gamma}}_0^{\mu\nu}(\mfg)\right)\rd_\mu\phi\rd_\nu\phi\|_{H^k}+\sum_{k\leq 2}\lambda^{k+1}\|\rd_t\left(\left({{\bf \Gamma}}^{\mu\nu}(\mfg)-{{\bf \Gamma}}_0^{\mu\nu}(\mfg)\right)\rd_\mu\phi\rd_\nu\phi\right)\|_{H^k}\leq C(C_1)\ep^2\lambda^2.$$
\end{proposition}
\begin{proof}
Using the compact support of $\phi$, {Proposition~\ref{product}} give{s}
\begin{equation*}
\begin{split}
&\sum_{k\leq 3}\lambda^k\|\left({{\bf\Gamma}}^{\mu\nu}(\mfg)-\Gamma_0^{\mu\nu}(\mfg)\right)\rd_\mu\phi\rd_\nu\phi\|_{H^k}\\
\leq &C(C_0)\left(\|\rd\phi\|_{L^\infty}^2\sum_{k\leq 3}\lambda^k\|{{\bf\Gamma}}^{\mu\nu}(\mfg)-{{\bf\Gamma}}_0^{\mu\nu}(\mfg)\|_{H^k(B(0,R_{supp}{+1}))}\right.\\
&\left.\qquad +\|{{\bf\Gamma}}^{\mu\nu}(\mfg)-{{\bf\Gamma}}_0^{\mu\nu}(\mfg)\|_{L^\infty(B(0,R_{supp}{+1}))}\|\rd\phi\|_{L^\infty}\sum_{k\leq 3}\lambda^k\|\rd\phi\|_{H^k}\right)\\
\leq &C(C_1)\left(\ep^2\cdot \ep\lambda^2 + \ep\lambda^2\cdot \ep \cdot \ep\right)\leq C(C_1)\ep^2\lambda^2,
\end{split}
\end{equation*}
where in the last step we have used the estimates in \eqref{B1}, \eqref{B2}, \eqref{BA1}, \eqref{BA2}, \eqref{BA3}{,} \eqref{BA4}{, \eqref{g2.Linfty}, \eqref{g3.Linfty}} and Corollary \ref{E.cor}.
\end{proof} 

\subsection{Definition of $\mfg_2$}\label{sec.g2.def}

Our goal in this subsection is to show that $\Delta(\mfg-\mfg_0-\mfg_1-\mfg_2)$ is appropriately small. This of course involve a suitable choice of $\mfg_2$. Recall from \eqref{g2.def} that we have given an expression for $\mfg_2$, but the functions $\mathcal G_{1,1,\bA}(\mfg)$, $\mathcal G_{1,2,\bA}(\mfg)$, $\mathcal G_{1,3,\bA}(\mfg)$, $\mathcal G_{2,1,\bA,\bB,\pm}(\mfg)$, $\mathcal G_{2,2,\bA,\bB,\pm}(\mfg)$ and $\mathcal G_{2,3,\bA,\bB,\pm}(\mfg)$ are yet to be defined. They will be defined in this subsection. Together with the calculations from the previous subsection, we then obtain a good estimate for $\Delta(\mfg-\mfg_0-\mfg_1-\mfg_2)$.

Before we define $\mfg_2$, we need another piece of notation. Let (compare \eqref{up1}, \eqref{up2} and \eqref{up3})
\begin{align}
\Upsilon_1(\gamma):=&-\frac{e^{2\gamma_0}}{4N_0^2}\delta^{ij}\delta^{k\ell}(L\beta_1)_{ik}(L\beta_0)_{{j}\ell},\label{up1.g}\\
\Upsilon_1(N):=&\f{e^{2\gamma_0}}{2N_0}\delta^{ij}\delta^{k\ell}(L\beta_1)_{ik}(L\beta_0)_{{j}\ell},\label{up1.N}\\
\Upsilon_1(\beta^i):=&{\delta^{jk}}\delta^{{i}\ell}\left((\rd_k\log N_0)(L\beta_1)_{{j}\ell}+(\rd_k\log N_1)(L\beta_0)_{{j}\ell}-2(\rd_k\gamma_0)(L\beta_1)_{{j}\ell}-2(\rd_k\gamma_1)(L\beta_0)_{{j}\ell}\right)\label{up1.b}.
\end{align}
The terms defined above are meant to be the ``main term'' in $\Upsilon-\Upsilon_0$. {We will show that i}t also {can be decomposed into the ``right type'' of high frequency terms} so that we can define $\mfg_2$ in the parametrix appropriately to ``remove this main term'' in $\Delta(\mfg-\mfg_0-\mfg_1)$. This is made precise in the next two propositions.

\begin{proposition}\label{up1.prop}
$\Upsilon_1(\mfg)$ can be written as 
\begin{equation*}
\begin{split}
\Upsilon_1 (\mfg)=&\lambda\sum_{\bA}\left(\mathcal G_{1,1,\bA}^{(\Upsilon)}(\mfg) \co{\frac{u_\bA}{\lambda}}+\mathcal G_{1,2,\bA}^{(\Upsilon)}(\mfg) \si{\frac{2u_\bA}{\lambda}}\right)\\
&+\lambda\sum_{\pm}\sum_{\bA}\sum_{\bB \neq \bA}\mathcal G_{2,1,\bA,\bB,\pm}^{(\Upsilon)}(\mfg)\si{\frac{ u_\bA \pm u_{\gra B}}{\lambda}} + \mathcal G_{error}^{(\Upsilon)}(\mfg),
\end{split}
\end{equation*}
for some $\mathcal G_{1,1,\bA}^{(\Upsilon)}(\mfg)$, $\mathcal G_{1,2,\bA}^{(\Upsilon)}(\mfg)$ and $\mathcal G_{2,1,\bA,\bB,\pm}^{(\Upsilon)}(\mfg)$ {which} are all compactly supported in $B(0,R_{supp})$ in the time interval $[0,1]$ and obey the estimates
$$\|\cdot\|_{H^7\cap {C^{7}}}+\|\rd_t(\cdot)\|_{H^6\cap {C^{6}}}\leq C(C_0)\ep^2, $$
and some $\mathcal G_{error}^{({\Upsilon})}(\mfg)$  also compactly supported in $B(0,R_{supp})$ in the time interval $[0,1]$ and which obeys the estimate 
$$\sum_{k\leq 3}\lambda^k \|\mathcal G_{error}^{(\Upsilon)}(\mfg)\|_{H^k}+\sum_{k\leq 2}\lambda^{k+1} \|\rd_t\mathcal G_{error}^{(\Upsilon)}(\mfg)\|_{H^k}\leq C(C_0)\ep^2\lambda^2.$$

\end{proposition}
\begin{proof}
This follows immediately from the definition of $\mfg_1$ in \eqref{g1.def}, the definition of $\Upsilon_1$ in \eqref{up1.g}, \eqref{up1.N} and \eqref{up1.b} and the estimates for the background solution in Corollary~\ref{lwp.small}.
\end{proof}

\begin{proposition}\label{upsilon.prop}
$$\sum_{k\leq 3}\lambda^k \|\Upsilon(\mfg)-\Upsilon_0(\mfg)-\Upsilon_1(\mfg)\|_{H^k_{\delta+2}}{+\sum_{k\leq 2}\lambda^{k+1} \|\rd_t\left(\Upsilon(\mfg)-\Upsilon_0(\mfg)-\Upsilon_1(\mfg)\right)\|_{H^k_{\delta+2}}}\leq C(C_1)\ep^2\lambda^2.$$
\end{proposition}
\begin{proof}
{\textbf{Structure of the terms.}
First, $\Upsilon(\mfg)$, $\Upsilon_0(\mfg)$ and $\Upsilon_1(\mfg)$ are defined such that $\Upsilon(\mfg)-\Upsilon_0(\mfg)-\Upsilon_1(\mfg)$ can be written as a sum of terms of the following types:
\begin{enumerate}[(I)]
\item a term with a factor of $\mfg - \mfg_0$ (without derivative) multiplied by factors depending on $\nab \mfg$,
\item a term with a quadratic factor $(\nab \mfg_1 + \nab \mfg_2 + \nab \mfg_3)\cdot (\nab \mfg_1 + \nab \mfg_2 + \nab \mfg_3)$, multiplied by factors depending on the background metric $\mfg_0$,
\item a term with a factor of $\nab \mfg_2$ or a factor of $\nab \mfg_3$, multiplied by factors depending on $\mfg_0$ and $\nab\mfg_0$.
\end{enumerate}
The key point here is that there are no terms which are linear in $\nab \mfg_1$ (and, say, multiplied by background quantities). Such terms would give estimates which are $O_\ep(\lambda)$ instead of $O_\ep(\lambda^2)$. To illustrate this structure, consider $\mfg = \gamma$ (the other cases are similar):
\begin{equation*}
\begin{split}
& \Upsilon(\gamma) - \Upsilon_0(\gamma) - \Upsilon_1(\gamma)\\
=& \left(-\f{e^{2\gamma}}{8N^2} + \f{e^{2\gamma_0}}{8N_0^2}\right) |L\beta|^2 + \f{e^{2\gamma_0}}{8N_0^2} \left(\de^{ij}\de^{k\ell} (L\beta)_{ik}(L\beta)_{j\ell} - \de^{ij}\de^{k\ell} (L\beta_0)_{ik}(L\beta_0)_{j\ell} - 2\de^{ij}\de^{k\ell} (L\beta_1)_{ik} (L\beta_0)_{j\ell} \right)\\
=& \underbrace{\left(-\f{e^{2\gamma}}{8N^2} + \f{e^{2\gamma_0}}{8N_0^2}\right) |L\beta|^2}_{=:I} + \f{e^{2\gamma_0}}{8N_0^2} \left(\underbrace{|L\beta-L\beta_0|^2}_{=:II} + \underbrace{2 \de^{ij}\de^{k\ell}(L\beta-L\beta_0-L\beta_1)_{ik}(L\beta_0)_{j\ell}}_{=:III}\right).
\end{split}
\end{equation*}
}
{Note that the terms $I$, $II$ and $III$ exactly satisfy the conditions $I$, $II$ and $III$ above.}

{\textbf{Decay at infinity.}} We consider first the issue of decay at infinity. {Let us first consider the product of the spatial derivative of $\mfg$.} This concern{s} only {products} of the form {$\nabla \mfg_0\nabla \mfg_3$ and $\nabla \mfg_3\nabla \mfg_3$}, since the other terms contain at least a factor which is compactly supported. {Now, $\mfg_0 = \chi(|x|)\log(|x|)(\mfg_0)_{asymp} + \wht \mfg_0$ and $\mfg_3=\chi(|x|)\log(|x|)(\mfg_3)_{asymp}+\wht \mfg_3$, where $(\mfg_0)_{asymp}$ and $(\mfg_3)_{asymp}$ are independent of $x$ (but potentially $t$-dependent) and the terms $\wht \mfg_0$ and $\wht \mfg_3$ decay better (cf. Corollary~\ref{lwp.small} and \eqref{BA4}).} 
{There are three types of products:}
\begin{itemize}
\item The product {$\nabla \wht \mfg_3\nabla \wht \mfg_3$ or $\nabla \wht \mfg_0\nabla \wht \mfg_3$} has more decay than necessary to be in $H^k_{\delta+2}$ (see Proposition \ref{produit}).
\item The product {of the form $\nabla (\chi({|x|})\log({|x|})(\mfg_3)_{asymp})\nabla (\chi({|x|})\log({|x|})(\mfg_3)_{asymp})$ (or similarly the product $\nabla (\chi({|x|})\log({|x|})(\mfg_0)_{asymp})\nabla (\chi({|x|})\log({|x|})(\mfg_3)_{asymp})$)} would not have enough decay, but such terms are \underline{absent}, since all the terms {in \eqref{up1}, \eqref{up2}, \eqref{up3}, \eqref{up1.g}, \eqref{up1.N}, \eqref{up1.b}} contain a factor of the form $\nabla \beta_3$ {and $(\beta_3)_{asymp}=0$ by \eqref{g.asymp.sign.2}}. 
\item It remains only to consider the {``cross terms''}
$$\partial_k {\left((N_3)_{asymp}\chi(|x|)\log(|x|)\right)}(L \beta_3)_{i{\ell}} \quad \partial_k \left((\gamma_3)_{asymp} \chi({|x|})\log({|x|})\right)(L \beta_3)_{i{\ell}}.$$
Note that the factors $ \partial_k {\left((N_3)_{asymp}\chi(|x|)\log(|x|)\right)}$ and $\partial_k {\left((\gamma_3)_{asymp}\chi(|x|)\log(|x|)\right)}$ are $O\left(\frac{1}{{|x|}}\right)$ as ${|x|}$ tend to infinity, so the product with $L\beta$ have the right decay {i.e., belongs to $H^k_{\de+2}$} (see Proposition \ref{produit2}).
\end{itemize}
{Finally, let us note that the terms on the RHS of \eqref{up1}, \eqref{up2}, \eqref{up3}, \eqref{up1.g}, \eqref{up1.N}, \eqref{up1.b} are not just products of derivatives of $\mfg$, but has factors of $e^{2\gamma}$ or $\f{1}{N}$. Nevertheless, by \eqref{g.asymp.sign.1}, \eqref{g.asymp.sign.2}, $\gamma_{asymp}\leq 0$, $N_c=1$, $N_{asymp}\geq 1$, hence all these factors are \underline{favorable} from the point of view of the weights. We will therefore suppress the discussion of them.}

{\textbf{Proving the estimates.}}
{We now prove the estimates taking into account the discussions above. For clarity of the exposition, we will only discuss the $L^2_{\de+2}$ estimates for terms of type $I$, $II$ and $III$ discussed above. The higher derivative estimates are similar: we note that
\begin{itemize}
\item All the norms we use are such that every additional spatial derivative ``costs'' at most an extra power of $\lambda^{-1}$.
\item We only use up to $H^2$ and $C^1$ norms for $\mfg_2$ and $H^2_\de$ and $C^1_{\de+1}$ norms $\wht \mfg_3$. Hence, according to \eqref{BA3}, \eqref{BA4}, \eqref{g2.Linfty} and \eqref{g3.Linfty}, we still have sufficient regularity to control these terms upon taking up to $3$ spatial derivatives.
\item When taking higher derivatives, there may be more nonlinear terms when the derivatives hit on $e^{2\gamma}$ or $\f{1}{N}$. Nevertheless, these terms are all easy to handle, since $\nab \mfg$ is bounded in $L^\infty$ (and the $e^{2\gamma}$ or $\f{1}{N}$ terms are favorable in terms of weights).
\item The above considerations also apply for up to one $\rd_t$ derivative, again because the additional $\rd_t$ derivative ``costs'' at most an extra power of $\lambda^{-1}$.
\end{itemize}
For terms with the structure as in $I$, we use Corollary~\ref{lwp.small}, \eqref{B2}, \eqref{BA3}, \eqref{BA4}, \eqref{g2.Linfty}, \eqref{g3.Linfty}, Lemma~\ref{der} and Proposition~\ref{holder} to obtain
\begin{equation*}
\begin{split}
\|I\|_{L^2_{\de+2}}\ls &\left(\| \mfg_1\|_{C^0_{\de+1}} + \| \mfg_2\|_{C^0_{\de+1}} + |(\mfg_3)_{asymp}| + \| \wht \mfg_3\|_{C^0_{\de+1}}\right) \left(\|\nab \wht \mfg\|_{L^4_{\de+\f 32}} \|\nab \wht \mfg\|_{L^4_{\de+\f 32}} + |\mfg_{asymp}| \|\nab \wht \mfg\|_{L^2_{\de+1}}\right)\\
\ls &\left(\| \mfg_1\|_{C^0_{\de+1}} + \| \mfg_2\|_{C^0_{\de+1}} + |(\mfg_3)_{asymp}| + \| \wht \mfg_3\|_{C^0_{\de+1}}\right) \left(\|\wht \mfg\|_{H^2_{\de}} \|\wht \mfg\|_{H^2_{\de}} + |\mfg_{asymp}| \|\wht \mfg\|_{H^1_{\de}}\right) \ls C(C_1)\ep \lambda^2.
\end{split}
\end{equation*}
For terms with the structure as in $II$, we use \eqref{B2}, \eqref{BA3}, \eqref{BA4}, Lemma~\ref{der} and Proposition~\ref{holder} to obtain\footnote{We also use the fact that $\mfg_1$ and $\mfg_2$ are compactly supported and we can put in arbitrary weights in the estimates.}
\begin{equation*}
\begin{split}
\|II\|_{L^2_{\de+2}}\ls & \left(\|\nab \mfg_1\|_{L^4_{\de+\f 32}} + \|\nab \mfg_2\|_{L^4_{\de+\f 32}} + |(\mfg_3)_{asymp}| + \|\nab \wht \mfg_3\|_{L^4_{\de + \f32}}\right)\\
&\quad \times\left(\|\nab \mfg_1\|_{L^4_{\de+\f 52}} + \|\nab \mfg_2\|_{L^4_{\de+\f 52}} + \|\nab \wht \mfg_3\|_{L^2_{\de+1}} + \|\nab \wht \mfg_3\|_{L^4_{\de+\f 32}} \right)\\
\ls & \left(\|\nab \mfg_1\|_{L^4_{\de+\f 32}} + \|\mfg_2\|_{H^2_{\de}} + |(\mfg_3)_{asymp}| + \|\wht \mfg_3\|_{H^2_{\de}} \right)\left(\|\nab \mfg_1\|_{L^4_{\de+\f {{5}}2}} + \|\mfg_2\|_{H^2_{\de+1}} + \|\wht \mfg_3\|_{H^2_{\de}} \right) \\
\ls & C(C_1)\ep^2\lambda^2,
\end{split}
\end{equation*}
where we have used
	$$\|\nab \mfg_1\|_{L^4_{\de+\f 32}}{+ \|\nab \mfg_1\|_{L^4_{\de+\f 52}}} \leq C(C_0)\ep\lambda,$$
which is a direct consequence of {\eqref{B2}, H\"older's inequality and the support properties of $\mfg_1$}.

For terms with the structure as in $III$, we use Corollary~\ref{lwp.small}, \eqref{BA3}, \eqref{BA4}, Lemma~\ref{der} and Proposition~\ref{holder} to obtain
\begin{equation*}
\begin{split}
 \|III\|_{L^2_{\de+2}}\ls  & \left( \|\nab \mfg_2\|_{L^2_{\de+1}} + |(\mfg_3)_{asymp}| + \|\nab \wht \mfg_3\|_{L^2_{\de + 1}}\right) \|\nab \wht \mfg_0\|_{C^0_{\de+2}} + \|\nab \wht \mfg_3\|_{L^2_{\de + 1}} |(\mfg_0)_{asymp}| \\
\ls & \left( \|\mfg_2\|_{H^1_{\de}} + |(\mfg_3)_{asymp}| + \|\wht \mfg_3\|_{H^1_{\de}}\right) \|\wht \mfg_0\|_{H^3_{\de}} + \| \wht \mfg_3\|_{H^1_{\de}} |(\mfg_0)_{asymp}| \ls C(C_1)\ep^2\lambda^2.
\end{split}
\end{equation*}
}{These estimates, and their higher derivative analogues, imply the conclusion of the proposition.
}
\end{proof}

We {now} define $\mfg_2$. {Recall that in \eqref{g2.def}, we have defined $\mfg_2$ modulo some functions on the RHS that we have not defined. Here, we define them using the decompositions in Propositions~\ref{g1.prop}, \ref{theta2.prop} and \ref{up1.prop}.}

\begin{df}\label{g2.def.RHS}
We define
\begin{align}
\mathcal G_{1,1,\bA}(\mfg):=&\mathcal G_{1,1,\bA}^{(\Delta)}(\mfg)+\mathcal G_{1,1,\bA}^{(\Theta)}(\mfg)+\mathcal G_{1,1,\bA}^{(\Upsilon)}(\mfg),\\
\mathcal G_{1,2,\bA}(\mfg):=&\mathcal G_{1,2,\bA}^{(\Delta)}(\mfg)+\mathcal G_{1,2,\bA}^{(\Theta)}(\mfg)+\mathcal G_{1,2,\bA}^{(\Upsilon)}(\mfg),\\
\mathcal G_{1,3,\bA}(\mfg):=&\mathcal G_{1,3,\bA}^{(\Theta)}(\mfg),\\
\mathcal G_{2,1,\bA,\bB,\pm}(\mfg):=&\mathcal G_{2,1,\bA,\bB,\pm}^{(\Delta)}(\mfg)+\mathcal G_{2,1,\bA,\bB,\pm}^{(\Theta)}(\mfg)+\mathcal G_{2,1,\bA,\bB,\pm}^{(\Upsilon)}(\mfg),\\
\mathcal G_{2,2,\bA,\bB,\pm}(\mfg):=&\mathcal G_{2,2,\bA,\bB,\pm}^{(\Theta)}(\mfg),\\
\mathcal G_{2,3,\bA,\bB,\pm}(\mfg):=&\mathcal G_{2,3,\bA,\bB,\pm}^{(\Theta)}(\mfg){,}
\end{align}
{where the terms on the RHS are as in Propositions~\ref{g1.prop}, \ref{theta2.prop} and \ref{up1.prop}.}
\end{df}
The choice of $\mfg_2$ gives the following:
\begin{proposition}\label{deltag2.prop}
$$ \sum_{k\leq 3} \lambda^k\|\Delta \mfg_2-\Gamma_0^{\mu\nu}(\mfg)\Theta^{(2)}_{\mu\nu}-\Upsilon_1(\mfg)+ \Delta\mfg_1-\Gamma_0^{\mu\nu}(\mfg)\Theta^{{(1)}}_{\mu\nu}\|_{H^k}{\leq} C(C_1)\ep^2 \lambda^2{,}$$
{and}
$${\sum_{k\leq 2} \lambda^{k+1}\|\rd_t(\Delta \mfg_2-\Gamma_0^{\mu\nu}(\mfg)\Theta^{(2)}_{\mu\nu}-\Upsilon_1(\mfg)+ \Delta\mfg_1-\Gamma_0^{\mu\nu}(\mfg)\Theta^{{(1)}}_{\mu\nu})\|_{H^k}{\leq} C(C_1)\ep^2 \lambda^2.}$$
\end{proposition}

\begin{proof}
{We will only prove the first estimate as the second estimate is similar.}
By \eqref{g2.def}, Definition~\ref{g2.def.RHS}, Propositions~\ref{g1.prop}, \ref{theta2.prop}, \ref{up1.prop}, {we note that $\mfg_2$ is defined so that when one expands $\Delta \mfg_2$, and when both derivatives hit on the highly oscillatory phase, the terms cancel with the main terms in Proposition~\ref{g1.prop} (i.e., the terms except for $\mathcal G_{error}^{(\Delta)}(\mfg)$), the term ${\bf\Gamma}^{\mu\nu}_0(\mfg)\Theta_{\mu\nu}^{(2)}$ and the main term{s} in $\Upsilon_1(\mfg)$ (i.e., the terms except for $\mathcal G_{error}^{(\Upsilon)}(\mfg)$)}. {Therefore, to estimate $\Delta \mfg_2-\Gamma_0^{\mu\nu}(\mfg)\Theta^{(2)}_{\mu\nu}-\Upsilon_1(\mfg)+\Delta\mfg_1-\Gamma_0^{\mu\nu}(\mfg)\Theta^{(2)}_{\mu\nu}$, it suffices to control the term $\mathcal G_{error}^{(\Delta)}(\mfg)$ from Proposition~\ref{g1.prop}{, the term $\mathcal G_{error}^{(\Upsilon)}(\mfg)$ in $\Upsilon_1(\mfg)$,} and the second derivatives of $\mfg_2$ where at least one derivative does not hit on the oscillating phase. Denoting by $\mathcal G$ any of $\mathcal G^{(\Delta)}$, $\mathcal G^{(\Theta)}$ or $\mathcal G^{(\Upsilon)}$, we have}
\begin{align*}
& \sum_{k\leq 3} \lambda^k\|\Delta \mfg_2-\Gamma_0^{\mu\nu}(\mfg)\Theta^{(2)}_{\mu\nu}-\Upsilon_1(\mfg)+\Delta\mfg_1-\Gamma_0^{\mu\nu}(\mfg)\Theta^{(2)}_{\mu\nu}\|_{H^k}\\
\lesssim &\sum_{k\leq 3} \lambda^k {(}\|\mathcal G_{error}^{(\Delta)}(\mfg)\|_{H^k} {+ \|\mathcal G_{error}^{(\Upsilon)}(\mfg)\|_{H^k})}+ \sum_{k\leq 3} \lambda^k\sum_{{\ell} \leq k} {\left(\lambda^{3+{\ell}-k} \|\mathcal G\|_{H^{2+{\ell}}}
+ \lambda^{2+{\ell}-k}\|\mathcal G\|_{H^{1+{\ell}}}\right)}\leq  C(C_1)\ep^2\lambda^2,
\end{align*}
where we have used the estimates in Propositions~\ref{g1.prop}, \ref{theta2.prop} and \ref{up1.prop} as well as bounds for the background solution in Corollary~\ref{lwp.small}.
\end{proof}

Finally, we conclude this subsection with the following estimate on $\Delta(\mfg-\mfg_0-\mfg_1-\mfg_2)$:
\begin{proposition}\label{g.para.est}
$$\sum_{k\leq 3}\lambda^k \|\Delta(\mfg-\mfg_0-\mfg_1-\mfg_2)\|_{H^k_{\delta+2}}+\sum_{k\leq 2}\lambda^{k+1} \|\rd_t\left(\Delta(\mfg-\mfg_0-\mfg_1-\mfg_2)\right)\|_{H^k_{\delta+2}}\leq C(C_1)\ep^2\lambda^2.$$
\end{proposition}
\begin{proof}
{We will only prove the first estimate as the second estimate is similar.} By \eqref{g.elliptic} and \eqref{g0.elliptic}, we have
\begin{equation*}
\begin{split}
\Delta (\mfg-\mfg_0)
=&\left({\bf \Gamma}(\mfg)^{\mu\nu}-{\bf \Gamma}_0(\mfg)^{\mu\nu}\right)\rd_\mu\phi\,\rd_\nu\phi +{\bf \Gamma}_0(\mfg)^{\mu\nu}\left(\rd_\mu\phi\,\rd_\nu\phi-\rd_\mu\phi_0\,\rd_\nu\phi_0\right)\\
&-\f 12 \sum_{\bA} F^2_{\bA}{\bf \Gamma}_0(\mfg)^{\mu\nu}(\rd_\mu u_{\bA})(\rd_\nu u_{\bA})+\Upsilon(\mfg)-\Upsilon_0(\mfg).
\end{split}
\end{equation*}
Using Proposition \ref{phi.phi.diff}, this implies
\begin{equation*}
\begin{split}
\Delta (\mfg-\mfg_0)
=&\left({\bf \Gamma}(\mfg)^{\mu\nu}-{\bf \Gamma}_0(\mfg)^{\mu\nu}\right)\rd_\mu\phi\,\rd_\nu\phi +{\bf \Gamma}_0(\mfg)^{\mu\nu}\left(\Theta^{(1)}_{\mu\nu}+\Theta^{(2)}_{\mu\nu}+\Theta^{(remainder)}_{\mu\nu}\right)+\Upsilon(\mfg)-\Upsilon_0(\mfg).
\end{split}
\end{equation*}
{By Corollary~\ref{lwp.small}, \eqref{remainder.est} and} Proposition \ref{gamma.prop}, {${\bf \Gamma}_0(\mfg)^{\mu\nu}\Theta^{(remainder)}_{\mu\nu}$ and} $\left({\bf \Gamma}(\mfg)^{\mu\nu}-{\bf \Gamma}_0(\mfg)^{\mu\nu}\right)\rd_\mu\phi\,\rd_\nu\phi$ satisfy the desired estimate. As usual, we denote {such terms} by {$\dots$}. {By Proposition~\ref{deltag2.prop},} we {then have}
\begin{equation*}
\begin{split}
\Delta (\mfg-\mfg_0-\mfg_1)
=&\Upsilon(\mfg)-\Upsilon_0(\mfg)+\Delta \mfg_2 -\Upsilon_1(\mfg) + \dots
\end{split}
\end{equation*}
We conclude thanks to Proposition~\ref{upsilon.prop}.
\end{proof}

\subsection{Estimates for $\mfg_1$ and $\mfg_2$}

Recall the definition of $\mfg_1$ in \eqref{g1.def}. By inspection and using the regularity of the background solution, we immediately have the following bounds, as is stated in \eqref{B2}.
\begin{proposition}\label{mfg1.prop}
$\mfg_1$ depends only on the background solution, is compactly supported in $B(0,R_{supp})$ and satisfies the estimate
$$\sum_{k\leq 8}\lambda^{k}\|\mfg_1\|_{H^k\cap {C^{k}}}+\sum_{k\leq 7}\lambda^{k+1}\|\rd_t\mfg_1\|_{H^k\cap {C^{k}}}+\sum_{k\leq 6}\lambda^{k+2}\|\rd_t^2\mfg_1\|_{H^k\cap {C^{k}}}\leq C(C_0)\ep^2\lambda^2.$$
\end{proposition} 

For $\mfg_2$, we again have by inspection that 
\begin{proposition}\label{g2.est}
$\mfg_2$ satisfies the following estimate:
$$
\sum_{k \leq 5}
\lambda^k\|\mfg_2\|_{H^k}+\sum_{k\leq 4}\lambda^{k+1} \|\partial_t \mfg_2\|_{H^k}\leq C(C_1) \ep^2 \lambda^3 .$$
\end{proposition}
\begin{proof}

{Using} the definition \eqref{g2.def} we have
$$\lambda^k \|\mfg_2\|_{H^k}\lesssim \lambda^k \sum_{{\ell}\leq k} \lambda^{3-k+{\ell}}\|\mathcal G\|_{H^{{\ell}}} \lesssim \ep^2 \lambda^3,$$
where we have {used} $\mathcal G$ {to denote} any of $\mathcal G^{\Delta}$, estimated thanks to Proposition \ref{g1.prop}, $\mathcal G^{\Theta}$, estimated by Proposition \ref{theta2.prop} and $\mathcal G^{\Upsilon}$, estimated by Proposition \ref{up1.prop}.
Notice that the regularity in this estimate is limited by that of $\wht F_\bA$. The estimate for $\partial_t \mfg_2$ is similar. 
\end{proof}
We note in particular that the estimate in {Proposition~\ref{g2.est}} improves the bootstrap assumption \eqref{BA3}.

\subsection{Estimates for the term $\mfg_3$}

Next, we obtain the estimates for $\mfg_3$. By Proposition \ref{g.para.est}, we already have a good estimate for $\Delta \mfg_3$ and $\Delta (\rd_t \mfg_3)$. Therefore, the following is a consequence of Corollary \ref{coro}:

\begin{cor}\label{g3.cor}
$$\mfg_3=(\mfg_3)_{asymp}(t)\chi({|x|})\log({|x|})+\wht{\mfg}_3,$$
where
$$|(\mfg_3)_{asymp}|+|\rd_t(\mfg_3)_{asymp}|+\sum_{k\leq 3}\lambda^k\|\wht{\mfg}_3\|_{H^{2+k}_{{\delta}}}+\sum_{k\leq 2}\lambda^{k+1}\|\rd_t\wht{\mfg}_3\|_{H^{{2}+k}_{{\delta}}}\leq C(C_1)\ep^2 \lambda^2.$$
\end{cor}

\begin{rk}
We have $(\beta_3)_{asymp}= \partial_t (\beta_3)_{asymp}= 0$, as a consequence of the local well-posedness result {(Theorem~\ref{lwp})}.
\end{rk}

Notice that Corollary \ref{g3.cor} is an improvement of \eqref{BA4}.

\subsection{Improved estimates for $\rd_t \gamma_3$}

The key is the following lemma:
\begin{lm}\label{g1.main.lemma}
$$\sum_{k\leq {7}} \lambda^k\|(\rd_t-\beta_0^i\rd_i) \gamma_1 -\f 12 \rd_i\beta_1^i\|_{H^k\cap {C^{k}}}\leq C(C_0)\ep^2 \lambda^2.$$
\end{lm}
\begin{proof}
This is proven by an explicit computation. For the purpose of this proof, it is helpful to introduce a notation for the background $e_0$ - we will denote ${\bf e}_0:=\rd_t-\beta_0^i\rd_i$. By \eqref{G.g} and \eqref{g1.def}, we have
\begin{equation*}
\begin{split}
\gamma_1=&\f 18\sum_\bA \frac{\lambda^2 F_\bA^2}{|\nabla u_\bA|^2}\left(  |\nab u_\bA|^2 +\f{e^{2\gamma_0}}{N^2_0}({\bf e}_0 u_\bA)( {\bf e}_0 u_\bA)\right) \co{\frac{2u_\bA}{\lambda}}\\
&+2\sum_\bA\frac{\lambda^2 F_\bA}{|\nabla u_\bA|^2} \left(  (\nab\phi_0)\cdot(\nab u_\bA) +\f{e^{2\gamma_0}}{N^2_0}({\bf e}_0 \phi_0)( {\bf e}_0 u_\bA)\right)   \si{\frac{u_\bA}{\lambda}}\\
&+{\f 12}\sum_{\pm} \sum_{\bA}\sum_{\bB \neq \bA}\frac{(\mp 1)\cdot \lambda^2 F_\bA F_{\gra B}}{|\nabla (u_\bA \pm u_{\gra B})|^2} \left(  (\nab u_\bA)\cdot(\nab u_{\bB}) +\f{e^{2\gamma_0}}{N^2_0}({\bf e}_0 u_\bA)( {\bf e}_0 u_\bB)\right)  \co{\frac{ u_\bA \pm u_{\gra B}}{\lambda}}.
\end{split}
\end{equation*}
Differentiating, we obtain
\begin{equation}\label{e0.g1}
\begin{split}
{\bf e}_0\gamma_1=&-\f 14\sum_\bA \frac{\lambda F_\bA^2}{|\nabla u_\bA|^2}({\bf e}_0 u_\bA)\left(  |\nab u_\bA|^2 +\f{e^{2\gamma_0}}{N^2_0}({\bf e}_0 u_\bA)^2\right) \si{\frac{2u_\bA}{\lambda}}\\
&+2\sum_\bA\frac{\lambda F_\bA}{|\nabla u_\bA|^2} ({\bf e}_0 u_\bA)\left(  (\nab\phi_0)\cdot(\nab u_\bA) +\f{e^{2\gamma_0}}{N^2_0}({\bf e}_0 \phi_0)( {\bf e}_0 u_\bA)\right)   \co{\frac{u_\bA}{\lambda}}\\
&-{\f 12}\sum_{\pm} \sum_{\bA}\sum_{\bB \neq \bA}\frac{(\mp 1)\cdot \lambda F_\bA F_{\gra B}}{|\nabla (u_\bA \pm u_{\gra B})|^2} ({\bf e}_0 (u_\bA\pm u_\bB))\\
&\qquad\qquad\times\left(  (\nab u_\bA)\cdot(\nab u_{\bB}) +\f{e^{2\gamma_0}}{N^2_0}({\bf e}_0 u_\bA)( {\bf e}_0 u_\bB)\right)  \si{\frac{ u_\bA \pm u_{\gra B}}{\lambda}}+\dots,
\end{split}
\end{equation}
where here, and below in this proof, we have used $\dots$ to denote terms with $\sum_{k\leq 3} \lambda^k\|\cdot\|_{H^k\cap {C^{k}}}$ norms bounded above by $C(C_0)\ep^2 \lambda^2$. (These terms arise when the derivative does not act on the oscillating factor{s}.) On the other hand, we compute $\beta^i_1$ according to \eqref{G.b1}, \eqref{G.b2} and \eqref{g1.def} to get
\begin{equation*}
\begin{split}
\beta^i_1=&\f 12\delta^{ij}\sum_\bA \frac{\lambda^2 F_\bA^2}{|\nabla u_\bA|^2}  ({\bf e}_0 u_\bA)( \rd_j u_\bA) \co{\frac{2u_\bA}{\lambda}}\\
&+4\delta^{ij}\sum_\bA\frac{\lambda^2 F_\bA}{|\nabla u_\bA|^2} \left(  ({\bf e}_0\phi_0)(\rd_j u_\bA) +({\bf e}_0 u_\bA)( \rd_j \phi_0)\right)   \si{\frac{u_\bA}{\lambda}}\\
&+2\delta^{ij}\sum_{\pm} \sum_{\bA}\sum_{\bB \neq \bA}\frac{(\mp 1)\cdot \lambda^2 F_\bA F_{\gra B}}{|\nabla (u_\bA \pm u_{\gra B})|^2}  ({\bf e}_0 u_\bA)(\rd_j u_{\bB})  \co{\frac{ u_\bA \pm u_{\gra B}}{\lambda}}.
\end{split}
\end{equation*}
Therefore, taking the divergence (with respect to ${\delta_{ij}}$), {(and noting that the last term remains unchanged under $\bA \leftrightarrow \bB$),} we obtain
\begin{equation}\label{div.b}
\begin{split}
\rd_i\beta^i_1=&-\sum_\bA \frac{{\lambda} F_\bA^2}{|\nabla u_\bA|^2}  ({\bf e}_0 u_\bA)| \nab u_\bA|^2 \si{\frac{2u_\bA}{\lambda}}\\
&+4\sum_\bA\frac{{\lambda} F_\bA}{|\nabla u_\bA|^2} \left(  ({\bf e}_0\phi_0)|\nab u_\bA|^2 +({\bf e}_0 u_\bA)(( \nab \phi_0)\cdot(\nab u_\bA))\right)   \co{\frac{u_\bA}{\lambda}}\\
&{-} \sum_{\pm} \sum_{\bA}\sum_{\bB \neq \bA}\frac{(\mp 1)\cdot {\lambda} F_\bA F_{\gra B}}{|\nabla (u_\bA \pm u_{\gra B})|^2}  ({\bf e}_0 u_\bA)((\nab u_{\bB}) \cdot(\nab(u_\bA\pm u_\bB)) \si{\frac{ u_\bA \pm u_{\gra B}}{\lambda}}\\
& {- \sum_{\pm} \sum_{\bA}\sum_{\bB \neq \bA}\frac{(\mp 1)\cdot \lambda F_\bA F_{\gra B}}{|\nabla (u_\bA \pm u_{\gra B})|^2}  ({\bf e}_0 u_\bB)((\nab u_{\bA}) \cdot(\nab(u_\bA\pm u_\bB)) \si{\frac{ u_\bA \pm u_{\gra B}}{\lambda}}+\dots}
\end{split}
\end{equation}
Finally, recalling that $u_\bA$ is an eikonal function on the background, we have $\f{e^{2\gamma_0}}{N_0^2}({\bf e}_0 u_\bA)^2=|\nab u_\bA|^2$. Therefore, in the expression $({\bf e}_0 \gamma_1- \f 12\rd_i\beta_1^i)$, the main contributions in \eqref{e0.g1} and \eqref{div.b} cancel exactly. This yields the proposition.
\end{proof}

Equipped with Lemma \ref{g1.main.lemma}, the necessary improved estimates for $\rd_t\gamma_3$ follows easily from the gauge conditions:
\begin{proposition}\label{dtg3.prop}
$$\|\rd_t\gamma_3\|_{L^2(B(0,R_{supp}{+1}))}\leq C(C_1)\ep^2\lambda^2.$$
\end{proposition}
\begin{proof}
Since the mean curvature of each $\Sigma_t$ vanishes, we have
$$(\rd_t-\beta^i\rd_i)\gamma =\f 12 \rd_i \beta^i .$$
The same formula holds for the background, with $\gamma \mapsto \gamma_0$ and $\beta^i\mapsto \beta^i_0$, i.e.,
$$(\rd_t-\beta_0^i\rd_i)\gamma_0 =\f 12 \rd_i \beta_0^i .$$
{Subtracting}, we thus obtain
$$(\rd_t-\beta_0^i\rd_i)(\gamma-\gamma_0)=\f 12 \rd_i (\beta^i-\beta_0^i)+(\beta^i-\beta^i_0)\rd_i\gamma,$$
i.e.,
$$(\rd_t-\beta_0^i\rd_i)(\gamma_1+\gamma_2+\gamma_3)=\f 12 \rd_i (\beta_1^i+\beta_2^i+\beta_3^i)+(\beta^i-\beta^i_0)\rd_i\gamma.$$
As a consequence,
\begin{equation*}
\begin{split}
&\|\rd_t\gamma_3\|_{L^{{2}}(B(0,R_{supp}{+1}))}\\
\leq &C(C_0)\left(\|{\rd} \mfg_2\|_{L^{{2}}(B(0,R_{supp}{+1}))}+\|\nabla \mfg_3\|_{L^{{2}}(B(0,R_{supp}{+1}))}+\|(\rd_t-\beta_0^i\rd_i)\gamma_1-\f 12 \rd_i\beta_1^i\|_{L^{{2}}(B(0,R_{supp}{+1}))}\right.\\
&\left.\qquad\quad+{\|\beta^i-\beta^i_0\|_{L^\infty(B(0,R_{supp}+1)} \|\rd_i\gamma\|_{L^2(B(0,R_{supp}+1)}}\right)\leq C(C_1)\ep^2\lambda^2,
\end{split}
\end{equation*}
where we have used \eqref{B2}, \eqref{BA3}, \eqref{BA4}{, \eqref{g2.Linfty}, \eqref{g3.Linfty},} as well as Lemma \ref{g1.main.lemma}.
\end{proof}
Since on the compact set $B(0,R_{supp}{+1})$, $\rd_t\gamma_3=\rd_t\wht \gamma_3$, we have thus improved the bootstrap assumption \eqref{BA5}.

\section{Putting everything together{: Proof of Theorem~\ref{main.thm.2}}}\label{sec.concl}
Thanks to Propositions \ref{whtF.est}, Corollary \ref{E.cor}, Proposition \ref{g2.est}, Corollary \ref{g3.cor} and Proposition \ref{dtg3.prop} we have proved the estimates {in the bootstrap assumptions \eqref{BA1}--\eqref{BA5}}, with $C_1\ep$ replaced by $C(C_0)\ep+ C(C_1)\ep^2,$
where $C(C_0)$ is a constant, depending on the background and independent of $C_1$. We can choose $C_1$ {sufficiently large} such that
$C_1 \geq 4C(C_0),$
and $\ep$ {sufficiently small} such that
$C(C_1)\ep \leq \frac{1}{4}C_1,$
we then have {improved} the {estimates in in the bootstrap assumptions \eqref{BA1}--\eqref{BA5},} with $C_1$ replaced by $\frac{C_1}{2}$. This {proves the bootstrap theorem (Theorem~\ref{thm:BS}).}

{Now, by the the local existence theorem (Theorem~\ref{lwp}), we see that when $F_{\bA} \equiv 0$ (which is the case under consideration), the solution only breaks down when at least one of the following holds:
\begin{enumerate}
\item (Higher norms of matter fields blow up)
$$\liminf_{t\to T_*}\left(\|\f{e^{2\gamma}}{N}(e_0\phi)\|_{H^k}(t) + \|\nab\phi\|_{H^k}(t) + \max_{\bA} \|F_{\bA} e^{\f{\gamma}{2}}\|_{H^k}(t) \right) = +\infty$$
\item (Lower norms of matter fields leave smallness regime\footnote{Note that with our estimates, we also have control over the support of the solution. Therefore, according to Theorem~\ref{lwp}, there is indeed an $\ep_{low}$ that we can talk about, which depends on this upper bound of the support of the solution, as well as $k$ and $\de$.})
$$\liminf_{t\to T_*} \left(\|\f{e^{2\gamma}}{N}(e_0\phi)\|_{L^\infty}(t) + \|\nab\phi\|_{L^\infty}(t) + \max_{\bA} \|F_{\bA} e^{\f{\gamma}{2}}\|_{L^\infty}(t) \right) > \ep_{low}.$$
\end{enumerate}
However, by choosing $\ep$ sufficiently small, the estimates in the bootstrap argument precisely show that neither of these can occur. Hence, the solution exists on the whole time interval $[0,1]$.
}

{It remains to show that we have the desired convergence, which follows easily from the parametrix \eqref{phi.para}, \eqref{g.para.def}, and estimates established in the bootstrap argument.} This concludes the proof of Theorem \ref{main.thm.2}.

\appendix

\section{Weighted Sobolev spaces}\label{weightedsobolev}
For {the} sake of complet{e}ness, we recall here useful properties on weighted Sobolev spaces. For relevant definitions, see Definition \ref{def.spaces}. {Unless otherwise stated, we will only be interested in weighted Sobolev spaces on $\mathbb R^2$. Most of the results can be found in \cite[Appendix I]{livrecb} (although we use slightly different notations).} 

\subsection{Embedding theorems}
The following lemma is an immediate consequence of the definition.
\begin{lm}\label{der} Let $m \geq 1$, $p\in [1, \infty)$ and $\delta \in \mathbb{R}$. Then {for $j=1,2$}, 
$$\| \partial_j u\|_{W^{m-1}_{\delta+1,p}} {\ls_{m,\de,p}} \|u \|_{W^m_{\delta,p}}.$$
Similarly, for $m \geq 1$, $\delta \in \mathbb{R}$, $j=1,2$, 
$$\| \partial_j u\|_{C^{m-1}_{\delta+1}} {\ls_{m,\de,p}} \|u \|_{C^m_{\delta}}.$$
\end{lm}

{We have an easy embedding result, which is a straightforward application of the H\"older's inequality:
\begin{lm}\label{weight:emb}
If $1 \leq p_1 \leq p_2 \leq \infty$ and $\delta_2-\delta_1 > 2\left(\f{1}{p_1}-\f{1}{p_2}\right)$, 
then we have the continuous embedding
$$W^{0}_{\delta_2,p_2}\subset W^{0}_{\delta_1,p_1}.$$
\end{lm}
}

The following simple lemma will be useful as well.

\begin{lm}\label{produit2}
	Let $\alpha \in \mathbb{R}$ and $g \in L^\infty_{loc}$ be such that
	$$|g(x)| \lesssim (1+|x|^2)^\alpha.$$
	Then the multiplication by $g$ maps $H^0_{\delta}$ to $H^0_{\delta -2\alpha}$ {with operator norm bounded by $\sup_{x\in \m R^2} \f{|g(x)|}{(1+|x|^2)^{\alp}}$}.
\end{lm}

Next, we have a Sobolev embedding theorem with weights:
\begin{prp}\label{holder} 
Let $s,m \in \m N \cup \{0\}$, $1<p<\infty$. The following Sobolev embedding theorems hold:
\begin{itemize}
\item Suppose $s >\frac{2}{p}$ {and $\beta \leq \delta +\frac{2}{p}$}. Then, we have the continuous embedding
	$$W^{s+m}_{\delta,p}\subset {C^{m}_{\beta}}.$$
\item Suppose $s < \f 2p$. Then, we have the continuous embedding
  $$W^{s+m}_{\delta,p}\subset W^{m}_{\delta+s, \f{np}{n-sp}}.$$
\item Suppose $s = \f 2p$. Then, we have the continuous embedding for all $q<\infty$
	$$W^{s+m}_{\delta,p}\subset W^{m}_{\delta+s, q}.$$
\end{itemize}
\end{prp}

{It will be convenient to prove a refined (in terms of scaling) Sobolev embedding theorem for the $H^2 \subset L^\infty$ embedding:
\begin{prp}\label{scaled.Sobolev}
The following holds for all functions $u$ such that the RHS is finite:
$$\| u \|_{L^\infty}\lesssim \| u \|_{L^2}^{\f 12} \left(\sum_{|\alp|=2}\|\nab^{\alp} u \|_{L^2}\right)^{\f 12}.$$
\end{prp}
\begin{proof}
By Proposition~\ref{holder}, for every $v$,
$$\| v \|_{L^\infty}\ls \| v \|_{L^2} + \sum_{|\alp|=2}\|\nab^{\alp} v \|_{L^2}.$$
Apply this to $v(x) = u_{\mu}(x) := u(\f{x}{\mu})$ for $\mu>0$, we obtain
$$\| u \|_{L^\infty}\ls \mu \| u \|_{L^2} + \mu^{-1} \sum_{|\alp|=2}\|\nab^{\alp} u \|_{L^2}.$$
Choose\footnote{We can assume without loss of generality that $\|u\|_{L^2} \neq 0$ for otherwise the conclusion is trivial.} $\mu = \|u\|_{L^2}^{-\f 12} \left(\sum_{|\alp|=2}\|\nab^{\alp} u \|_{L^2}\right)^{\f 12}$ yields the conclusion.
\end{proof}
}

\subsection{{Product estimates}}

{We have three product estimates. The first can be found in \cite[Appendix I]{livrecb}.}
\begin{prp}\label{produit}
	Let $s,s_1,s_2 \in \m N \cup \{0\}$, $p \in [1,\infty]$, $\de, \de_1, \de_2 \in \m R$. Assume that $s\leq \min(s_1,s_2)$ and $s<s_1+s_2-{\f 2 p}$. Let $\delta < \delta_1 + \delta_2 + {\f 2p}$. Then $\forall (u,v) \in {W^{s_1}_{\delta_1, p}\times W^{s_2}_{\delta_2, p}}$,
	$$\|uv\|_{{W^s_{\delta,p}}} \ls_{{s,s_1,s_2,p,\de,\de_1,\de_2}} \|u\|_{{W^{s_1}_{\delta_1,p}}} \|v\| _{{W^{s_2}_{\delta_2, p}}}.$$
\end{prp}

{The second is for unweighted Sobolev spaces, which can be found in \cite[Appendix~A]{Tao}.}
{
\begin{prp}\label{product}
 Let $s \in \m N$. Then $\forall (u,v)\in (H^s\cap L^\infty) \times (H^s\cap L^\infty)$,
   $$\| uv \|_{H^s} \ls_{s} \|u\|_{H^s} \|v\|_{L^\infty} + \|u\|_{L^\infty} \|v\|_{H^s}.$$
\end{prp}
}
{The third product estimate is an immediate corollary of Proposition~\ref{product}, which is useful when one of the two functions is compactly supported.
\begin{prp}\label{product.local}
 Let $s \in \m N$. Then $\forall (u,v)\in (H^s\cap L^\infty) \times (H^s\cap L^\infty)$ such that $supp(u)\subset B(0,R_{supp})$,
   $$\| uv \|_{H^s} \ls_{s,R_{supp}} \|u\|_{H^s} \|v\|_{L^\infty(B(0,R_{supp}+1))} + \|u\|_{L^\infty} \|v\|_{H^s(B(0,R_{supp}+1))}.$$
	\end{prp}
\begin{proof}
Let $\eta$ be a cutoff function compactly supported in $B(0,R_{supp}+1)$ which is $\equiv 1$ in $B(0,R_{supp})$. Noting that $uv = uv\eta$, we apply Proposition~\ref{product} with $(u,v\eta)$ instead of $(u,v)$.
\end{proof}
}

\subsection{{Inversion of the Laplacian}}
We then discuss the invertibility of the Laplacian on weighted Sobolev spaces. The following theorem is due to McOwen:

\begin{thm}\label{laplacien}(Theorem 0 in \cite{laplacien})
	Let $m\in \mathbb{Z}$, $m\geq 0$, $1<p<\infty$ and $-\frac{2}{p}+m<\delta<m+1-\frac{2}{p}$. The Laplace operator $\Delta :W^{2+m}_{\delta,p} \rightarrow W^{m}_{\delta+2,p}$ is an injection with closed range 
	$$\left \{f \in W^{m}_{\delta+2,p}\; | \;\int fv =0 \quad \forall v \in \cup_{i=0}^m \mathcal{H}_i \right \},$$
	where $\mathcal{H}_i$ is the set of harmonic polynomials of degree $i$.
	Moreover, $u$ obeys the estimate
	$$\|u\|_{W^{2+m}_{\delta,p}} \leq C(\delta,m,p)\|\Delta u\|_{W^{m}_{\delta+2,p}},$$
	where $C(\delta,m,p) > 0$ is a constant depending on $\delta$, $m$ and $p$.
\end{thm}
An immediate corollary is the following:
\begin{cor} \label{coro} Let $-1<\delta<0$ and $f\in H^0_{\delta+2}$. Then there exists a solution $u$ of 
	$$\Delta u =f$$
	which can be written 
	$$u=\frac{1}{2\pi}\left(\int f\right)\chi({|x|})\log({|x|}) +v,$$
	where $\chi:[0,\infty)\to \mathbb R$ is a smooth cutoff function $=0$ for ${|x|}\leq 1$ and $=1$ for ${|x|}\geq 2$, and $v \in H^2_{\delta}$ is such that
	$\|v\|_{H^2_\delta} \leq C(\delta)\|f\|_ {H^0_{ \delta+2}}$.
\end{cor}

\end{document}